\documentclass{colt2019} % Include author names

\title[Nonparametric Finite Time LTI System Identification]{Nonparametric Finite Time LTI System Identification}
\usepackage{times}
\coltauthor{%
	\Name{Tuhin Sarkar} \Email{tsarkar@mit.edu}\\
	\addr 77 Massachusetts Ave, Cambridge, MA 02139
	\AND
	\Name{Alexander Rakhlin} \Email{rakhlin@mit.edu}\\
	\addr 77 Massachusetts Ave, Cambridge, MA 02139 
	\AND
	\Name{Munther A. Dahleh} \Email{dahleh@mit.edu}\\
	\addr 77 Massachusetts Ave, Cambridge, MA 02139
}
\usepackage{graphicx}
\usepackage{tikz}
\usetikzlibrary{fit,positioning,arrows,automata,calc}
\tikzset{
	main/.style={circle, minimum size = 5mm, thick, draw =black!80, node distance = 10mm},
	connect/.style={-latex, thick},
	box/.style={rectangle, draw=black!100}
}

\usepackage{epigraph}

\usepackage{csquotes}
\usepackage{hyperref}       % hyperlinks
\usepackage{url}       
\usepackage{relsize}     % simple URL typesetting
\usepackage{booktabs}       % professional-quality tables
\usepackage{amsfonts}       % blackboard math symbols
\usepackage{nicefrac}       % compact symbols for 1/2, etc.
\usepackage{microtype}      % microtypography
\usepackage{color}
\usepackage{scalerel}
\usepackage[toc,page]{appendix}
\usepackage[font=small,labelfont=bf]{caption}
\usepackage{subcaption}
\usepackage{blindtext}
\usepackage{amssymb}
\usepackage{mathtools}
\usepackage{dsfont}
\usepackage{amsmath}
\usepackage{amsfonts}
\usepackage{float}
\usepackage{graphicx}
\graphicspath{ {images/} }
\usepackage{letltxmacro}%
\usepackage{thmtools,thm-restate}
\usepackage{relsize}
\usepackage{algorithm}
\usepackage{algorithmic}

\DeclarePairedDelimiterX\Basics[1](){ #1}

\DeclarePairedDelimiter\abs{\lvert}{\rvert}

\DeclarePairedDelimiterX{\infdivx}[2]{(}{)}{#1\;\delimsize\|\;#2}

\makeatletter
\newcommand{\distas}[1]{\mathbin{\overset{#1}{\kern\z@\sim}}}%
\newsavebox{\mybox}\newsavebox{\mysim}
\newcommand{\distras}[1]{%
	\savebox{\mybox}{\hbox{\kern3pt$\scriptstyle#1$\kern3pt}}%
	\savebox{\mysim}{\hbox{$\sim$}}%
	\mathbin{\overset{#1}{\kern\z@\resizebox{\wd\mybox}{\ht\mysim}{$\sim$}}}%
}
\makeatother

\usepackage{enumitem}
\newlist{inparaenum}{enumerate}{2}% allow two levels of nesting in an enumerate-like environment
\setlist[inparaenum]{nosep}% compact spacing for all nesting levels
\setlist[inparaenum,1]{label=\bfseries\arabic*.}% labels for top level
\setlist[inparaenum,2]{label=\arabic{inparaenumi}\emph{\alph*})}% labels for second level

\usepackage{bm}
\usepackage{breqn}
\usepackage{physics}
\usepackage{enumitem}
\usepackage{bbold}
\usepackage[colorinlistoftodos]{todonotes}

\newtheorem{thm}{Theorem}[section]
\newtheorem*{thm*}{Theorem}
\newtheorem{cor}{Corollary}[section]
\newtheorem{lem}{Lemma}[section]
\newtheorem{prop}{Proposition}[section]
\newtheorem*{prop*}{Proposition}

\newtheorem{assumption}{Assumption}

\newcommand{\Hc}{\mathcal{H}}

\newcommand{\Mc}{\mathcal{M}}
\newcommand{\tT}{\tilde{T}}
\newcommand{\hsigma}{\hat{\sigma}}
\newcommand{\tM}{\widetilde{M}}

\newcommand{\tA}{\tilde{A}}
\newcommand{\tB}{\tilde{B}}

\newcommand{\bl}{\Big |}

\newcommand{\sigA}{\Sigma^S}
\newcommand{\sigB}{\Sigma^P}

\newcommand{\Zb}{\mathbb{Z}}

\newcommand{\tC}{\tilde{C}}
\newcommand{\Ic}{\mathcal{I}}
\newcommand{\Vc}{\mathcal{V}}
\newcommand{\bcF}{\bm{\mathcal{F}}}

\newcommand{\Nc}{{\mathcal{N}}}

\newcommand{\Bc}{\mathcal{B}}
\newcommand{\Ex}{\mathbb{E}}
\newcommand{\Pb}{\mathbb{P}}
\newcommand{\Rb}{\mathbb{R}}
\newcommand{\Oc}{\mathcal{O}}
\newcommand{\Rc}{\mathcal{R}}
\newcommand{\Sc}{\mathcal{S}}
\newcommand{\Tc}{\mathcal{T}}

\newcommand{\Ec}{\mathcal{E}}

\newcommand{\Lc}{\mathcal{L}}
\newcommand{\Gc}{\mathcal{G}}

\newcommand{\Ac}{\mathcal{A}}

\newcommand{\Cc}{\mathcal{C}}
\newcommand{\Dc}{\mathcal{D}}
\newcommand{\hHc}{\hat{\Hc}}

\newcommand{\hd}{\hat{d}} 
\newcommand{\ds}{d_{*}}
\newcommand{\dso}{d_{*}^{(1)}}
\newcommand{\dsf}{d_{*}^{(\kappa^2)}}
\newcommand{\hdo}{\hat{d}^{(1)}}
\newcommand{\hdf}{\hat{d}^{(\kappa^2)}}	
\newcommand{\dk}{d_{\kappa}}
\newcommand{\Hcinf}{\Hc_{0, \infty, \infty}}

\newcommand{\bhSigma}{\bar{\hat{\Sigma}}}

\newcommand{\subg}{\mathsf{subg}}
\newcommand{\wh}{\widehat}
\DeclarePairedDelimiter{\p}{(}{)}
\DeclarePairedDelimiter{\nrm}{\|}{\|}
\begin{document}
	
	\maketitle
	
	\begin{abstract}
		We address the problem of learning the parameters of a stable linear time invariant (LTI) system or linear dynamical system (LDS) with unknown latent space dimension, or order, from a single time--series of noisy input-output data. We focus on learning the best lower order approximation allowed by finite data. Motivated by subspace algorithms in systems theory, where the doubly infinite system Hankel matrix captures both order and good lower order approximations, we construct a Hankel-like matrix from noisy finite data using ordinary least squares. This circumvents the non-convexities that arise in system identification, and allows accurate estimation of the underlying LTI system. Our results rely on careful analysis of self-normalized martingale difference terms that helps bound identification error up to logarithmic factors of the lower bound. We provide a data-dependent scheme for order selection and find an accurate realization of system parameters, corresponding to that order, by an approach that is closely related to the Ho-Kalman subspace algorithm. We demonstrate that the proposed model order selection procedure is not overly conservative, i.e., for the given data length it is not possible to estimate higher order models or find higher order approximations with reasonable accuracy. 
	\end{abstract}
	
	\smallskip
	
	% MODEL
	\section{Introduction}
\label{introduction}
Finite-time system identification---the problem of estimating the system parameters given a finite single time series of its output---is an important problem in the context of control theory, time series analysis, robotics, and economics, among many others. In this work, we focus on parameter estimation and model approximation of linear time invariant (LTI) systems or linear dynamical system (LDS), which are described by
\begin{align}
X_{t+1} &= A X_t + B U_t + \eta_{t+1} \nonumber \\
Y_t &= C X_t + w_t. \label{dt_lti}
\end{align}
Here $C \in \Rb^{p \times n}, A \in \Rb^{n \times n}, B \in \Rb^{n \times m}$; $\{\eta_t, w_t\}_{t=1}^{\infty}$ are process and output noise, $U_t$ is an external control input, $X_t$ is the latent state variable and $Y_t$ is the observed output. The goal here is parameter estimation, \textit{i.e.}, learning $(C, A, B)$ from a single finite time series of $\{Y_t, U_t\}_{t=1}^T$ when the order, $n$, is unknown. Since typically $p,m < n$, it becomes challenging to find suitable parametrizations of LTI systems for provably efficient learning. When $\{X_j\}_{j=1}^{\infty}$ are observed (or, $C$ is known to be the identity matrix), identification of $(C, A, B)$ in Eq.~\eqref{dt_lti} is significantly easier, and ordinary least squares (OLS) is a statistically optimal estimator. It is, in general, unclear how (or if) OLS can be employed in the case when $X_t$'s are not observed. 
% For instance, one natural approach is to consider the input--output version of Eq.~\eqref{dt_lti}
% \begin{equation*}
%     Y_t = CA^{l-1}BU_{t-l} + \underbrace{\sum_{m \neq l}^{t-1} CA^{m-1}B U_{t-m} + \sum_{m=0}^{t-1} CA^{m} \eta_{t-m} + w_t}_{=\epsilon_t}
% \end{equation*}
% and, for any fixed $l$, regress $\{Y_t\}_{t=l+1}^{\infty}$ against $\{U_{t-l}\}_{t=l+1}^{\infty}$. This approach has two major limitations: first, $\{\epsilon_t\}_{t=1}^{\infty}$ is a correlated sequence even when $\eta_t$'s are independent. Second, we can only obtain $CA^{l-1}B$, which is insufficient to learn $(C, A, B)$. These limitations make it hard to extend previous analyses of such OLS based system identification methods [SR: which analyses? with $C=I$?]. 

To motivate the study of a lower-order approximation of a high-order system, consider the following example:
%The following example shows how a high order LTI system can be approximated reasonably well by a lower order one. 
 \begin{example}
 \label{truncation_example}
Consider $M_1=(A_1,B_1,C_1)$ with 
 \begin{align}
A_1 &= \begin{bmatrix}
0 & 1 & 0 & 0 & \hdots & 0 \\
0 & 0 & 1 & 0 & \hdots & 0 \\
%\vdots & \vdots & \vdots & \ddots & \vdots & \vdots \\
\vdots & \vdots & \vdots & \vdots & \ddots & \vdots \\
0 & 0 & 0 & 0 & \hdots & 1 \\
-a& 0 & 0 & 0 & \hdots & 0 
\end{bmatrix}_{n \times n} B_1 = \begin{bmatrix}
0 \\
0  \\
\vdots \\
0 \\
1 
\end{bmatrix}_{n \times 1} C_1 = B_1^{\top}\label{M1}
\end{align}
where $na \ll 1$ and $n > 20$. Here the order of $M_1$ is $n$. However, it can be approximated well by $M_2$ which is of a much lower order and given by
\begin{align}
A_2 &= \begin{bmatrix}
0 & 0 \\
1 & 0 
\end{bmatrix} \hspace{2mm} B_2 = \begin{bmatrix}
0 \\
1 
\end{bmatrix} ~~~ C_2 = B_2^{\top}. \label{M2}
\end{align}
For the same input $U_t$, if $Y^{(1)}_t, Y^{(2)}_t$ be the output generated by $M_1$ and $M_2$ respectively then a simple computation shows that 
\[
\sup_{U}\sum_{t=1}^{\infty}\frac{(Y^{(1)}_t- Y^{(2)}_t)^2}{U_t^2} \leq 4n^2a^2 \ll 1
\]
This suggests that the actual value of $n$ is not important; rather there exists an effective order, $r$ (which is $2$ in this case). This lower order model captures ``most'' of the LTI system.
\end{example}  
Since the true model order is not known in many cases, we emphasize a nonparametric approach to identification: one which adaptively selects the best model order for the given data and approximates the underlying LTI system better as $T$ (length of data) grows. The key to this approach will be designing an estimator $\hat{M}$ from which we obtain a realization $(\hat{C}, \hat{A}, \hat{B})$ of the selected order.

\subsection{Related Work}
Linear time invariant systems are an extensively studied class of models in control and systems theory. These models are used in feedback control systems (for example in  planetary soft landing systems for rockets~\citep{2013lossless}) and as linear approximations to many non--linear systems that nevertheless work well in practice. In the absence of process and output noise, subspace-based system identification methods are known to learn $(C, A, B)$ (up to similarity transformation)\citep{ljung1987system,van2012subspace}. These typically involve constructing a Hankel matrix from the input--output pairs and then obtaining system parameters by a singular value decomposition. Such methods are inspired by the celebrated Ho-Kalman realization algorithm~\citep{ho1966effective}. The correctness of these methods is predicated on the knowledge of $n$ or presence of infinite data. Other approaches include rank minimization-based methods for system identification \citep{fazel2013hankel,grussler2018low}, further relaxing the rank constraint to a suitable convex formulation. However, there is a lack of statistical guarantees for these algorithms, and it is unclear how much data is required to obtain accurate estimates of system parameters from finite noisy data. Empirical methods such as the EM algorithm \citep{roweis1999unifying} are also used in practice; however, these suffer from non-convexity in problem formulation and can get trapped in local minima. Learning simpler approximations to complex models in the presence of finite noisy data was studied in~\citet{venkatesh2001system} where identification error is decomposed into error due to approximation and error due to noise; however the analysis assumes the knowledge of a ``good'' parametrization and does not provide statistical guarantees for learning the system parameters of such an approximation.

More recently, there has been a resurgence in the study of statistical identification of LTI systems from a single time series in the machine learning community. In cases when $C = I$, \textit{i.e.}, $X_t$ is observed directly, sharp finite time error bounds for identification of $A, B$ from a single time series are provided in~\citet{faradonbeh2017finite,simchowitz2018learning,sarkar2018}. The approach to finding $A, B$ is based on a standard ordinary least squares (OLS) given by 
$$ (\hat{A}, \hat{B}) = \arg \min_{A, B} \sum_{t=1}^T ||X_{t+1} - [A, B][X_t^{\top}, U_t^{\top}]^{\top}||_2^2.$$ 
Another closely related area is that of online prediction in time series~\citet{hazan2018spectral,agarwal2018time}. Finite time regret guarantees for prediction in linear time series are provided in~\citet{hazan2018spectral}. The approach there circumvents the need for system identification and instead uses a filtering technique that convolves the time series with eigenvectors of a specific Hankel matrix.    

{ Closest to our work is that of~\citet{oymak2018non}. Their algorithm, which takes inspiration from the Kalman--Ho algorithm, assumes the knowledge of model order $n$. This limits the applicability of the algorithm in two ways: first, it is unclear how the techniques can be extended to the case when $n$ is unknown---as is usually the case---and, second, in many cases $n$ is very large and a much lower order LTI system can be a very good approximation of the original system. In such cases, constructing the order $n$ estimate might be unnecessarily conservative (See Example~\ref{truncation_example}). Consequently, the error bounds do not reflect accurate dependence on the system parameters. 
	
When $n$ is unknown, it is unclear when a singular value decomposition should be performed to obtain the parameter estimates via Ho-Kalman algorithm. This leads to the question of model order selection from data. For subspace based methods, such problems have been addressed in~\cite{shibata1976selection} and~\cite{bauer2000order}. These papers address the question of estimating order in the context of subspace methods. Specifically, order estimation is achieved by analyzing the information contained in the estimated singular values and/or estimated innovation variance. Furthermore, they provide guarantees for asymptotic consistency of the methods described. Another line of literature studied in~\cite{ljung2015regularization} for example, approaches the identification of systems with unknown order by first learning the largest possible model that fits the data and then performing model reduction to obtain the final system. Although one can show that asymptotically this method outputs the true model, we show that such a two step procedure may underperform in a finite time setting. A possible explanation for this could be that learning the largest possible model with finite data over-fits on the exogenous noise and therefore gives poor model estimates.

Other related work on identifying finite impulse response approximations include~\citet{goldenshluger1998nonparametric,tu2017non}; but they do not discuss parameter estimation or reduced order modeling. Several authors~\citet{campi2002finite,shah2012linear,hardt2016gradient} and references therein have studied the problem of system identification in different contexts. However, they fail to capture the correct dependence of system parameters on error rates. More importantly, they suffer from the same limitation as~\citet{oymak2018non} that they require the knowledge of $n$.
	\section{Mathematical Preliminaries} 
\label{sec:math_preliminaries}
%\subsection{Notations and Definitions}
% \begin{definition}
% A finite dimensional linear time invariant (LTI) system is parametrized by $C \in \Rb^{p \times n}, A \in \Rb^{n \times n}, B \in \Rb^{n \times m}$ with the following dynamics
% \begin{align}
% X_t &= A X_{t-1} + B U_t + \eta_t \nonumber\\
% Y_t &= C X_t + w_t \label{dt_lti}
% \end{align}
% Here $U_t$ is an external input, $Y_t$ is the observed output, $X_t$ is the hidden state and $\{\eta_t, w_t\}$ are process and output noise respectively. 
% \end{definition}
Throughout the paper, we will refer to an LTI system with dynamics as Eq.~\eqref{dt_lti} by $M=(C, A, B)$. For a matrix $A$, let $\sigma_i(A)$ be the $i^{\text{th}}$ singular value of $A$ with $\sigma_i(A) \geq \sigma_{i+1}(A)$. Further, $\sigma_{\max}(A) = \sigma_1(A) = \sigma(A)$. Similarly, we define $\rho_i(A) = |\lambda_i(A)|$, where $\lambda_i(A)$ is an eigenvalue of $A$ with $\rho_i(A) \geq \rho_{i+1}(A)$. Again, $\rho_{\max}(A) = \rho_1(A) = \rho(A)$. 
\begin{definition}
A matrix $A$ is \textit{Schur stable} if $\rho_{\max}(A) < 1$. 
\end{definition}

We will only be interested in the class of LTI systems that are Schur stable. Fix $\gamma > 0$ (and possibly much greater than $1$). The model class $\Mc_r$ of LTI systems parametrized by $r \in \Zb_{+}$ is defined as \begin{equation}
    \label{model_class_eq}
    \Mc_r = \{(C, A, B)\hspace{1mm}|\hspace{1mm} C \in \Rb^{p \times r},A \in \Rb^{r \times r},B \in \Rb^{r \times m}, \rho(A) < 1, \sigma(A) \leq \gamma\}.
\end{equation}
\begin{definition}
\label{hankel_matrix}
The $(k, p, q)$--dimensional Hankel matrix for $M=(C, A, B)$ as 
	\begin{align*}
	\Hc_{k, p, q}(M) = \begin{bmatrix}
	CA^{k}B & CA^{k+1}B & \hdots & CA^{q+k-1}B \\
	CA^{k+1}B & CA^{k+2}B & \hdots & CA^{q+k}B \\
	\vdots & \vdots & \ddots & \vdots \\
	CA^{p+k-1}B & \hdots & \hdots & CA^{p+q+k-2}B
	\end{bmatrix}
	\end{align*}
and its associated Toeplitz matrix as
\begin{align*}
		\Tc_{k, d}(M) = \begin{bmatrix}
		0 & 0 & \hdots & 0 & 0 \\
		CA^{k}B & 0 & \hdots & 0 & 0\\
		\vdots & \ddots & \ddots & \vdots & 0 \\
		CA^{d+k-3}B &  \hdots & CA^{k}B & 0 & 0\\
		CA^{d+k-2}B &CA^{d+k-3}B&  \hdots &  CA^{k}B & 0
		\end{bmatrix}.
\end{align*}
\end{definition}
We will slightly abuse notation by referring to $\Hc_{k, p, q}(M)  = \Hc_{k, p, q}$. Similarly for the Toeplitz matrices $\Tc_{k, d}(M) = \Tc_{k, d}$. The matrix $\Hc_{0, \infty, \infty}(M)$ is known as the \textit{system Hankel matrix} corresponding to $M$, and its rank is known as the \textit{model order} (or simply \textit{order}) of $M$. The system Hankel matrix has two well-known properties that make it useful for system identification. First, the rank of $\Hc_{0, \infty, \infty}$ has an upper bound $n$. Second, it maps the ``past'' inputs to ``future'' outputs. These properties are discussed in detail in appendix as Section~\ref{control_hankel}. For infinite matrices $\Hc_{0, \infty, \infty}$, $||\Hc_{0, \infty, \infty}||_2\triangleq ||\Hc_{0, \infty, \infty}||_{\text{op}}$, \textit{i.e.}, the operator norm.
\begin{definition}
\label{transfer_function}
The \textit{transfer function} of $M=(C, A, B)$ is given by $G(z) = C(zI - A)^{-1}B$ where $z \in \mathbb{C}$.
\end{definition}
 The transfer function plays a critical role in control theory as it relates the input to the output. Succinctly, the transfer function of an LTI system is the Z--transform of the output in response to a unit impulse input. Since for any invertible $S$ the LTI systems 
% \begin{equation}
%     \label{similarity_transform}
    $M_1 = (CS^{-1}, S A S^{-1}, S B), M_2 = (C, A ,  B)$
%\end{equation}
have identical transfer functions,  identification may not be unique, but equivalent up to a transformation $S$, \textit{i.e.}, $(C, A, B) \equiv (CS, S^{-1}AS, S^{-1}B)$.
Next, we define a system norm that will be important from the perspective of model identification and approximation. 
\begin{definition}
\label{system_norm}
The $\Hc_{\infty}$--\textit{system norm} of a Schur stable LTI system $M$ is given by $$||M||_{\infty} = \sup_{\omega \in \Rb} \sigma_{\max}(G(e^{j\omega})).$$
Here, $G(\cdot)$ is the transfer function of $M$. The $r$--truncation of the transfer function is defined as
\begin{equation}
G_r \coloneqq [CB, CAB, \hdots, CA^{r-1}B]. \label{eq:g_trunc}
\end{equation}
%Furthermore, the Hankel system norm is defined as $||M||_H = ||\Hc_{0, \infty, \infty}(M)||_2$.
\end{definition}
For a stable LTI system $M$ we have
\begin{prop}[Lemma 2.2~\cite{glover1987model}]
    \label{sys_norm}
Let $M$ be a LTI system then 
\[
||M||_H = \sigma_1 \leq ||M||_{\infty} \leq 2(\sigma_1 + \ldots + \sigma_n)
\]
where $\sigma_i$ are the singular values of $\Hcinf(M)$. 
\end{prop}
{For any matrix $Z$, define $Z_{m:n, p:q}$ as the submatrix including row $m$ to $n$ and column $p$ to $q$. Further, $Z_{m:n, :}$ is the submatrix including row $m$ to $n$ and all columns and a similar notion exists for $Z_{:, p:q}$. Finally, we define balanced truncated models which will play an important role in our algorithm. 
\begin{definition}[\cite{kung1981optimal}]
\label{balanced_truncations_definition}
Let $\Hc_{0, \infty, \infty}(M) = U \Sigma V^{\top}$ where $\Sigma \in \Rb^{n \times n}$ ($n$ is the model order). Then for any $r \leq n$, the $r$--order balanced truncated model parameters are given by 
$$C_r = [U\Sigma^{1/2}]_{1:p, 1:r}, A_r = \Sigma_{1:r, 1:r}^{-1/2}U_{:, 1:r}^{\top} [U\Sigma^{1/2}]_{p+1:, 1:r}, B_r = [\Sigma^{1/2}V^{\top}]_{1:r, 1:m}.$$
For $r > n$, the $r$--order balanced truncated model parameters are the $n$--order truncated model parameters.
\end{definition}

\begin{definition}
	\label{subgaussian_rv}
	We say a random vector $v\in\Rb^{d}$ is subgaussian with variance proxy $\tau^{2}$ if 
	$$
	\sup_{||\theta||_2=1}\sup_{p\geq 1}\left\{p^{-1/2}\left(\Ex[\left|\langle v, \theta \rangle\right|^{p}]\right)^{1/p}\right\} = \tau$$ 
	and $\Ex[v] = \textbf{0}$. We denote this by $v\sim\subg(\tau^{2})$.
\end{definition} 
}
A fundamental result in model reduction from systems theory is the following
\begin{thm}[Theorem 21.26~\cite{zhou1996robust}]
    \label{balanced_truncate_model}
    Let $M = (C, A, B)$ be the true model of order $n$ and $M_r = (C_r, A_r, B_r)$ be its balance truncated model of order $r < n$. Assume that $\sigma_r \neq \sigma_{r+1}$. Then 
    \[
    ||M - M_r||_{\infty} \leq 2(\sigma_{r+1} + \sigma_{r+2} + \ldots + \sigma_n)
    \]
    where $\sigma_i$ are the Hankel singular values of $M$.
\end{thm}
Critical to obtaining refined error rates, will be a result from the theory of self--normalized martingales, an application of the pseudo-maximization technique in \citep[Theorem 14.7]{pena2008self}:
\begin{thm}
	\label{selfnorm_main}
	Let $\{\bcF_t\}_{t=0}^{\infty}$ be a filtration. Let $\{\eta_{t} \in \Rb^m, X_t \in \Rb^d\}_{t=1}^{\infty}$ be stochastic processes such that $\eta_t, X_t$ are $\bcF_t$ measurable and $\eta_t$ is $\bcF_{t-1}$-conditionally $\subg(L^2)$ for some $L > 0$.
    For any $t \geq 0$, define 
	$
	V_t = \sum_{s=1}^t X_s X_s^{\prime}, S_t = \sum_{s=1}^t \eta_{s+1} X_s
	$.
	Then for any $\delta > 0, V \succ 0$ and all $t \geq 0$ we have with probability at least $1-\delta$ 
	\[
 S_t^{\top}(V + V_t)^{-1}S_t  \leq  4L^2 \p*{\log{\frac{1}{\delta}} + \log{\frac{\text{det}(V+V_t)}{\text{det}(V)}} + m}.
	\]
\end{thm} 
The proof of this result can be found as Theorem~\ref{thm:selfnorm_main}.

We denote by $c$ universal constants which can change from line to line. For numbers $a, b$, we define $a \wedge b \triangleq \min{(a, b)} $ and $a \vee b \triangleq \max{(a, b)}$. \\

Finally, for two matrices $M_1 \in \Rb^{l_1 \times l_1}, M_2 \in \Rb^{l_2 \times l_2}$ with $l_1 < l_2$, $M_1 - M_2 \triangleq \tilde{M}_1 - M_2$ where $\tilde{M}_1 = \begin{bmatrix}M_1 & 0_{l_1 \times l_2 -l_1} \\
0_{l_2-l_1 \times l_1} & 0_{l_2 - l_1 \times l_2 -l_1}
\end{bmatrix}$.

\begin{prop}[System Reduction]
	\label{reduction}
Let $||S-P|| \leq \epsilon$ and the singular values of $S$ be arranged as follows:
\begin{equation*}
\sigma_1(S) > \ldots > \sigma_{r-1}(S) > \sigma_r(S) \geq \sigma_{r+1}(S) \geq \ldots \geq \sigma_s(S) >  \sigma_{s+1}(S) > \ldots \sigma_n(S) > \sigma_{n+1}(S) = 0 
\end{equation*}
Furthermore, let $\epsilon$ be such that 
\begin{equation}
\epsilon \leq \inf_{\{1 \leq i \leq r-1\} \cup \{s+1 \leq i \leq n\}} \Big(\frac{\sigma_i(P) - \sigma_{i+1}(P)}{2}\Big).
\end{equation}
Define $K_0 = [1,2, \hdots, r-1] \cup [s+1,s+2,\hdots,n]$, then
\begin{align*}
||U^S_{K_0} (\sigA_{K_0})^{1/2} - U^P_{K_0} (\sigB_{K_0})^{1/2}||_2 &\leq 2\sqrt{\sum_{i=1}^{r-1}\frac{\sigma_i \epsilon^2}{(\sigma_i - \sigma_{i+1})^2 \wedge (\sigma_{i-1} - \sigma_{i})^2}} \\
&+ 2\sqrt{ \frac{\sigma_s \epsilon^2}{((\sigma_{r-1} - \sigma_{s}) \wedge (\sigma_{r} - \sigma_{s+1}))^2}} + \sup_{1\leq i \leq s}|\sqrt{\sigma_i} - \sqrt{\hat{\sigma}_i}| 
\end{align*}
and $\sigma_i = \sigma_i(S), \hat{\sigma}_i = \sigma_i(P)$.
\end{prop}
The proof is provided in Proposition~\ref{reduction2} in the appendix. This is an extension of Wedin's result that allows us to scale the recovery error of the $r^{th}$ singular vector by only condition number of that singular vector. This is useful to represent the error of identifying a $r$-order approximation as a function of the $r^{th}$-singular value only.

We briefly summarize our contributions below.
}
\section{Contributions}
\label{contributions}
In this paper we provide a purely data-driven approach to system identification from a single time--series of finite noisy data. Drawing from tools in systems theory and the theory of self--normalized martingales, we offer a nearly optimal OLS-based algorithm to learn the system parameters. We summarize our contributions below:
\begin{itemize}[leftmargin=*]
    \item The central theme of our approach is to estimate the infinite system Hankel matrix (to be defined below) with increasing accuracy as the length $T$ of data grows. By utilizing a specific reformulation of the input--output relation in Eq.~\eqref{dt_lti} we reduce the problem of Hankel matrix identification to that of regression between appropriately transformed versions of output and input. The OLS solution is a matrix $\hat{\Hc}$ of size $\hd$. More precisely, we show that with probability at least $1-\delta$,
    \begin{align*}
	\bl \bl \hat{\Hc} - \Hc_{0, \hd, \hd} \bl \bl_2 &\lesssim \sqrt{\frac{\beta^2\hd }{T}} \sqrt{p\hd + \log{\frac{T}{\delta}}}
	\end{align*}
	for $T$ above a certain threshold, where $\Hc_{0, \hd, \hd}$ is the $p\hd \times m\hd$ principal submatrix of the system Hankel. Here $\beta$ is the $\Hc_{\infty}$--system norm.
    
    \item We show that by growing $\hd$ with $T$ in a specific fashion, $\hHc$ becomes the minimax optimal estimator of the system Hankel matrix. The choice of $\hd$ for a fixed $T$ is purely data-dependent and does not depend on spectral radius of $A$ or $n$. 
    
    \item It is well known in systems theory that SVD of the doubly infinite system Hankel matrix gives us $A, B, C$. However, the presence of finite noisy data prevents learning these parameters accurately. We show that it is always possible to learn the parameters of a lower-order approximation of the underlying system. This is achieved by selecting the top $k$ singular vectors of $\hHc$. The estimation guarantee corresponds to \textit{model selection} in Statistics. More precisely, for every $k \leq \hd$ if $(A_k, B_k, C_k)$ are the parameters of a $k$-order balanced approximation of the original LTI system and $(\hat{A}_k, \hat{B}_k, \hat{C}_k)$ are the estimates of our algorithm then for $T$ above a certain threshold we have
    \begin{align*}
        ||C_k - \hat{C}_k||_2 + ||A_k - \hat{A}_k||_2 + ||B_k - \hat{B}_k||_2 &\lesssim \sqrt{\frac{\beta^2 \hd}{\hsigma_k^2 T}} \sqrt{p\hd + \log{\frac{T}{\delta}}}
    \end{align*}
    with probability at least $1 - \delta$ where $\hsigma_i$ is the $i^{\text{th}}$ largest singular value of $\hHc$.
    
%    \item The lower order $k$ is obtained by using a novel singular value thresholding scheme that depends purely on data, and works under mild assumptions. We show that the proposed thresholding scheme is not overly conservative, \textit{i.e.}, there exist higher order LTI systems that cannot be identified accurately with the given data length.
\end{itemize}

	\section{Problem Formulation and Discussion}
\label{problem}
%%%%%%%%%%%%
\subsection{Data Generation}
Assume there exists an unknown $M = (C, A, B) \in \Mc_n$ for some unknown $n$. Let the transfer function of $M$ be $G(z)$. Suppose we observe the noisy output time series $\{Y_t \in \Rb^{p \times 1}\}_{t=1}^T$ in response to user chosen input series, $\{U_t \in \Rb^{m \times 1}\}_{t=1}^T$. We refer to this data generated by $M$ as $Z_T = \{(U_t, Y_t)\}_{t=1}^T$. We enforce the following assumptions on $M$. 

 \begin{assumption}
\label{noise_dist}
The noise process $\{\eta_t, w_t\}_{t=1}^{\infty}$ in the dynamics of $M$ given by Eq.~\eqref{dt_lti} are i.i.d. and $\eta_t, w_t$ are isotropic with subGaussian parameter $1$. Furthermore, $X_0 = 0$ almost surely. We will only select inputs, $\{U_t\}_{t=1}^T$, that are isotropic subGaussian with subGaussian parameter $1$. 
\end{assumption}
The input--output map of Eq.~\eqref{dt_lti} can be represented in multiple alternate ways. One commonly used reformulation of the input--output map in systems and control theory is the following
\begin{equation*}
    \begin{bmatrix}
    Y_1 \\
    Y_{2} \\
    \vdots \\
    Y_{T}
    \end{bmatrix} = \Tc_{0, T}\begin{bmatrix}
    U_1 \\
    U_2 \\
    \vdots \\
    U_T 
    \end{bmatrix} + \Tc\Oc_{0, T} \begin{bmatrix}
    \eta_1 \\
    \eta_2 \\
    \vdots \\
    \eta_T 
    \end{bmatrix} + \begin{bmatrix}
    w_1 \\
    w_2 \\
    \vdots \\
    w_T 
    \end{bmatrix}
\end{equation*}
where $\Tc\Oc_{k, d}$ is defined as the Toeplitz matrix corresponding to process noise $\eta_t$ (similar to Definition~\ref{hankel_matrix}):
\begin{align*}
	\Tc\Oc_{k, d} = \begin{bmatrix}
		0 & 0 & \hdots & 0 & 0 \\
		CA^{k} & 0 & \hdots & 0 & 0\\
		\vdots & \ddots & \ddots & \vdots & 0 \\
		CA^{d+k-3} &  \hdots & CA^{k} & 0 & 0\\
		CA^{d+k-2} & CA^{d+k-3} &  \hdots &  CA^{k} & 0
	\end{bmatrix}.
\end{align*}
$||\Tc_{0, T}||_2, ||\Tc\Oc_{0, T}||_2$ denote observed amplifications of the control input and process noise respectively. Note that stability of $A$ ensures $||\Tc_{0, \infty}||_2, ||\Tc\Oc_{0, \infty}||_2 < \infty$. Suppose both $\eta_t, w_t = 0$ in Eq.~\eqref{dt_lti}. Then it is a well-known fact that 
\begin{equation}
||M||_{\infty} = \sup_{U_t}\sqrt{\frac{\sum_{t=0}^{\infty} Y_t^{\top}Y_t}{\sum_{t=0}^{\infty} U_t^{\top}U_t}} \implies ||M||_{\infty} = ||\Tc_{0, \infty}||_2 \geq ||\Hc_{0, \infty, \infty}||_2. \label{Linf_norm}
\end{equation}
% \begin{assumption}
%     \label{minimal}
% $M=\Mc(C, A, B)$ is minimal.
% \end{assumption}
% \begin{remark}
%     \label{non-minimal}
% The assumption of minimality is not too limiting because if the rank of $\Hc_{0, \infty, \infty}(C, A, B) = r < n$, we could always find the reduced minimal system. [Show in appendix]
% \end{remark}

\begin{assumption}
    \label{non-zero-gap}
There exist universal constants $\beta, R \geq 1$ such that $||\Tc_{0, \infty}||_2 \leq \beta,~~ \frac{||\Tc\Oc_{0, \infty}||_2}{||\Tc_{0, \infty}||_2} \leq R.$
% Let $\sigma_i = \sigma_i(\Hc_{\infty}(M))$ with $\sigma_{i} \geq \sigma_{i+1}$. Then there exists $\Delta_{+} > 0$ such that 
% \[
% \Delta_{+} = \inf_{i} \Big(1 - \frac{\sigma_{i+1}}{\sigma_i}\Big)
% \]
% \textit{i.e.}, all Hankel singular values are distinct.
\end{assumption}
\begin{remark}[$\Hc_{\infty}$-norm estimation]
\label{non-zero-gap_rmk}
Assumption~\ref{non-zero-gap} implies that an upper bound to the $\Hc_{\infty}$--norm of the system. It is possible to estimate $||M||_{\infty}$ from data (See~\cite{tu2018approximation} and references therein). It is reasonable to expect that error rates for identification of the parameters $(C, A, B)$ depend on the \textit{noise-to-signal ratio} $\frac{||\Tc\Oc_{0, \infty}||_2}{||\Tc_{0, \infty}||_2}$, \textit{i.e.}, identification is much harder when the ratio is large.  
\end{remark}

\begin{remark}[$R$ estimation]
\label{noise-to-signal}
The noise to signal ratio hyperparameter can also be estimated from data, by allowing the system to run with $U_t = 0$ and taking the average $\ell_2$ norm of the output $Y_t$, \textit{i.e.}, $(1/T)\sum_{t=1}^T \|Y_t\|^2_2$. For the purpose of the results of the paper we simply assume an upper bound on $R$. If $U_t$ was $\subg(L)$ instead of $\subg(1)$, the noise-to-signal ratio is modified to $R/L$ instead.
\end{remark}

\begin{table*}
	%\begin{strip}
	\begin{center}
		\begin{tabular}{|l|}
			\hline
			$m$: Input dimension, $p$: Output dimension\\
			\hline
			$\gamma$: Known upper bound on $||A||_2$ \\
			\hline
			$\delta$: Error probability \\
			\hline
			$c, \Cc$: Known absolute constants\\
			\hline
			$R$: Known noise to signal ratio, or, $\frac{||\Tc\Oc_{0, \infty}||_2}{||\Tc_{0, \infty}||_2}$ \\
			\hline
			$\beta$: Known upper bound on $\Hc_{\infty}$-norm of LTI system\\
			\hline
			$\Dc(T) =  \{d| T \geq cm^2 d \log^2{(d)}\log^2{(m^2/\delta)} + c d \log^3{(2d)}\}$\\
			\hline
			$\sigma_A = \sum_{l=1}^d ||CA^lB||_2$, $\sigma_B = \sum_{l=1}^d ||CA^l||_2$\\
			\hline 
			$\sigma_C = \sqrt{\sigma\p*{\sum_{k=1}^d \Tc_{d+k, T}^{\top}\Tc_{d+k, T}}}$, $\sigma_D = \sqrt{\sigma\p*{\sum_{k=1}^d \Tc\Oc_{d+k, T}^{\top}\Tc\Oc_{d+k, T}}}$ \\
			\hline
			$\alpha(l) = \sqrt{l}\p*{\sqrt{\frac{lp + \log{(T/\delta)} + m}{T}}}$  \\
			\hline 
		\end{tabular}
		\caption{Summary of constants} \label{notation}
	\end{center}
	%\end{strip}
\end{table*}
%%%%%%%

	\section{Algorithmic Details}
\label{algorithm}
We will now represent the input--output relationship in terms of the Hankel and Toeplitz matrices defined before. Fix a $d$, then for any $l$ we have
 \begin{align}
 \begin{bmatrix}
 Y_{l} \\
 Y_{l+1} \\
\vdots \\
Y_{l+d-1}
\end{bmatrix} &= \Hc_{0, d, d} \begin{bmatrix}
 U_{l-1} \\
 U_{l-2} \\
 \vdots \\
 U_{l-d}
 \end{bmatrix} +  \Tc_{0, d}\begin{bmatrix}
 U_{l} \\
 U_{l+1} \\
 \vdots \\
 U_{l+d-1}
 \end{bmatrix} +  \Oc_{0, d, d} \begin{bmatrix}
 \eta_{l-1} \\
 \eta_{l-2} \\
 \vdots \\
 \eta_{l-d+1}
 \end{bmatrix} +  \Tc\Oc_{0, d}\begin{bmatrix}
 \eta_{l} \\
 \eta_{l+1} \\
 \vdots \\
 \eta_{l+d-1}
 \end{bmatrix} \nonumber \\
 &+ \Hc_{d, d, l-d-1} \begin{bmatrix}
 U_{l-d-1} \\
 U_{l-d-1} \\
 \vdots \\
 U_{1}
 \end{bmatrix} + \Oc_{d, d, l-d-1} \begin{bmatrix}
 \eta_{l-d-1} \\
 \eta_{l-d-1} \\
 \vdots \\
 \eta_{1}
 \end{bmatrix} + \begin{bmatrix}
 w_{l} \\
 w_{l+1} \\
 \vdots \\
 w_{l+d-1}
 \end{bmatrix} \label{input-output-eq} 
\end{align}
or, succinctly,
\begin{align}
\tilde{Y}^{+}_{l, d} &= \Hc_{0, d, d}\tilde{U}^{-}_{l-1, d} + \Tc_{0, d}\tilde{U}^{+}_{l, d} + \Hc_{d, d, l-d-1}\tilde{U}^{-}_{l-d-1, l-d-1} \nonumber\\
&+ \Oc_{0, d, d}\tilde{\eta}^{-}_{l-1, d} + \Tc\Oc_{0, d}\tilde{\eta}^{+}_{l, d} + \Oc_{d, d, l-d-1}\tilde{\eta}^{-}_{l-d-1, l-d-1} + \tilde{w}^{+}_{l, d}\label{compact-dynamics} 
\end{align}
Here 
\begin{align*}
	\Oc_{k, p, q} &= \begin{bmatrix}
	CA^{k} & CA^{k+1} & \hdots & CA^{q+k-1} \\
	CA^{k+1} & CA^{k+2} & \hdots & CA^{d+k} \\
	\vdots & \vdots & \ddots & \vdots \\
	CA^{p+k-1} & \hdots & \hdots & CA^{p+q+k-2}
	\end{bmatrix}, 
\tilde{Y}^{-}_{l, d} =  \begin{bmatrix}
Y_{l} \\
Y_{l-1} \\
\vdots \\
Y_{l-d+1}
\end{bmatrix}, \tilde{Y}^{+}_{l, d} =  \begin{bmatrix}
Y_{l} \\
Y_{l+1} \\
\vdots \\
Y_{l+d-1}
\end{bmatrix}.
\end{align*}
Furthermore, $\tilde{U}^{-}_{l, d}, \tilde{\eta}^{-}_{l, d}$ are defined similar to $\tilde{Y}^{-}_{l, d}$ and $\tilde{U}^{+}_{l, d}, \tilde{\eta}^{+}_{l, d}, \tilde{w}^{+}_{l, d}$ are similar to $\tilde{Y}^{+}_{l, d}$. The $+$ and $-$ signs indicate moving forward and backward in time respectively. This representation will be at the center of our analysis. 

There are three key steps in our algorithm which we describe in the following sections: 
\begin{itemize}
	\item [(a)] Hankel submatrix estimation: Estimating $\Hc_{0, l, l}$ for every $1 \leq l \leq T$. We refer to the estimators as $\{\hHc_{0, l, l}\}_{l=1}^T$. 
	\item [(b)] Model Selection: From the estimators $\{\hHc_{0, l, l}\}_{l=1}^T$ select $\hHc_{0, \hd, \hd}$ in a data dependent way such that it ``best'' estimates $\Hc_{0, \infty, \infty}$.
	\item [(c)] Parameter Recovery: For every $k \leq \hd$, we do a singular value decomposition of $\hHc_{0,\hd, \hd}$ to obtain parameter estimates for a ``good'' $k$-order approximation of the true model.
\end{itemize}

\subsection{Hankel Submatrix Estimation}
\label{hankel_est}
% In Algorithms~\ref{alg:learn_ls} and \ref{alg:svd} we describe the system identification procedure in detail 
The goal of our systems identification is to estimate either $\Hc_{0, n, n}$ or $\Hc_{0, \infty, \infty}$. Since we only have finite data and no apriori knowledge of $n$ it is not possible to directly estimate the unknown matrices. The first step then is to estimate all possible Hankel submatrices that are ``allowed'' by data, \textit{i.e.}, $\Hc_{0,d,d}$ for $d \leq T$. For a fixed $d$, Algorithm~\ref{alg:learn_ls} estimates the $d \times d$ principal submatrix $\Hc_{0, d, d}$.
 
\begin{algorithm}[h]
	\caption{LearnSystem($T, d, m , p$)}
	\label{alg:learn_ls}
	\textbf{Input} $T=\text{Horizon for learning}$ \\
	%$\tau = \text{Threshold for singular value}$ \\
	$d= \text{Hankel Size}$ \\
	$m = \text{Input dimension}$\\
	$p=\text{Output dimension}$ \\
% 	$k=\text{Desired model order to learn}$ \\
	\textbf{Output} System Parameters: $\hHc_{0, d, d}$
% 	, (\hat{C}_k, \hat{A}_k, \hat{B}_k)\}$
	\begin{algorithmic}[1]
		\STATE Generate $2T$ i.i.d. inputs $\{U_j \sim \Nc(0, I_{m \times m}) \}_{j=1}^{2T}$. 
		\STATE Collect $2T$ input--output pairs $\{U_j, Y_j\}_{j=1}^{2T}$.
		\STATE $\hat{\Hc}_{0, d, d} = \arg \min_{\Hc} \sum_{l=0}^{T-1} ||\tilde{Y}^{+}_{l+d+1, d} - \Hc \tilde{U}^{-}_{l+d, d}||_2^2$ \\
%		\STATE $\Hc_0 = \text{Hankel}(\hat{\Hc}_{0, d, d})$
% 		\STATE $(\hat{C}_k, \hat{A}_k, \hat{B}_k) = \text{Hankel2Sys}(\hat{\Hc}_{0, d, d}, k, m, p)$.\\
		\RETURN $\hHc_{0, d, d}$
	\end{algorithmic}
\end{algorithm}

It can be shown that 
\begin{equation}
\label{eq:hankel_estimator}
\hHc_{0, d, d} = \Big(\sum_{l=0}^{T-1}\tilde{Y}^{+}_{l+d+1, d}(\tilde{U}^{-}_{l+d, d})^{\top}\Big)\Big(\sum_{l=0}^{T-1}\tilde{U}^{-}_{l+d, d}(\tilde{U}^{-}_{l+d, d})^{\top}\Big)^{+}
\end{equation}
and by running the algorithm $T$ times, we obtain $\{\hHc_{0, d, ,d}\}_{d=1}^{T}$. A key step in showing that $\hHc_{0,d,d}$ is a good estimator for $\Hc_{0,d,d}$ is to prove the finite time isometry of $V_T = \sum_{l=0}^{T-1}\tilde{U}^{-}_{l+d, d}(\tilde{U}^{-}_{l+d, d})^{\top}$, \textit{i.e.}, the sample covariance matrix.
\begin{lem}
	\label{energy_conc_main}
	Define 
	\begin{align*}
	T_{0}(\delta, d) = cm^2 d \log^2{(d)}\log^2{(m^2/\delta)} + c d \log^3{(2d)}
	\end{align*}
	where $c$ is some universal constant.  Define the sample covariance matrix $V_T \coloneqq \sum_{l=0}^{T-1}\tilde{U}^{-}_{l+d, d}(\tilde{U}^{-}_{l+d, d})^{\top}$. We have with probability $1 - \delta$ and for $T > T_0(\delta, d)$
	\begin{align}
	\label{yt_bnd}
	\frac{1}{2} T I \preceq V_T \preceq \frac{3}{2} T I
	\end{align}
\end{lem}
Lemma~\ref{energy_conc_main} allows us to write Eq.~\eqref{eq:hankel_estimator} as $\hHc_{0, d, d} = \Big(\sum_{l=0}^{T-1}\tilde{Y}^{+}_{l+d+1, d}(\tilde{U}^{-}_{l+d, d})^{\top}\Big)\Big(\sum_{l=0}^{T-1}\tilde{U}^{-}_{l+d, d}(\tilde{U}^{-}_{l+d, d})^{\top}\Big)^{-1}$ with high probability and  upper bound estimation error for $d \times d$ principal submatrix.
\begin{thm}
	\label{hankel_convergence}
	Fix $d$ and let $\hat{\Hc}_{0, d, d}$ be the output of Algorithm~\ref{alg:learn_ls}. Then for any $0 < \delta < 1$ and $T \geq T_0(\delta, d)$, we have with probability at least $1-\delta$
	
	\begin{align*}
	\bl \bl \hat{\Hc}_{0, d, d} - \Hc_{0, d, d} \bl \bl_2 &\leq 4 \sigma \sqrt{\frac{1}{T}} \sqrt{pd + \log{\frac{1}{\delta}} + m}.
	\end{align*}
	Here $T_0(\delta, d) = cm^2 d \log^2{(d)}\log^2{(m^2/\delta)} + c d \log^3{(2d)}$, $c$ is a universal constant and $\sigma = \max{(\sigma_A,\sigma_B, \sigma_C, \sigma_D)}$ from Table~\ref{notation}. 
\end{thm}
\begin{proof}
We outline the proof here. Recall Eq.~\eqref{input-output-eq}, \eqref{compact-dynamics}. Then for a fixed $d$ 
$$\hat{\Hc}_{0, d, d} = \Big(\sum_{l=0}^{T-1} \tilde{Y}^{+}_{l+d+1, d} (\tilde{U}^{-}_{l+d, d})^{\top}\Big) V_T^{+}.$$ 
Then the identification error is  
\begin{align}
\bl \bl \hat{\Hc}_{0, d, d} - \Hc_{0, d, d} \bl \bl_2 &= \bl \bl V_T^{+} \Big(\sum_{l=0}^{T-1} \tilde{U}^{-}_{l+d, d}  \tilde{U}^{+ \top}_{l+d+1, d} \Tc_{0, d}^{ \top} + \tilde{U}^{-}_{l+d, d} \tilde{U}^{- \top}_{l, l}\Hc_{d, d, l}^{\top} +  \tilde{U}^{-}_{l+d, d} \tilde{w}^{+ \top}_{l+d+1, d}\nonumber\\
&+ \tilde{U}^{-}_{l+d, d}\tilde{\eta}^{- \top}_{l+d, d} \Oc_{0, d, d}^{\top} + \tilde{U}^{-}_{l+d, d} \tilde{\eta}^{+ \top}_{l+d+1, d}\Tc\Oc^{\top}_{0, d} + \tilde{U}^{-}_{l+d, d} \tilde{\eta}^{- \top}_{l, l} \Oc_{d, d, l}^{\top} \Big)\bl \bl_2 \nonumber \\
&= ||V_T^{+}E||_2 \label{diff_eq} 
\end{align}
with 
\begin{align*}
E &= \sum_{l=0}^{T-1} \tilde{U}^{-}_{l+d, d}  \tilde{U}^{+ \top}_{l+d+1, d} \Tc_{0, d}^{ \top} + \tilde{U}^{-}_{l+d, d} \tilde{U}^{- \top}_{l, l}\Hc_{d, d, l}^{\top} +  \tilde{U}^{-}_{l+d, d} \tilde{w}^{+ \top}_{l+d+1, d}\nonumber\\
&+ \tilde{U}^{-}_{l+d, d}\tilde{\eta}^{- \top}_{l+d, d} \Oc_{0, d, d}^{\top} + \tilde{U}^{-}_{l+d, d} \tilde{\eta}^{+ \top}_{l+d+1, d}\Tc\Oc^{\top}_{0, d} + \tilde{U}^{-}_{l+d, d} \tilde{\eta}^{- \top}_{l, l} \Oc_{d, d, l}^{\top}.
\end{align*}
By Lemma~\ref{energy_conc_main} we have, whenever $T \geq T_0(\delta, d)$, with probability at least $1-\delta$ 
\begin{equation}
\frac{TI}{2} \preceq V_T \preceq \frac{3TI}{2}. \label{sample_cov_bnd}
\end{equation}
This ensures that, with high probability, that $V_T^{-1}$ exists and decays as $O(T^{-1})$. The next step involves showing that $||E||_2$ grows at most as $\sqrt{T}$ with high probability. This is reminiscent of Theorem~\ref{selfnorm_main} and the theory of self--normalized martingales. However, unlike that cases the conditional sub-Gaussianity requirements do not hold here. For example, let $\bcF_l = \sigma(\eta_1, \ldots, \eta_l)$ then $\Ex[v^{\top}\tilde{\eta}^{-}_{l+1, l+1} | \bcF_l] \neq 0$ for all $v$ since $\{\tilde{\eta}^{-}_{l+1, l+1}\}_{l=0}^{T-1}$ is not an independent sequence. As a result it is not immediately obvious on how to apply Theorem~\ref{selfnorm_main} to our case.  Under the event when Eq.~\eqref{sample_cov_bnd} holds (which happens with high probability), a careful analysis of the normalized cross terms, \textit{i.e.}, $V_T^{-1/2}E$ shows that $||V_T^{-1/2}E||_2 = O(1)$ with high probability. This is summarized in Propositions~\ref{error_prob1}-\ref{error_prob3}. The idea is to decompose $E$ into a linear combination of independent subgaussians and reduce it to a form where we can apply Theorem~\ref{selfnorm_main}. This comes at the cost of additional scaling in the form of system dependent constants -- such as the $\Hc_{\infty}$--norm. Then we can conclude with high probability that $||\hat{\Hc} - \Hc_{0, d, d}||_2 \leq ||V_T^{-1/2}||_2 ||V_T^{-1/2}E||_2 \leq T^{-1/2} O(1)$. The full proof has been deferred to Section~\ref{hankel_conv_proof} in Appendix~\ref{appendix_error}.
\end{proof}
\begin{remark}
\label{dcT_equiv}
Recall $\Dc(T)$ from Table~\ref{notation}. Since 
\[
d \in \Dc(T) \implies T \geq T_0(\delta, d)
\]
we can restate Theorem~\ref{hankel_convergence} as follows: for a fixed $T$, we have with probability at least $1-\delta$ that 
$$\bl \bl \hat{\Hc}_{0, d, d} - \Hc_{0, d, d} \bl \bl_2 \leq 4 \sigma \sqrt{\frac{1}{T}} \sqrt{pd + \log{\frac{1}{\delta}} + m}$$ 
when $d \in \Dc(T)$.
%\begin{equation}
%    \label{d_finite_time}
%    d \in \Dc(T).
%\end{equation}
\end{remark}
We next present bounds on $\sigma$ in Theorem~\ref{hankel_convergence}. From the perspective of model selection in later sections, we require that $\sigma$ be known. In the next proposition we present two bounds on $\sigma$, the first one depends on unknown parameters and recovers the precise dependence on $d$. The second bound is an apriori known upper bound and incurs an additional factor of $\sqrt{d}$.
\begin{prop}
	\label{prop:sigma_bnd}
	
	\textbf{$\sigma$ upper bound independent of $d$}:
	\[
	\sigma \leq \frac{c_n}{(1 - \rho(A))^2}
	\]
	where $c_n$ depends only on $n$.\\
	
	\vspace{3mm}
	
	\textbf{$\sigma$ upper bound dependent on $d$}:
	\[
	\sigma \leq \beta R \sqrt{d}.
	\]
	where $R$ is the noise-to-signal ratio as in Table~\ref{notation}
	 
\end{prop}
\begin{proof}
	
By Gelfand's formula, since $||A^d||_2 \leq c(n) \rho_{\max}(A)^d$ where $\rho_{\max}(A) < 1$ and $c(n)$ is a constant that only depends on $n$, it implies that
\begin{align*}
\sigma_A  &= \sum_{l=0}^d ||CA^lB||_2 \leq \sum_{l=0}^{\infty} ||CA^lB||_2 \leq \sum_{l=0}^{\infty} c(n) \rho(A)^l = \frac{c(n)}{1 - \rho(A)},
\end{align*}
and
\begin{align*}
||\Tc_{d+k, T}||_2 \leq \sum_{l=0}^{T-1} ||CA^{d+k+l}B||_2 \leq \frac{c(n) \rho(A)^{d+k}}{1-\rho(A)}. 
\end{align*}
Then 
\[
\sigma_C = \sqrt{\sigma\p*{\sum_{k=1}^d \Tc_{d+k, T}^{\top}\Tc_{d+k, T}}} \leq \frac{c(n) \rho(A)^{d}}{(1-\rho(A))^2} \leq \frac{c(n)}{(1-\rho(A))^2}. 
\]
Similarly, there exists a finite upper bound on $\sigma_B, \sigma_D$ by replacing $CA^lB$ and  $\Tc_{d+k, T}$ with $CA^l$ and $\Tc \Oc_{d+k, T}$ respectively. For the $d$ independent upper bound, we have 
\[
\sigma_A  = \sum_{l=0}^d ||CA^lB||_2 \leq \sqrt{d} \sqrt{\sum_{l=0}^d ||CA^lB||^2_2} \leq \sqrt{d} ||M||_H \leq \sqrt{d} \beta.
\]
Since $\sigma\p*{\Tc_{d+k, T}^{\top}\Tc_{d+k, T}} \leq \beta$, then 
\[
\sigma_C = \sqrt{\sigma\p*{\sum_{k=1}^d \Tc_{d+k, T}^{\top}\Tc_{d+k, T}}} \leq \beta \sqrt{d}.
\]
For the $\sigma_B, \sigma_D$ we get an extra $R$ because $\Tc \Oc_{0, \infty} \leq \beta R$.
\end{proof}
The key feature of the data dependent upper bound is that it only depends on $\beta$ and $R$ which are known apriori. 

Recall that $G_d = [CB, CAB, \hdots, CA^{d-1}B]$, \textit{i.e.}, the $d$-order FIR truncation of $G(z)$. Since the $p$ rows of the $\Hc_{0,d,d}$ matrix corresponds to $G_d$ we can obtain estimators for any $d$-order FIR.
\begin{cor}
	\label{cor:fir}
	Let $\wh{G}_d = \hHc_{0, d, d}[1:p, :]$ denote the first $p$-rows of $\hHc_{0, d, d}$. Then for any $0 < \delta < 1$ and $T \geq T_0(\delta, d)$, we have with probability at least $1-\delta$, 
	\[
	||\wh{G}_d - G_d||_2 \leq 4 \sigma \sqrt{\frac{1}{T}} \sqrt{pd + \log{\frac{1}{\delta}} + m}.
	\]
\end{cor} 
\begin{proof}
	Proof follows because ${G}_d = \Hc_{0, d, d}[1:p, :]$ and Theorem~\ref{hankel_convergence}.
\end{proof}

Next, we show that the error in Theorem~\ref{hankel_convergence} is minimax optimal (up to logarithmic factors) and cannot be improved by any estimation method.
\begin{prop}
	\label{lower_bnd_hankel}
	Let $\sqrt{T} \geq c$ where $c$ is an absolute constant. Then for any estimator $\hHc$ of $\Hc_{0, \infty, \infty}$ we have 
	\[
	\sup_{\hHc}\Ex[||\hHc - \Hc_{0, \infty, \infty}||_2] \geq  c_n \cdot {\sqrt{\frac{\log{T}}{T}}}
	\]
	where $c_{n} > 0$ is a constant that is independent of $T$ but can depend on system level parameters.
\end{prop}
\begin{proof}
	Assume the contrary that 
	\[
	\sup_{\hHc}\Ex[||\hHc - \Hc_{0, \infty, \infty}||_2] = o\p*{\sqrt{\frac{\log{T}}{T}}}. 
	\]
	Then recall that $[\Hc_{0, \infty, \infty}]_{1:p, :} = [CB, CAB, \hdots,]$ and $G(z) = z^{-1}CB + z^{-2} CAB + \hdots$. Similarly we have $\hat{G}(z)$. Define 
	\[
	||G - \hat{G}||_2 = \sqrt{\sum_{k=0}^{\infty}||CA^kB - \hat{C}\hat{A}^k\hat{B}||_2^2}.
	\]
	If $\sup_{\hHc}\Ex[||\hHc - \Hc_{0, \infty, \infty}||_2] = o\Big(\sqrt{\frac{\log{T}}{T}}\Big)$, then since $||\hHc - \Hc_{0, \infty, \infty}||_2 \geq ||G - \hat{G}||_2$ we can conclude that 
	\[
	\Ex[||G - \hat{G}||_2] = o\p*{\sqrt{\frac{\log{T}}{T}}}
	\]
	which contradicts Theorem 5 in~\citep{goldenshluger1998nonparametric}. Thus, $\sup_{\hHc}\Ex[||\hHc - \Hc_{0, \infty, \infty}||_2] \geq c_n \cdot {\sqrt{\frac{\log{T}}{T}}}$.
\end{proof}
\subsection{Model Selection}
\label{sec:model_selection}
At a high level, we want to choose $\hHc_{0, \hd, \hd}$ from $\{\hHc_{0,d,d}\}_{d=1}^T$ such that $\hHc_{0, \hd, \hd}$ is a good estimator of $\Hc_{0, \infty, \infty}$. Our idea of model selection is motivated by~\citep{goldenshluger1998nonparametric}. For any $\hHc_{0,d,d}$, the error from $\Hc_{0, \infty, \infty}$ can be broken as:
\[
||\hat{\Hc}_{0, d, d} - \Hc_{0, \infty, \infty}||_2 \leq \underbrace{||\hat{\Hc}_{0, d, d} - \Hc_{0, d, d}||_2}_{=\text{Estimation Error}}  + \underbrace{||\Hc_{0, d, d} - \Hc_{0, \infty, \infty}||_2}_{=\text{Truncation Error}}.
\]
We would like to select a $d=\hd$ such that it balances the truncation and estimation error in the following way:
\[
c_2 \cdot \text{Data dependent upper bound } \geq c_1 \cdot \text{Estimation Error} \geq \text{Truncation Error}
\]
where $c_i$ are absolute constants. Such a balancing ensures that 
\begin{equation}
||\hat{\Hc}_{0, \hd, \hd} - \Hc_{0, \infty, \infty}||_2 \leq c_2 \cdot (1/c_1 + 1) \cdot \text{Data dependent upper bound }. \label{eq:balancing}
\end{equation}
Note that such a balancing is possible because the estimation error increases as $d$ grows and truncation error decreases with $d$. Furthermore, a data dependent upper bound for estimation error can be obtained from Theorem~\ref{hankel_convergence}. Unfortunately $(C, A, B)$ are unknown and it is not immediately clear on how to obtain such a bound for truncation error. 

To achieve this, we first define a truncation error proxy, \textit{i.e.}, how much do we truncate if a specific $\hHc_{0,d,d}$ is used. For a given $d$, we look at $||\hHc_{0,d,d} - \hHc_{0, l, l}||_2$ for $l \in \Dc(T) \geq d$. This measures the additional error incurred if we choose $\hHc_{0, d, d}$ as an estimator for $\Hc_{0, \infty, \infty}$ instead of $\hHc_{0, l, l}$ for $l > d$. Then we pick $\hd$ as follows:
\begin{equation}
\hd \coloneqq  \inf\Bigg\{d \Bigg| ||\hHc_{0,d,d} - \hHc_{0, l, l}||_2 \leq 16 \beta R \cdot \alpha(l) \quad{} \forall l \in \Dc(T) \geq d\Bigg\}. \label{eq:hd_eq}    
\end{equation}
Recall that $\alpha(l) = \sqrt{\frac{l \log{(l/\delta)} + pl^2 + ml}{T}}$, where $\sqrt{\frac{ \log{(l/\delta)} + pl + m}{T}}$ denotes how much estimation error is incurred in learning $l \times l$ Hankel submatrix, the extra $\beta \sqrt{l}$ is incurred because we need a data dependent, albeit coarse, upper bound on the estimation error.

A key step will be to show that for any $l \geq d$, whenever 
\[
||\hHc_{0,d,d} - \hHc_{0, l, l}||_2 \leq c \beta R \cdot \alpha(l)
\]
ensures that  
\[
||\hHc_{0,d,d} - \Hc_{0,\infty,\infty}||_2 \leq c \beta R \cdot  \alpha(l) \quad{} \text{and} \quad{} ||\hHc_{0,l,l} - \Hc_{0,\infty,\infty}||_2 \leq c \beta R \cdot \alpha(l) 
\]
and there is no gain in choosing a larger Hankel submatrix estimate. By picking the smallest $d$ for which such a property holds for all larger Hankel submatrices, we ensure that a regularized model is estimated that ``agrees'' with the data.

%%For a fixed $T$, estimation error increases as $d$ increases but truncation error decreases and such a $d$ can be found. By balancing the estimation and truncation errors in this way we can control the total error $||\hat{\Hc}_{0, d, d} - \Hc_{0, \infty, \infty}||_2$ and obtain $O(T^{-1/2})$ error rate which is the optimal parametric error rate. Although truncation error is unknown we can provide a data--dependent upper bound for the estimation error (Theorem~\ref{hankel_convergence}). Consequently, $||\hat{\Hc}_{0, d, d} - \Hc_{0, \infty, \infty}||_2$ is bounded from above by twice this bound for the ideal choice of $d$.
%
%\paragraph{Choice of $d$:}
%In Algorithm~\ref{alg:learn_ls} we use the following rule to pick $d$, motivated by~\citep{goldenshluger1998nonparametric}. A key departure from the case there and other related work is that we allow for process noise, \textit{i.e.}, $\eta_t$ is not identically zero. Our results reflect this by the additional $R$ factor. 
\begin{algorithm}[H]
	\caption{Choice of $d$}
	\label{alg:d_choice}
	\textbf{Output} $\hHc_{0, \hd, \hd}, \quad{} \hd$
% 	, (\hat{C}_k, \hat{A}_k, \hat{B}_k)\}$
	\begin{algorithmic}[1]
		\STATE $\Dc(T) = \Big\{d \Big| d \leq \frac{T}{c m^2 \log^3{(Tm/\delta)}} \Big\}, \alpha(h) = \sqrt{h}\Big(\sqrt{\frac{m + hp + \log{\p*{\frac{T}{\delta}}}}{T}} \Big)$. 
		\STATE $d_0(T, \delta) =\inf \Big\{ l \Big| ||\hat{\Hc}_{0, l, l} - \hat{\Hc}_{0, h, h}||_2 \leq 16 \beta R (\alpha(h)  + 2\alpha(l) )   \hspace{2mm}\forall h \in \Dc(T), h  \geq l \Big\}$. 
		\STATE $\hat{d} = \max{\p*{d_0(T, \delta), \log{\p*{\frac{T}{\delta}}}}}$ \\
%		\STATE $\Hc_0 = \text{Hankel}(\hat{\Hc}_{0, d, d})$
% 		\STATE $(\hat{C}_k, \hat{A}_k, \hat{B}_k) = \text{Hankel2Sys}(\hat{\Hc}_{0, d, d}, k, m, p)$.\\
		\RETURN $\hHc_{0, \hd, \hd}, \quad{} \hd$
	\end{algorithmic}
\end{algorithm}
%
%From Remark~\ref{dcT_equiv} it is clear that $\Dc(T)$ is the set of $d$ values for which Theorem~\ref{hankel_convergence} holds with probability at least $1-\delta$. We focus on only these values of $d$ for which we can compute $\Hc_{0, d, d}$ with high accuracy. In Line $2$, while computing $d_0(T, \delta)$, note that $\hat{\Hc}_{0, l, l}$ are estimates of $\Hc_{0, l, l}$ for $l \in \Dc(T)$ and can be obtained from a single stream of data by repeatedly using Algorithm~\ref{alg:learn_ls} for every $l \in \Dc(T)$. By picking the least $l$ such that for every $h \geq l$ we have
%\[
%||\hat{\Hc}_{0, l, l} - \hat{\Hc}_{0, h, h}||_2 \leq \Cc \beta R (\alpha(h) + 2\alpha(l)) 
%\]
%we ensure that 
%\[
%||\hat{\Hc}_{0, l, l} -{\Hc}_{0, h, h}||_2 \leq ||\hat{\Hc}_{0, l, l} - \hat{\Hc}_{0, h, h}||_2 + ||{\Hc}_{0, h, h} - \hat{\Hc}_{0, h, h}||_2 = O (\beta R \alpha(h)).
%\]
%This implies that $\hHc_{0, h, h}$ does no better in estimating $\Hc_{0, h, h}$ than $\hHc_{0, l, l}$ (in an order sense). Algorithm~\ref{alg:d_choice} chooses the smallest $\hHc_{0, l, l}$ that works well in estimating the larger Hankel submatrices. 

We now state the main estimation result for $\Hc_{0, \infty, \infty}$ for $d = \hd$ as chosen in Algorithm~\ref{alg:d_choice}. Define 
\begin{align}
    T_{*}(\delta) &= \inf\Big\{T \Big |d_{*}(T, \delta) \in \Dc(T), \hspace{2mm} d_{*}(T, \delta) \leq 2 d_{*}\p*{\frac{T}{256}, \delta} \Big\}  \label{ts_delta}
\end{align}
where 
\begin{equation}
    d_{*}(T, \delta) = \inf\Bigg\{d \Bigg| 16 \beta R  \alpha(d) \geq ||\Hc_{0, d, d} - \Hc_{0, \infty, \infty}||_2 \Bigg\} . \label{ds_delta}
\end{equation}
A close look at Eq.~\eqref{ds_delta} reveals that picking $d = \ds(T, \delta)$ ensures the balancing of Eq.~\eqref{eq:balancing}. However, $d_*(T, \delta)$ depends on unknown quantities and is unknown. In such a case, $\hd$ in Eq.~\eqref{eq:hd_eq} becomes a proxy for $d_*(T, \delta)$. From an algorithmic stand point, we no longer need any unknown information; the unknown parameter only appear in $T_{*}(\delta)$, which is only required to make the theoretical guarantee of Theorem~\ref{hankel_est_thm} below.
\begin{thm}
    \label{hankel_est_thm}
Whenever we have $T \geq T_{*}(\delta)$ we have with probability at least $1-\delta$ that 
\[
||\hHc_{0, \hd, \hd} - \Hc_{0, \infty, \infty}||_2 \leq 12c\beta R\p*{\sqrt{\frac{m \hd + p\hat{d}^2+\hd \log{\frac{T}{\delta}}}{T}}}. 
\]
\end{thm}
The proof of Theorem~\ref{hankel_est_thm} can be found as Proposition~\ref{hd_size} in Appendix~\ref{sec:hankel_est}. We see that the error between $\hHc_{0, \hd, \hd}$ and $\Hc_{0, \infty, \infty}$ can be upper bounded by a purely data dependent quantity. The next proposition shows that $\hd$ does not grow more that logarithmically in $T$.
\begin{prop}
	\label{prop:hd_growth}
	Let $T \geq T_{*}(\delta)$, $d_{*}(T, \delta)$ be as in Eq.~\eqref{ds_delta}. Then with probability at least $1 -  \delta$ we have 
	\[
	\hd \leq d_{*}(T, \delta) \vee \log{\Big(\frac{T}{\delta}\Big)}.
	\]
	Furthermore, 
	\[
d_{*}(T, \delta) \leq \frac{c \log{(c T + \log{\frac{1}{\delta}})} - \log{R} +  \log{\beta}}{\log{\frac{1}{\rho(A)}}}.	
	\]
\end{prop}
The effect of unknown quantities, such as the spectral radius, are subsumed in the finite time condition $T \geq T_*(\delta)$ and appear in an upper bound for $\hd$; however this information is not needed from an algorithmic perspective as the selection of $\hd$ is agnostic to the knowledge of $\rho(A)$. The proof of proposition can be found as Propositions~\ref{d_ds_rel} and~\ref{t1_exist}.

\subsection{Parameter Recovery}
\label{params_rec}
Next we discuss finding the system parameters. To obtain system parameters we use a balanced truncation algorithm on $\hHc_{0, \hd, 
\hd}$ where $\hd$ is the output of Algorithm~\ref{alg:d_choice}. The details are summarized in Algorithm~\ref{alg:svd} where $\Hc = \hHc_{0, \hd, \hd}$. 
\begin{algorithm}[H]
	\caption{Hankel2Sys($T, \hd, k, m, p$)}
	\label{alg:svd}
	\textbf{Input} $T = \text{Horizon for Learning}$ \\
	$\hd = \text{Hankel Size}$ \\
	$m = \text{Input dimension}$\\
	$p=\text{Output dimension}$ \\
	%$\tau = \text{Threshold for singular value}$ \\
	%$k= \text{Desired model order to learn}$ \\
	\textbf{Output} System Parameters: $(\hat{C}_{\hd}, \hat{A}_{\hd}, \hat{B}_{\hd})$
	\begin{algorithmic}[1]
	    \STATE $\Hc = \Hc_{0, \hd, \hd}$
	    \STATE Pad $\Hc$ with zeros to make of dimension $4p\hd \times 4m\hd$
		\STATE $U, \Sigma, V \leftarrow$ SVD of $\Hc$ \\
		\STATE $U_{\hd}, V_{\hd} \leftarrow$ top $\hd$ singular vectors\\
		\STATE $\hat{C}_{\hd} \leftarrow$ first $p$ rows of $U_{\hd} \Sigma_{\hd}^{1/2}$ \\
		\STATE $\hat{B}_{\hd} \leftarrow$ first $m$ columns of $ \Sigma_{\hd}^{1/2} V^{\top}_{\hd}$\\
		\STATE $Z_0 = [U_{\hd} \Sigma_{\hd}^{1/2}]_{1:4p\hd-p, :}, Z_1 = [U_{\hd} \Sigma_{\hd}^{1/2}]_{p+1:, :}$\\
		\STATE $\hat{A}_{\hd} \leftarrow (Z_0^{\top} Z_0)^{-1} Z_0^{\top}Z_1$.
%		\STATE $(\hat{C}_k, \hat{A}_k, \hat{B}_k) = \text{Balanced Truncation}(\hat{C}, \hat{A}, \hat{B})$
		\RETURN $(\hat{C}_{\hd}, \hat{A}_{\hd}, \hat{B}_{\hd})$
	\end{algorithmic}
\end{algorithm}
%%%%%%%%
% \paragraph{Interpretation of $k$:} 
% Our final goal is to find a realization $(C \in \Rb^{p \times n}, A \in \Rb^{n \times n}, B \in \Rb^{n \times m})$ of the underlying model. Since finite data limits the complexity of models that can be learned, $k$ denotes the order of the best lower dimensional approximation that can be learned given data. Specifically, $k$ is the number of singular vectors of $\hHc_{0, d, d}$ that are reliably close to the singular vectors (and values) of $\Hc_{0, \infty, \infty}$. We use these singular values and vectors to obtain model approximations of the true model by balanced truncation.
To state the main result we define a quantity that measures the singular value weighted subspace gap of a matrix $S$:
\[
\Gamma(S, \epsilon) = \sqrt{ {\sigma}^1_{\max}/\zeta_{1}^2 + {\sigma}^2_{\max}/\zeta_{2}^2 + \hdots + {\sigma}^{l}_{\max}/\zeta_{l}^2},
\]
where $S = U \Sigma V^{\top}$ and ${\Sigma}$ is arranged into blocks of singular values such that in each block $i$ we have $\sup_j \sigma^{i}_{j} - \sigma^{i}_{j+1} \leq \epsilon$, \textit{i.e.}, 
\[
{\Sigma} = \begin{bmatrix}
\Lambda_1 & 0 & \ldots & 0 \\
0 & \Lambda_2 & \ldots & 0 \\
\vdots & \vdots & \ddots & 0 \\
0 & 0 & \ldots & \Lambda_l \\
\end{bmatrix}
\]
where $\Lambda_i$ are diagonal matrices, $\sigma^{i}_{j}$ is the $j^{th}$ singular value in the block $\Lambda_i$ and $\sigma^{i}_{\min}, \sigma^{i}_{\max}$ are the minimum and maximum singular values of block $i$ respectively. Furthermore,
\[
\zeta_{i} = \min{({\sigma}^{i-1}_{\min}-{\sigma}^{i}_{\max}, {\sigma}^{i}_{\min}-{\sigma}^{i+1}_{\max})}
\]
for $1 < i < l$, $\zeta_{1} ={\sigma}^{1}_{\min}-{\sigma}^{2}_{\max}$ and $\zeta_{l} = \min{({\sigma}^{l-1}_{\min}-{\sigma}^{l}_{\max}, {\sigma}^{l}_{\min})}$. Informally, the $\zeta_i$ measure the singular value gaps between each blocks. It should be noted that $l$, the number of separated blocks, is a function of $\epsilon$ itself. For example: if $\epsilon = 0$ then the number of blocks correspond to the number of distinct singular values. On the other hand, if $\epsilon$ is very large then $l = 1$.

\begin{thm}
	\label{balanced_truncation}
 Let $M$ be the true unknown model and
\[
\epsilon = 12c\beta R\p*{\sqrt{\frac{m\hd + p\hat{d}^2+\hd \log{\frac{T}{\delta}}}{T}}}.
\]
Then whenever $T \geq T_{*}(\delta)$, we have with probability at least $1-\delta$:
\[
\begin{rcases*}
||C_\hd - \hat{C}_\hd||_2  \\
||B_\hd - \hat{B}_\hd||_2  \\
||A_\hd - \hat{A}_\hd||_2 
\end{rcases*} \leq \bar{\gamma} \epsilon  \Gamma(\hHc_{0, \hd, \hd}, 2\epsilon) + \bar{\gamma}\sup_{1\leq i \leq \hd}\p*{\sqrt{\hsigma^i_{\max}} - \sqrt{\hat{\sigma}^i_{\min}}} + \bar{\gamma} \cdot \frac{\epsilon \wedge \sqrt{\hsigma_{\hd} \epsilon}}{\sqrt{\hsigma_{\hd}}}
\]
where $\sup_{1\leq i \leq \hd}\sqrt{\hsigma^i_{\max}} - \sqrt{\hat{\sigma}^i_{\min}} \leq \frac{2 }{\sqrt{\hsigma_{\hd}}} \epsilon \hd \wedge \sqrt{2 \hd \epsilon}$ and $\bar{\gamma} = \max{(4\gamma, 8)}$.

\end{thm}
The proof of Theorem~\ref{balanced_truncation} follows directly from Theorem~\ref{hd_size} where we show 
\[
||\hHc_{0, \hd, \hd} - \Hc_{0, \infty, \infty}||_2 \leq \epsilon
\]
and Proposition~\ref{prop_c_select}. Theorem~\ref{balanced_truncation} provides an error bound between parameters (of model order $\hd$) when true order is unknown. The subspace gap measure, $\Gamma(\hHc_{0, \hd, \hd}, 2\epsilon)$, is bounded even when $\epsilon = 0$. To see this, note that when $\epsilon = 0$, $\hat{\Hc}_{0, \hd, \hd}$ corresponds exactly to $\Hc_{0, \hd, \hd}$. In that case, the number of blocks correspond to the number of distinct singular values of $\Hc_{0, \hd, \hd}$, and  $\zeta_{n_i}$ then corresponds to singular value gap between the unequal singular values. As a result $\Gamma(\hHc_{0, \hd, \hd}, 2\epsilon) = \Delta < \infty$.  Then the bound decays as $\epsilon = O\p*{\sqrt{\hd^2/T}}$ for singular values $\hat{\sigma}_{\hd} > \hd \epsilon$, but for much smaller singular values the bound decays as $\sqrt{\epsilon} =  O\p*{\p*{d^2/T}^{1/4}}$. 

To shed more light on the behavior of our bounds, we consider the special case of known order. If $n$ is the model order, then we can set $\hd = n$. If $\sigma_i = \sigma_i(\Hc_{0, \infty, \infty})$, then for large enough $T$ one can ensure that 
$$\min_{\sigma_i \neq \sigma_{i+1}}(\sigma_i - \sigma_{i+1})/2 > \epsilon,$$
\textit{i.e.}, $\epsilon$ is less than the singular value gap and small enough that the spectrum of $\hHc_{0,n,n}$ is very close to that of $\Hc_{0, \infty, \infty}$. Consequently $\hsigma_n \geq \sigma_n/2$ and we have that   
\begin{equation}
\begin{rcases*}
||C_n - \hat{C}_n||_2  \\
||B_n - \hat{B}_n||_2  \\
||A_n - \hat{A}_n||_2 
\end{rcases*} \leq  \bar{\gamma} \epsilon \Delta + \bar{\gamma}\epsilon / \sqrt{\sigma_n} = c\beta \bar{\gamma} R\p*{\sqrt{\frac{pn^2+n \log{\frac{T}{\delta}}}{\sigma_n T}}}. \label{eq:known_order} 
\end{equation}
This upper bound is (nearly) identical to the bounds obtained in~\cite{oymak2018non} for the known order case. The major advantage of our result is that we do not require any information/assumption on the LTI system besides $\beta$. Nonparametric approaches to estimating $\beta$ have been studied in~\cite{tu2017non}.   

	\subsection{Order Estimation Lower Bound}
In Theorem~\ref{balanced_truncation} it is shown that whenever $T = \Omega\p*{\frac{1}{\sigma_\hd^2}}$ we can find an accurate $\hd$--order approximation. Now we show that if $T = O\p*{\frac{1}{\sigma_\hd^2}}$ then there is always some non--zero probability with which we can not recover the singular vector corresponding to the $\sigma_{\hd+1}$.  We prove the following lower bound for model order estimation when inputs $\{U_t\}_{t=1}^T$ are active and bounded which we define below
\begin{definition}
\label{active_input}
An input sequence $\{U_t\}_{t=1}^T$ is said to be active if $U_t$ is allowed to depend on past history $\{U_l, Y_l\}_{l=1}^{t-1}$. The input sequence is bounded if $\Ex[U_t^{\top}U_t] \leq 1$ for all $t$.
\end{definition}
Active inputs allow for the case when input selection can be adaptive due to feedback.
\begin{thm}
    \label{inf_lowerbnd}
Fix $\delta > 0, \zeta \in (0, 1/2)$. Let $M_1, M_2$ be two LTI systems and $\sigma_i^{(1)}, \sigma_i^{(2)}$ be the $i^{th}$-Hankel singular values respectively. Let $\frac{\sigma^{(1)}_1}{\sigma^{(1)}_2} \leq  \frac{2}{\zeta}$ and $\sigma^{(2)}_2 = 0$. Then whenever $T \leq \frac{\Cc R^2}{\zeta^2}\log{\frac{2}{\delta}}$ we have
\[
\sup_{M \in \{M_1, M_2\}}\Pb_{Z_T \sim M}(\text{order}(\hat{M}(Z_T)) \neq \text{order}(M)) \geq \delta
\]
Here $Z_T=\{U_t, Y_t\}_{t=1}^T \sim M$ means $M$ generates $T$ data points $\{Y_t\}_{t=1}^T$ in response to active and bounded inputs $\{U_t\}_{t=1}^T$ and $\hat{M}(Z_T)$ is any estimator. 
\end{thm}
\begin{proof}
The proof can be found in appendix in Section~\ref{lower_bnd} and involves using Fano's (or Birge's) inequality to compute the minimax risk between the probability density functions generated by two different LTI systems:
\begin{align}
\label{canonical_form}
A_0 &= \begin{bmatrix}
0 & 1 & 0\\
0 & 0 & 0 \\
\zeta & 0 & 0
\end{bmatrix}, A_1 = A_0 , B_0 = \begin{bmatrix}
0 \\
0 \\
\sqrt{\beta}/R
\end{bmatrix},  B_1 = \begin{bmatrix}
0 \\
\sqrt{\beta}/R \\
\sqrt{\beta}/R
\end{bmatrix}, C_0 = \begin{bmatrix}
0 & 0 & \sqrt{\beta} R
\end{bmatrix}, C_1 =C_0.
\end{align}
$A_0, A_1$ are Schur stable whenever $|\zeta| < 1$.
\end{proof}
Theorem~\ref{inf_lowerbnd} shows that when the time required to recover higher order models depends inversely on the condition number. Specifically, to correctly distinguish between an order $1$ and order $2$ model $T \geq \Omega(2/\zeta^2)$ where $\zeta$ is the condition number of the $2$-order model. We compare this to our upper bound in Theorem~\ref{balanced_truncation} and Eq.~\eqref{eq:known_order}, assume $\Gamma(\hHc_{0, \hd, \hd}, 2\epsilon) \leq \Delta$ for all $\epsilon \in [0,1]$ and $\hd \epsilon \leq \hat{\sigma}_{\hd}$, then since parameter error, $\Ec$, is upper bounded as
\[
\Ec \leq c\beta \Delta R\p*{\sqrt{\frac{m \hd + p\hd^2+\hd \log{\frac{T}{\delta}}}{\sigma_{\hd} T}}}, 
\]
we need
\[
\frac{T}{\log{\frac{T}{\delta}}} \geq \Omega \p*{\frac{\beta^2 \Delta^2 R^2 \hd^2}{\sigma_{\hd}^2}}
\]
to correctly identify $\hd$-order model. The ratio $(\beta/\sigma_{\hd})$ is equal to the condition number of the Hankel matrix. In this sense, the model selection followed by singular value thresholding is not too conservative in terms of $R$ (the signal-to-noise ratio) and conditioning of the Hankel matrix.
	\section{Experiments}
\label{experiments}
The experiments in this paper are for the single trajectory case. A detailed analysis for system identification from multiple trajectories can be found in~\cite{tu2017non}. Suppose that the LTI system generating data, $M$, has transfer function given by
\begin{equation}
\label{fir}
G(z) = \alpha_0 + \sum_{l=1}^{149} \alpha_l \rho^l z^{-l}, \hspace{2mm}  \rho < 1
\end{equation}
where $\alpha_i \sim \Nc(0, 1)$. $M$ is a finite dimensional LTI system or order $150$ with parameters as $M = (C \in \Rb^{1 \times 150}, A \in \Rb^{150 \times 150}, B \in \Rb^{150 \times 1})$. For these illustrations, we assume a balanced system and choose $R=1, \delta = 0.05$. We estimate $\beta_{0.6} = 15, \beta_{0.9 }= 40, \beta_{0.99} = 140$, pick $U_t \sim \Nc(0, 1)$ and $\{w_t, \eta_t\} \sim \{\Nc(0, 1), \Nc(0, I) \}$ respectively.
% \begin{figure}
%     \centering
%     \begin{subfigure}[b]{0.3\textwidth}
%     \centering
%         \includegraphics[width=\textwidth]{images/rh_90_d_choice.jpg}
%         \caption{Variation of $d = \hd$ as $T$}
%         \label{fig:d_choice}
%     \end{subfigure}
%     \hfill %add desired spacing between images, e. g. ~, \quad, \qquad, \hfill etc. 
%       %(or a blank line to force the subfigure onto a new line)
%     \begin{subfigure}[b]{0.3\textwidth}
%     \centering
%         \includegraphics[width=0.1\textwidth]{images/rh_90_r_choice.jpg}
%         \caption{A tiger}
%         \label{fig:tiger}
%     \end{subfigure}
% \end{figure}
%%%%%%%%%%%%%%%%%%%%%%%CHANGE HERE%%%%%%%%%%%%%%%%%%%%%
\begin{figure}[H]
	\centering
	\includegraphics[width=0.8\textwidth]{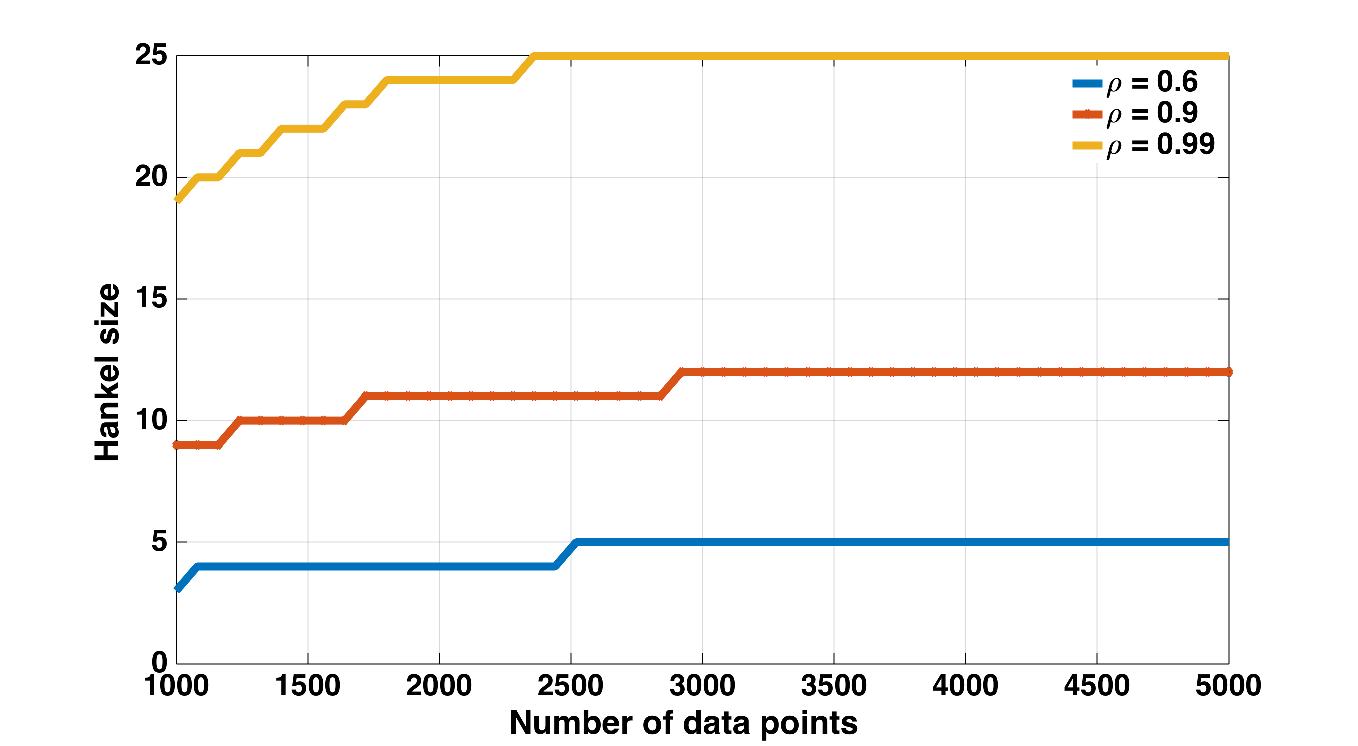}
	\caption{Variation of Hankel size $=\hd$ with $T$ for different values of $\rho$}
	\label{fig:d_choice}
\end{figure}
%%%%%%%%%%%%%%%%%%%%%%%CHANGE HERE%%%%%%%%%%%%%%%%%%%%%
Fig.~\ref{fig:d_choice} shows how $d = \hd$ change with the number of data points for different values of $\rho$. When $\rho = 0.6$, \textit{i.e.}, small, $\hd$ does not grow too big with $T$ even when the number of data points is increased. This shows that a small model order is sufficient to specify system dynamics. On the other hand, when $\rho = 0.99$, \textit{i.e.}, closer to instability the $\hd$ required is much larger, indicating the need for a higher order.  Although $\hd$ implicitly captures the effect of spectral radius, the knowledge of $\rho$ is not required for $\hd$ selection.  

In principle, our algorithm increases the Hankel size to the ``appropriate'' size as the data increases. We compare this to a deterministic growth policy $d = \log{(T)}$ and the SSREGEST algorithm~\cite{ljung2015regularization}. The SSREGEST algorithm first learns a large model from data and then performs model reduction to obtain a final model. In contrast, we go to reduced model directly by picking a small $\hd$. This reduces the sensitivity to noise. 

In Fig.~\ref{fig:err_log} shows the model errors for a deterministic growth policy $d = \log{(T)}$ and our algorithm. Although the difference is negligible when $\rho = 0.6$ (small), we see that our algorithm does better $\rho = 0.99$ due to its adaptive nature, \textit{i.e.}, $\hd$ responds faster for our algorithm.
\begin{figure}[H]
	\centering
	\includegraphics[width=\textwidth]{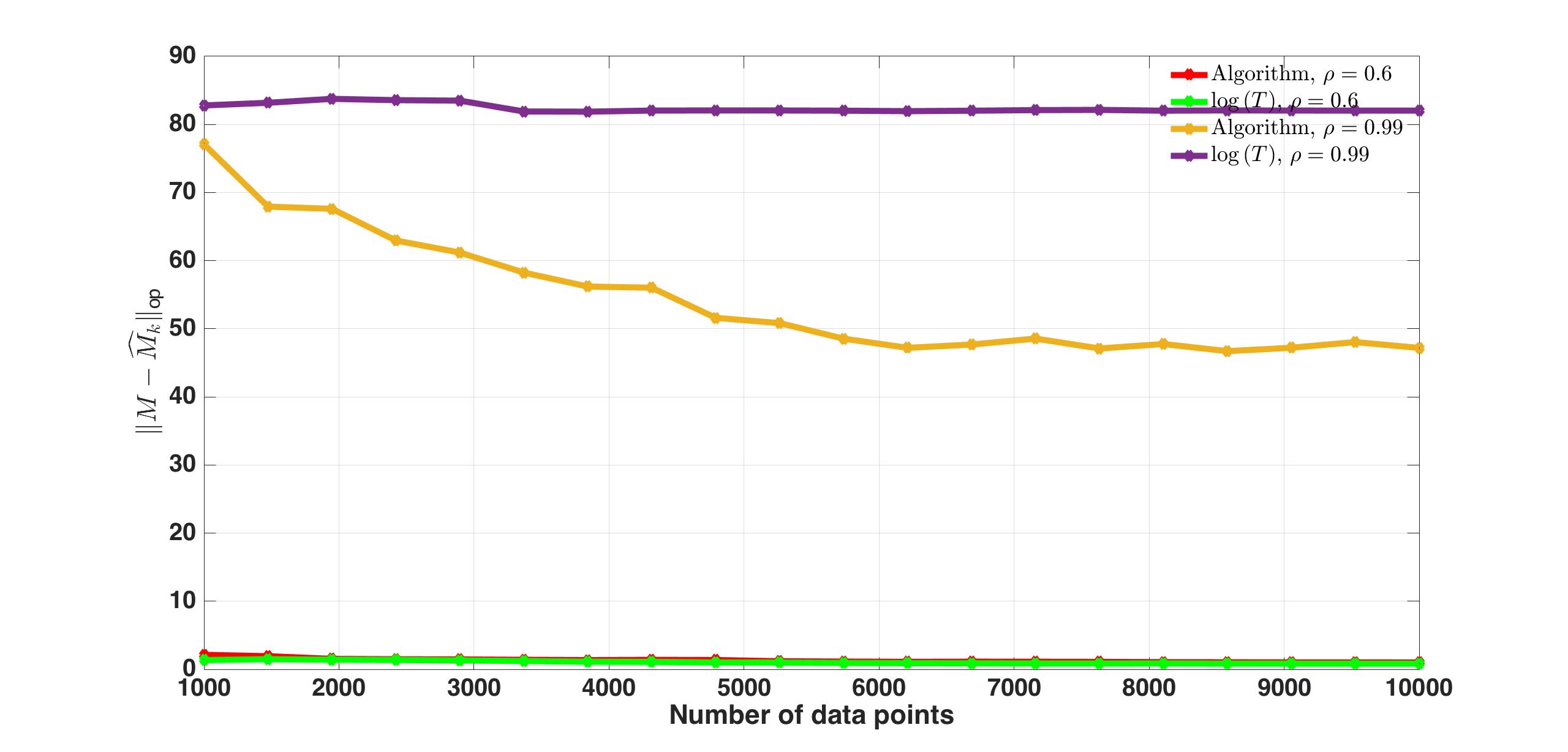}
	\caption{Variation of $||M - \widehat{M}_k||_{\text{op}}$ for different values of $\rho$. Here $k=\hd$ for our algorithm and $k = \log{(T)}$. Furthermore, $||\cdot||_{\text{op}}$ is the Hankel norm.}
	\label{fig:err_log}
\end{figure}

Finally, for the case when $\rho = 0.9, \beta = 40$, we show the model errors for SSREGEST and our algorithm as $T$ increases. Although asymptotically both algorithms perform the same, it is clear that for small $T$ our algorithm is more robust to the presence of noise.
\begin{center}
	\begin{tabular}{ | c | c | c|}
		\hline
		T	  &  SSREGEST & Our Algorithm \\
		\hline
		$500$ & $6.21 \pm 1.35$  & $13.37 \pm 3.7$\\
		\hline
		$\approx 850$ & $30.20 \pm 7.55$ & $11.25 \pm 2.89$  \\
		\hline	
		$\approx 1200$ & $26.80 \pm 8.94$ & $9.83 \pm 2.60$ \\
		\hline
		$1500$ & $23.27 \pm 10.65$ & $9.17 \pm 2.30$\\
		\hline	
		$2000$ & $26.38 \pm 12.88$ & $7.70 \pm 1.60$ \\
		\hline 
	\end{tabular}
\end{center}
%
%
%
%Now we fix $\rho = 0.9$ for the remainder of the experiments. Although the number of data points increases to $T = 5 \times 10^3$, observe that in Fig.~\ref{fig:d_choice}  $\hd$, the size of $\hHc_{0, \hd, \hd}$, does not increase more than $15$. This suggests that the LTI system can be represented very parsimoniously. For parameter estimation, we need to recover the singular vectors of $\hHc_{0, \hd, \hd}$. Although the true model order is $150$, due to noisy data we can only recover an order $15$ approximation. Furthermore, in Fig.~\ref{fig:r_err_choice} we plot the error between the true order--$k$ approximation, $G_k$, and the estimated order--$k$ approximation, ${G}^{\text{est}}_k$. For $k = 4, 8, 10$ the errors are low and comparable, but when $k =  16$ the error increases. This suggests that the theoretical threshold is not too conservative.
%%  \begin{figure}[H]
%%      \subfigure[$||G-G^{\text{est}}_k||_H$ vs. $T$ for $k = \hd$]{\includegraphics[width=0.5\textwidth]{images/error_tru_model_rh_90.jpg}\label{fig:tru_err_choice}}
%%      \subfigure[$||G_k-G^{\text{est}}_k||_H$ vs. $T$]{\includegraphics[width=0.5\textwidth]{images/error_r_90.jpg}\label{fig:r_err_choice}}
%%      \caption{Error between true and estimated models}
%%      \label{figs:errors}
%% \end{figure}
%In Fig.~\ref{fig:tru_err_choice} we also show the error between estimated model approximations of of different model orders and the true model.
	\section{Discussion}
\label{discussion}
We propose a new approach to system identification when we observe only finite noisy data. Typically, the order of an LTI system is large and unknown and a priori parametrizations may fail to yield accurate estimates of the underlying system. However, our results suggest that there always exists a lower order approximation of the original LTI system that can be learned with high probability. The central theme of our approach is to recover a good lower order approximation that can be accurately learned. Specifically, we show that identification of such approximations is closely related to the singular values of the system Hankel matrix. In fact, the time required to learn a $\hd$--order approximation scales as $T = \Omega(\frac{\beta^2}{\sigma_\hd^2})$ where $\sigma_\hd$ is the $\hd$--the singular value of system Hankel matrix. This means that system identification does not explicitly depend on the model order $n$, rather depends on $n$ through $\sigma_n$. As a result, in the presence of finite data it is preferable to learn only the ``significant'' (and perhaps much smaller) part of the system when $n$ is very large and $\sigma_n \ll 1$. Algorithm~\ref{alg:learn_ls} and \ref{alg:svd} provide a guided mechanism for learning the parameters of such significant approximations with optimal rules for hyperparameter selection given in Algorithm~\ref{alg:d_choice}. 

Future directions for our work include extending the existing low--rank optimization-based identification techniques, such as~\citep{fazel2013hankel,grussler2018low}, which typically lack statistical guarantees. Since Hankel based operators occur quite naturally in general (not necessarily linear) dynamical systems, exploring if our methods could be extended for identification of such systems appears to be an exciting direction.
	\bibliography{bibliography}

\begin{thebibliography}{38}
\providecommand{\natexlab}[1]{#1}
\providecommand{\url}[1]{\texttt{#1}}
\expandafter\ifx\csname urlstyle\endcsname\relax
  \providecommand{\doi}[1]{doi: #1}\else
  \providecommand{\doi}{doi: \begingroup \urlstyle{rm}\Url}\fi

\bibitem[Abbasi-Yadkori et~al.(2011)Abbasi-Yadkori, P{\'a}l, and
  Szepesv{\'a}ri]{abbasi2011improved}
Yasin Abbasi-Yadkori, D{\'a}vid P{\'a}l, and Csaba Szepesv{\'a}ri.
\newblock Improved algorithms for linear stochastic bandits.
\newblock In \emph{Advances in Neural Information Processing Systems}, pages
  2312--2320, 2011.

\bibitem[A{\c{c}}{\i}kme{\c{s}}e et~al.(2013)A{\c{c}}{\i}kme{\c{s}}e, Carson,
  and Blackmore]{2013lossless}
Beh{\c{c}}et A{\c{c}}{\i}kme{\c{s}}e, John~M Carson, and Lars Blackmore.
\newblock Lossless convexification of nonconvex control bound and pointing
  constraints of the soft landing optimal control problem.
\newblock \emph{IEEE Transactions on Control Systems Technology}, 21\penalty0
  (6):\penalty0 2104--2113, 2013.

\bibitem[Agarwal et~al.(2018)Agarwal, Amjad, Shah, and Shen]{agarwal2018time}
Anish Agarwal, Muhammad~Jehangir Amjad, Devavrat Shah, and Dennis Shen.
\newblock Time series analysis via matrix estimation.
\newblock \emph{arXiv preprint arXiv:1802.09064}, 2018.

\bibitem[Allen-Zhu and Li(2016)]{allen2016lazysvd}
Zeyuan Allen-Zhu and Yuanzhi Li.
\newblock Lazysvd: Even faster svd decomposition yet without agonizing pain.
\newblock In \emph{Advances in Neural Information Processing Systems}, pages
  974--982, 2016.

\bibitem[Bauer(2000)]{bauer2000order}
Dietmar Bauer.
\newblock Order estimation for subspace methods.
\newblock 2000.

\bibitem[Boucheron et~al.(2013)Boucheron, Lugosi, and
  Massart]{boucheron2013concentration}
St{\'e}phane Boucheron, G{\'a}bor Lugosi, and Pascal Massart.
\newblock \emph{Concentration inequalities: A nonasymptotic theory of
  independence}.
\newblock Oxford university press, 2013.

\bibitem[Campi and Weyer(2002)]{campi2002finite}
Marco~C Campi and Erik Weyer.
\newblock Finite sample properties of system identification methods.
\newblock \emph{IEEE Transactions on Automatic Control}, 47\penalty0
  (8):\penalty0 1329--1334, 2002.

\bibitem[Dempster et~al.(1977)Dempster, Laird, and Rubin]{roweis1999unifying}
Arthur~P Dempster, Nan~M Laird, and Donald~B Rubin.
\newblock Maximum likelihood from incomplete data via the em algorithm.
\newblock \emph{Journal of the royal statistical society. Series B
  (methodological)}, pages 1--38, 1977.

\bibitem[Faradonbeh et~al.(2017)Faradonbeh, Tewari, and
  Michailidis]{faradonbeh2017finite}
Mohamad Kazem~Shirani Faradonbeh, Ambuj Tewari, and George Michailidis.
\newblock Finite time identification in unstable linear systems.
\newblock \emph{arXiv preprint arXiv:1710.01852}, 2017.

\bibitem[Fazel et~al.(2013)Fazel, Pong, Sun, and Tseng]{fazel2013hankel}
Maryam Fazel, Ting~Kei Pong, Defeng Sun, and Paul Tseng.
\newblock Hankel matrix rank minimization with applications to system
  identification and realization.
\newblock \emph{SIAM Journal on Matrix Analysis and Applications}, 34\penalty0
  (3):\penalty0 946--977, 2013.

\bibitem[Glover(1984)]{glover1984all}
Keith Glover.
\newblock All optimal hankel-norm approximations of linear multivariable
  systems and their $l_{\infty}$-error bounds.
\newblock \emph{International journal of control}, 39\penalty0 (6):\penalty0
  1115--1193, 1984.

\bibitem[Glover(1987)]{glover1987model}
Keith Glover.
\newblock Model reduction: a tutorial on hankel-norm methods and lower bounds
  on l2 errors.
\newblock \emph{IFAC Proceedings Volumes}, 20\penalty0 (5):\penalty0 293--298,
  1987.

\bibitem[Goldenshluger(1998)]{goldenshluger1998nonparametric}
Alexander Goldenshluger.
\newblock Nonparametric estimation of transfer functions: rates of convergence
  and adaptation.
\newblock \emph{IEEE Transactions on Information Theory}, 44\penalty0
  (2):\penalty0 644--658, 1998.

\bibitem[Grussler et~al.(2018)Grussler, Rantzer, and
  Giselsson]{grussler2018low}
Christian Grussler, Anders Rantzer, and Pontus Giselsson.
\newblock Low-rank optimization with convex constraints.
\newblock \emph{IEEE Transactions on Automatic Control}, 2018.

\bibitem[Hardt et~al.(2016)Hardt, Ma, and Recht]{hardt2016gradient}
Moritz Hardt, Tengyu Ma, and Benjamin Recht.
\newblock Gradient descent learns linear dynamical systems.
\newblock \emph{arXiv preprint arXiv:1609.05191}, 2016.

\bibitem[Hazan et~al.(2018)Hazan, Lee, Singh, Zhang, and
  Zhang]{hazan2018spectral}
Elad Hazan, Holden Lee, Karan Singh, Cyril Zhang, and Yi~Zhang.
\newblock Spectral filtering for general linear dynamical systems.
\newblock \emph{arXiv preprint arXiv:1802.03981}, 2018.

\bibitem[Ho and Kalman(1966)]{ho1966effective}
BL~Ho and Rudolph~E Kalman.
\newblock Effective construction of linear state-variable models from
  input/output functions.
\newblock \emph{at-Automatisierungstechnik}, 14\penalty0 (1-12):\penalty0
  545--548, 1966.

\bibitem[Kung and Lin(1981)]{kung1981optimal}
S~Kung and D~Lin.
\newblock Optimal hankel-norm model reductions: Multivariable systems.
\newblock \emph{IEEE Transactions on Automatic Control}, 26\penalty0
  (4):\penalty0 832--852, 1981.

\bibitem[Ljung(1987)]{ljung1987system}
Lennart Ljung.
\newblock \emph{System identification: theory for the user}.
\newblock Prentice-hall, 1987.

\bibitem[Ljung et~al.(2015)Ljung, Singh, and Chen]{ljung2015regularization}
Lennart Ljung, Rajiv Singh, and Tianshi Chen.
\newblock Regularization features in the system identification toolbox.
\newblock \emph{IFAC-PapersOnLine}, 48\penalty0 (28):\penalty0 745--750, 2015.

\bibitem[Meckes et~al.(2007)]{meckes2007spectral}
Mark Meckes et~al.
\newblock On the spectral norm of a random toeplitz matrix.
\newblock \emph{Electronic Communications in Probability}, 12:\penalty0
  315--325, 2007.

\bibitem[Oymak and Ozay(2018)]{oymak2018non}
Samet Oymak and Necmiye Ozay.
\newblock Non-asymptotic identification of lti systems from a single
  trajectory.
\newblock \emph{arXiv preprint arXiv:1806.05722}, 2018.

\bibitem[Pe{\~n}a et~al.(2008)Pe{\~n}a, Lai, and Shao]{pena2008self}
Victor~H Pe{\~n}a, Tze~Leung Lai, and Qi-Man Shao.
\newblock \emph{Self-normalized processes: Limit theory and Statistical
  Applications}.
\newblock Springer Science \& Business Media, 2008.

\bibitem[Sarkar and Rakhlin(2018)]{sarkar2018}
Tuhin Sarkar and Alexander Rakhlin.
\newblock How fast can linear dynamical systems be learned?
\newblock \emph{arXiv preprint arXiv:1812.0125}, 2018.

\bibitem[Shah et~al.(2012)Shah, Bhaskar, Tang, and Recht]{shah2012linear}
Parikshit Shah, Badri~Narayan Bhaskar, Gongguo Tang, and Benjamin Recht.
\newblock Linear system identification via atomic norm regularization.
\newblock In \emph{2012 IEEE 51st IEEE Conference on Decision and Control
  (CDC)}, pages 6265--6270. IEEE, 2012.

\bibitem[Shibata(1976)]{shibata1976selection}
Ritei Shibata.
\newblock Selection of the order of an autoregressive model by akaike's
  information criterion.
\newblock \emph{Biometrika}, 63\penalty0 (1):\penalty0 117--126, 1976.

\bibitem[Simchowitz et~al.(2018)Simchowitz, Mania, Tu, Jordan, and
  Recht]{simchowitz2018learning}
Max Simchowitz, Horia Mania, Stephen Tu, Michael~I Jordan, and Benjamin Recht.
\newblock Learning without mixing: Towards a sharp analysis of linear system
  identification.
\newblock \emph{arXiv preprint arXiv:1802.08334}, 2018.

\bibitem[Tu et~al.(2017)Tu, Boczar, Packard, and Recht]{tu2017non}
Stephen Tu, Ross Boczar, Andrew Packard, and Benjamin Recht.
\newblock Non-asymptotic analysis of robust control from coarse-grained
  identification.
\newblock \emph{arXiv preprint arXiv:1707.04791}, 2017.

\bibitem[Tu et~al.(2018{\natexlab{a}})Tu, Boczar, and
  Recht]{tu2018approximation}
Stephen Tu, Ross Boczar, and Benjamin Recht.
\newblock On the approximation of toeplitz operators for nonparametric
  $\mathcal{H}_{\infty}$--norm estimation.
\newblock In \emph{2018 Annual American Control Conference (ACC)}, pages
  1867--1872. IEEE, 2018{\natexlab{a}}.

\bibitem[Tu et~al.(2018{\natexlab{b}})Tu, Boczar, and Recht]{tu2018minimax}
Stephen Tu, Ross Boczar, and Benjamin Recht.
\newblock Minimax lower bounds for $\mathcal{H}_{\infty} $-norm estimation.
\newblock \emph{arXiv preprint arXiv:1809.10855}, 2018{\natexlab{b}}.

\bibitem[Tyrtyshnikov(2012)]{tyrtyshnikov2012brief}
Eugene~E Tyrtyshnikov.
\newblock \emph{A brief introduction to numerical analysis}.
\newblock Springer Science \& Business Media, 2012.

\bibitem[van~de Geer and Lederer(2013)]{van2013bernstein}
Sara van~de Geer and Johannes Lederer.
\newblock The bernstein--orlicz norm and deviation inequalities.
\newblock \emph{Probability theory and related fields}, 157\penalty0
  (1-2):\penalty0 225--250, 2013.

\bibitem[Van Der~Vaart and Wellner(1996)]{van1996weak}
Aad~W Van Der~Vaart and Jon~A Wellner.
\newblock Weak convergence.
\newblock In \emph{Weak convergence and empirical processes}, pages 16--28.
  Springer, 1996.

\bibitem[Van~Overschee and De~Moor(2012)]{van2012subspace}
Peter Van~Overschee and BL~De~Moor.
\newblock \emph{Subspace identification for linear systems:
  Theory—Implementation—Applications}.
\newblock Springer Science \& Business Media, 2012.

\bibitem[Venkatesh and Dahleh(2001)]{venkatesh2001system}
Saligrama~R Venkatesh and Munther~A Dahleh.
\newblock On system identification of complex systems from finite data.
\newblock \emph{IEEE Transactions on Automatic Control}, 46\penalty0
  (2):\penalty0 235--257, 2001.

\bibitem[Vershynin(2010)]{vershynin2010introduction}
Roman Vershynin.
\newblock Introduction to the non-asymptotic analysis of random matrices.
\newblock \emph{arXiv preprint arXiv:1011.3027}, 2010.

\bibitem[Wedin(1972)]{wedin1972perturbation}
Per-{\AA}ke Wedin.
\newblock Perturbation bounds in connection with singular value decomposition.
\newblock \emph{BIT Numerical Mathematics}, 12\penalty0 (1):\penalty0 99--111,
  1972.

\bibitem[Zhou et~al.(1996)Zhou, Doyle, and Glover]{zhou1996robust}
K~Zhou, JC~Doyle, and K~Glover.
\newblock Robust and optimal control, 1996.

\end{thebibliography}


\begin{thebibliography}{6}
\providecommand{\natexlab}[1]{#1}
\providecommand{\url}[1]{\texttt{#1}}
\expandafter\ifx\csname urlstyle\endcsname\relax
  \providecommand{\doi}[1]{doi: #1}\else
  \providecommand{\doi}{doi: \begingroup \urlstyle{rm}\Url}\fi

\bibitem[Krahmer et~al.(2014)Krahmer, Mendelson, and
  Rauhut]{krahmer2014suprema}
Felix Krahmer, Shahar Mendelson, and Holger Rauhut.
\newblock Suprema of chaos processes and the restricted isometry property.
\newblock \emph{Communications on Pure and Applied Mathematics}, 67\penalty0
  (11):\penalty0 1877--1904, 2014.

\bibitem[Meckes et~al.(2007)]{meckes2007spectral}
Mark Meckes et~al.
\newblock On the spectral norm of a random toeplitz matrix.
\newblock \emph{Electronic Communications in Probability}, 12:\penalty0
  315--325, 2007.

\bibitem[Tyrtyshnikov(2012)]{tyrtyshnikov2012brief}
Eugene~E Tyrtyshnikov.
\newblock \emph{A brief introduction to numerical analysis}.
\newblock Springer Science \& Business Media, 2012.

\bibitem[van~de Geer and Lederer(2013)]{van2013bernstein}
Sara van~de Geer and Johannes Lederer.
\newblock The bernstein--orlicz norm and deviation inequalities.
\newblock \emph{Probability theory and related fields}, 157\penalty0
  (1-2):\penalty0 225--250, 2013.

\bibitem[Van Der~Vaart and Wellner(1996)]{van1996weak}
Aad~W Van Der~Vaart and Jon~A Wellner.
\newblock Weak convergence.
\newblock In \emph{Weak convergence and empirical processes}, pages 16--28.
  Springer, 1996.

\bibitem[Vershynin(2010)]{vershynin2010introduction}
Roman Vershynin.
\newblock Introduction to the non-asymptotic analysis of random matrices.
\newblock \emph{arXiv preprint arXiv:1011.3027}, 2010.

\end{thebibliography}
	\newpage
	\tableofcontents
	\newpage
	\section{Preliminaries}
\label{sec:preliminaries}
\begin{thm}[Theorem 5.39~\cite{vershynin2010introduction}]
	\label{539_versh}
	if $E$ is a $T \times md$ matrix with independent sub--Gaussian isotropic rows with subGaussian parameter $1$ then with probability at least $1 - 2e^{-ct^2}$ we have
	\[
	\sqrt{T} - C \sqrt{md} - t \leq \sigma_{\min}(E) \leq  \sqrt{T} + C \sqrt{md} + t
	\]
\end{thm}
\begin{prop}[\cite{vershynin2010introduction}]
	\label{eps_net}
	We have for any $\epsilon < 1$ and any $w \in \Sc^{d-1}$ that 
	\[
	\Pb(||M|| > z) \leq (1 + 2/\epsilon)^d \Pb\p*{||Mw|| > \frac{z}{(1-\epsilon)}}
	\]
\end{prop}
\begin{thm}[Theorem 1~\cite{meckes2007spectral}]]
	\label{toep_norm}
	Suppose $\{X_i \in \Rb^{m}\}_{i=1}^{\infty}$ are independent, $\Ex[X_j] = \textbf{0}$ for all $j$, and $X_{ij}$ are independent $\subg(1)$ random variables. Then $\Pb(||T_d|| \geq c m \sqrt{d\log{2d}} + t) \leq e^{-t^2/d}$ where 
	\[
	T_n = \begin{bmatrix}
	X_0 & X_1 & \hdots & X_{d-1} \\
	X_1 & X_0 & \hdots & X_{d-2} \\
	\vdots & \ddots & \ddots & \vdots \\
	X_{d-1} & \hdots & \hdots & X_0
	\end{bmatrix}
	\]
\end{thm}
%\begin{proof}
%	The proof of this exact statement can be combined with Theorem 1 and Page 319 in~\cite{meckes2007spectral}.
%\end{proof}
\begin{thm}[Hanson--Wright Inequality]
	\label{thm:hanson_wright}
	Given a subgaussian vector $X = [X_1, X_2, \hdots, X_n] \in \Rb^n$ with $\sup_i \nrm*{X_i}_{\psi_2} \leq K$. Then for any $B \in \Rb^{n \times n}$ and $t \geq 0$
	\[
	\Pb\p*{\|XBX^{\top} - \Ex[XBX^{\top}]\| \leq t} \leq 2 \exp\p*{\max \p*{\frac{-ct}{K^2 \nrm{B}},  \frac{-ct^2}{K^4 \nrm{B}^2_{HS}}}}.
	\]
\end{thm}

\begin{prop}[Lecture 2 \cite{tyrtyshnikov2012brief}]
	\label{prop:lower_tri}
	Suppose that $L$ is the lower triangular part of a matrix $A \in \Rb^{d \times d}$. Then 
	\[
	\nrm*{L}_2 \leq \log_2{(2d)} \nrm*{A}_2.
	\]
\end{prop}

Let $\psi$ be a nondecreasing, convex function with $\psi(0) = 0$ and $X$ a random variable. Then the Orlicz norm $||X||_{\psi}$ is defined as 
\[
||X||_{\psi} = \inf \Big\{\alpha > 0: \Ex[\psi(|X|/\alpha)] \leq 1\Big\}.
\]
Let $(B, d)$ be an arbitrary semi--metric space. Denote by $N(\epsilon, d)$ is the minimal number of balls of radius $\epsilon$ needed to cover $B$.

\begin{thm}[Corollary 2.2.5 in~\cite{van1996weak}]
	\label{thm:emp_process}
	The constant $K$ can be hosen such that 
	\[
	||\sup_{s, t}|X_s -X_t|||_{\psi} \leq K \int_{0}^{\text{diam}(B)} \psi^{-1}(N(\epsilon/2, d))d\epsilon
	\]
	where $\text{diam}(B)$ is the diameter of $B$ and $d(s, t) = ||X_s - X_t||_{\psi}$.
\end{thm}
\begin{thm}[Theorem 1 in~\cite{abbasi2011improved}]
	\label{thm:selfnorm_main}
	Let $\{\bcF_t\}_{t=0}^{\infty}$ be a filtration. Let $\{\eta_{t} \in \Rb^m, X_t \in \Rb^d\}_{t=1}^{\infty}$ be stochastic processes such that $\eta_t, X_t$ are $\bcF_t$ measurable and $\eta_t$ is $\bcF_{t-1}$-conditionally $\subg(L^2)$ for some $L > 0$.
	For any $t \geq 0$, define 
	$
	V_t = \sum_{s=1}^t X_s X_s^{\prime}, S_t = \sum_{s=1}^t  X_s\eta_{s+1}^{\top}
	$.
	Then for any $\delta > 0, V \succ 0$ and all $t \geq 0$ we have with probability at least $1-\delta$ 
	\[
	S_t^{\top}(V + V_t)^{-1}S_t  \leq  2L^2 \p*{\log{\frac{1}{\delta}} + \log{\frac{\text{det}(V+V_t)}{\text{det}(V)}} + m}.
	\]
\end{thm} 
\begin{proof}
Define $M = (V + V_t)^{-1/2}S_t$. Now we use Proposition~\ref{eps_net} and setting $\epsilon = 1/2$,
\[
\Pb(||M||_2 > z) \leq 5^m \Pb(||Mw||_2 > 2z)
\]
for $w \in \Sc^{m-1}$. Then we can use Theorem 1 in~\cite{abbasi2011improved}, and with probability at least $1 - \delta$ we have 
\[
||Mw||^2_2 \leq 2L^2 \p*{\log{\frac{1}{\delta}} + \log{\frac{\text{det}(V+V_t)}{\text{det}(V)}}}.
\]
By $\delta \rightarrow 5^{-m} \delta$, we have with probability at least $1 - 5^{-m} \delta$
\[
||Mw||_2 \leq \sqrt{2}L  \sqrt{\p*{m \log{(5)} + \log{\frac{1}{\delta}} + \log{\frac{\text{det}(V+V_t)}{\text{det}(V)}}}}.
\]
Then with probability at least $1 - \delta$,
\[
||M||_2 \leq \sqrt{\frac{\log{(5)}}{2}} L \sqrt{\p*{m  + \log{\frac{1}{\delta}} + \log{\frac{\text{det}(V+V_t)}{\text{det}(V)}}}}.
\]
\end{proof}
	
\begin{lem}
	\label{bound_toeplitz}
	For any $M = (C, A, B)$, we have that 
	\[
	||\Bc^{v}_{T \times mT}|| = \sqrt{ \sigma\Big(\sum_{k=1}^d \Tc_{d+k, T}^{\top}\Tc_{d+k, T}\Big)}
	\]
	Here $\Bc^{v}_{T \times mT}$ is defined as follows: $\beta = \Hc_{d, d, T}^{\top} v = [\beta_1^{\top}, \beta_2^{\top}, \ldots, \beta_T^{\top}]^{\top}$.
	\begin{align*}
	\Bc^{v}_{T \times mT}= \begin{bmatrix}
	\beta_1^{\top} & 0 & 0& \ldots \\
	\beta^{\top}_2 & \beta^{\top}_1 & 0 & \ldots\\
	\vdots & \vdots & \ddots & \vdots \\
	\beta_T^{\top} & \beta_{T-1}^{\top}& \ldots &  \beta^{\top}_1
	\end{bmatrix}
	\end{align*}
	and $||v||_2 = 1$.
\end{lem}
\begin{proof}
	For the matrix $\Bc^{v}$ we have 
	\begin{align*}
	\Bc^{v} u &= \begin{bmatrix}
	\beta_1^{\top} u_1 \\
	\beta_1^{\top} u_2 + \beta_2^{\top} u_1 \\
	\beta_1^{\top} u_3 + \beta_2^{\top} u_2 + \beta_3^{\top} u_1 \\
	\vdots \\
	\beta_1^{\top} u_T + \beta_2^{\top} u_{T-1} + \ldots + \beta_T^{\top} u_1
	\end{bmatrix} = \begin{bmatrix}
	v^{\top}\begin{bmatrix}
	CA^{d+1}B u_1 \\
	CA^{d+2}B u_1 \\
	\vdots \\
	CA^{2d}B u_1
	\end{bmatrix} \\
	v^{\top}\begin{bmatrix}
	CA^{d+2}B u_1 + CA^{d+1}B u_2 \\
	CA^{d+3}B u_1 + CA^{d+2}B u_2 \\
	\vdots \\
	CA^{2d+1}B u_1 + CA^{2d}B u_2 
	\end{bmatrix} \\
	\vdots \\
	v^{\top}\begin{bmatrix}
	CA^{T+d}B u_1 + \hdots + CA^{d+1}B u_T\\
	CA^{T+d+2}B u_1  + \hdots + CA^{d+2}B u_T\\
	\vdots \\
	CA^{T+2d-1}B u_1 + \hdots + CA^{2d}B u_T
	\end{bmatrix}
	\end{bmatrix} \\
	&= \Vc \begin{bmatrix}
	\begin{bmatrix}
	CA^{d+1}B u_1 \\
	CA^{d+2}B u_1 \\
	\vdots \\
	CA^{2d}B u_1
	\end{bmatrix} \\
	\begin{bmatrix}
	CA^{d+2}B u_1 + CA^{d+1}B u_2 \\
	CA^{d+3}B u_1 + CA^{d+2}B u_2 \\
	\vdots \\
	CA^{2d+1}B u_1 + CA^{2d}B u_2 
	\end{bmatrix} \\
	\vdots \\
	\begin{bmatrix}
	CA^{T+d}B u_1 + \hdots + CA^{d+1}B u_T\\
	CA^{T+d+2}B u_1  + \hdots + CA^{d+2}B u_T\\
	\vdots \\
	CA^{T+2d-1}B u_1 + \hdots + CA^{2d}B u_T
	\end{bmatrix}
	\end{bmatrix}\\
	&= \Vc \underbrace{\begin{bmatrix}
		CA^{d+1}B & 0 & 0 & \hdots & 0\\
		CA^{d+2}B & 0 & 0 & \hdots & 0 \\
		\vdots & \vdots & \vdots & \vdots & \vdots\\
		CA^{2d}B & 0 & 0 & \hdots & 0 \\
		CA^{d+2}B & CA^{d+1}B & 0 & \hdots & 0 \\
		CA^{d+3}B & CA^{d+2}B & 0 & \hdots & 0 \\
		\vdots & \vdots & \vdots & \vdots & \vdots\\
		CA^{2d+1}B &  CA^{2d}B & 0 & \hdots & 0 \\
		\vdots & \vdots & \vdots & \vdots & \vdots\\
		CA^{T+d-1}B & CA^{T+d}B & CA^{T+d-1}B & \hdots & CA^{d+1}B\\
		CA^{T+d+2}B & CA^{T+d+1}B & CA^{T+d}B & \hdots & CA^{d+2}B\\
		\vdots & \vdots & \vdots & \vdots & \vdots\\
		CA^{T+2d-1}B & CA^{T+2d-1}B & CA^{T+2d-2}B & \hdots & CA^{2d}B\\
		\end{bmatrix}}_{=S}\begin{bmatrix}
	u_1 \\
	u_2 \\
	\vdots \\
	u_T
	\end{bmatrix}
	\end{align*}
	It is clear that $||\Vc||_2, ||u||_2 = 1$ and for any matrix $S$, $||S||$ does not change if we interchange rows of $S$. Then we have 
	\begin{align*}
	|| S ||_2 &= \sigma\p*{\begin{bmatrix}
	CA^{d+1}B & 0 & 0 & \hdots & 0\\
	CA^{d+2}B & CA^{d+1}B & 0 & \hdots & 0 \\
	\vdots & \vdots & \vdots & \vdots & \vdots\\
	CA^{T+d+1}B & CA^{T+d}B & CA^{T+d-1}B & \hdots & CA^{d+1}B\\
	CA^{d+2}B & 0 & 0 & \hdots & 0 \\
	CA^{d+3}B & CA^{d+2}B & 0 & \hdots & 0 \\
	\vdots & \vdots & \vdots & \vdots & \vdots\\
	CA^{T+d+2}B & CA^{T+d+1}B & CA^{T+d}B & \hdots & CA^{d+2}B\\
	\vdots & \vdots & \vdots & \vdots & \vdots\\
	CA^{2d}B & 0 & 0 & \hdots & 0 \\
	CA^{2d+1}B & CA^{2d}B & 0 & \hdots & 0 \\
	\vdots & \vdots & \vdots & \vdots & \vdots\\
	CA^{T+2d-1}B & CA^{T+2d-1}B & CA^{T+2d-2}B & \hdots & CA^{2d}B\\
	\end{bmatrix}} \\
	&= \sigma \p*{\begin{bmatrix}
	\Tc_{d+1, T} \\
	\Tc_{d+2, T} \\
	\vdots \\
	\Tc_{2d, T}
	\end{bmatrix}} = \sqrt{ \sigma\Big(\sum_{k=1}^d \Tc_{d+k, T}^{\top}\Tc_{d+k, T}\Big)}
	\end{align*} 
\end{proof}	

\begin{prop}[Lemma 4.1~\cite{simchowitz2018learning}]
	\label{conditional_net}
	Let $S$ be an invertible matrix and $\kappa(S)$ be its condition number. Then for a $\frac{1}{4 \kappa}$--net of $\Sc^{d-1}$ and an arbitrary matrix $A$, we have
	\[
	||SA||_2 \leq 2 \sup_{v \in \Nc_{\frac{1}{4\kappa}}} \frac{||v^{\prime} A||_2}{||v^{\prime} S^{-1}||_2}
	\]	
\end{prop}
\begin{proof}
	For any vector $v \in \Nc_{\frac{1}{4 \kappa}}$ and $w$ be such that $||SA||_2 = \frac{||w^{\prime} A||_2}{||w^{\prime} S^{-1}||_2}$ we have 
	\begin{align*}
	||SA||_2 - \frac{||v^{\prime} A||_2}{||v^{\prime} S^{-1}||_2} &\leq \bl\frac{||w^{\prime} A||_2}{||w^{\prime} S^{-1}||_2} - \frac{||v^{\prime} A||_2}{||v^{\prime} S^{-1}||_2}\bl \\
	&= \bl \frac{||w^{\prime} A||_2}{||w^{\prime} S^{-1}||_2} - \frac{||v^{\prime} A||_2}{||w^{\prime} S^{-1}||_2} + \frac{||v^{\prime} A||_2}{||w^{\prime} S^{-1}||_2} - \frac{||v^{\prime} A||_2}{||v^{\prime} S^{-1}||_2} \bl \\
	&\leq ||SA||_2 \frac{\frac{1}{4 \kappa}||S^{-1}||_2}{||w^{\prime}S^{-1}||_2} + ||SA||_2\bl \frac{||v^{\prime}S^{-1}||_2}{||w^{\prime}S^{-1}||_2} - 1 \bl \\
	&\leq \frac{||SA||_2}{2}
	\end{align*}
\end{proof}
	\section{Control and Systems Theory Preliminaries}
\label{control-sys}
\subsection{Sylvester Matrix Equation}
\label{sylvester_section}
Define the discrete time Sylvester operator $S_{A, B}: \Rb^{n \times n} \rightarrow \Rb^{n \times n}$
\begin{equation}
    \label{sylvester_operator}
\Lc_{A, B}(X) = X - AXB 
\end{equation}

Then we have the following properties for $\Lc_{A, B}(\cdot)$. 

\begin{prop}
\label{sylvester_prop}
Let $\lambda_i, \mu_i$ be the eigenvalues of $A, B$ then $\Lc_{A, B}$ is invertible if and only if for all $i, j$
\[
\lambda_i \mu_j \neq 1
\]
\end{prop}

Define the discrete time Lyapunov operator for a matrix $A$ as $\Lc_{A, A^{\prime}}(\cdot) = S^{-1}_{A, A^{\prime}}(\cdot)$. Clearly it follows from Proposition~\ref{sylvester_prop} that whenever $\lambda_{\max}(A) < 1$ we have that the $S_{A, A^{\prime}}(\cdot)$ is an invertible operator. 

Now let $Q \succeq 0$ then
\begin{align}
    S_{A, A^{\prime}}(Q) &= X \nonumber \\
    \implies X &= A XA^{\prime} + Q \nonumber \\
    \implies X &= \sum_{k=0}^{\infty} A^k Q A^{\prime k} \label{x_sol}
\end{align}
Eq.~\eqref{x_sol} follows directly by substitution and by Proposition~\ref{sylvester_prop} is unique if $\rho(A) < 1$. Further, let $Q_1 \succeq Q_2 \succeq 0$ and $X_1, X_2$ be the corresponding solutions to the Lyapunov operator then from Eq.~\eqref{x_sol} that 
\begin{align*}
X_1, X_2 &\succeq 0 \\
X_1 &\succeq X_2
\end{align*}
\subsection{Properties of System Hankel matrix}
\label{control_hankel}
\begin{itemize}
    \item \textbf{Rank of system Hankel matrix}: For $M=(C, A, B) \in \Mc_n$, the system Hankel matrix, $\Hc_{0, \infty, \infty}(M)$, can be decomposed as follows:
\begin{equation}
\label{svd_hankel}
    \Hc_{0, \infty, \infty} (M)= 	\underbrace{\begin{bmatrix}
    C \\
    CA \\
    \vdots \\
    CA^{d} \\
    \vdots
    \end{bmatrix} }_{=\Oc} \underbrace{\begin{bmatrix}
    B & AB & \hdots & A^d B & \hdots
    \end{bmatrix}}_{=\Rc}
\end{equation}
It follows from definition that $\text{rank}(\Oc), \text{rank}(\Rc) \leq n$ and as a result $\text{rank}(\Oc\Rc) \leq n$. The system Hankel matrix rank, or $\text{rank}(\Oc\Rc)$, which is also the model order(or simply order), captures the complexity of $M$. If $\text{SVD}(\Hc_{0, \infty, \infty}) = U \Sigma V^{\top}$, then $\Oc = U\Sigma^{1/2} S, \Rc = S^{-1} \Sigma^{1/2}V^{\top}$. By noting that 
\[
CA^{l}S = CS (S^{-1} A S)^{l}, S^{-1} A^l B =  (S^{-1} A S)^{l} S^{-1}B
\]
we have obtained a way of recovering the system parameters (up to similarity transformations). Furthermore, $\Hc_{0, \infty, \infty}$ uniquely (up to similarity transformation) recovers $(C, A, B)$.

\item \textbf{Mapping Past to Future}: $\Hc_{0, \infty, \infty}$ can also be viewed as an operator that maps ``past'' inputs to ``future'' outputs. In Eq.~\eqref{dt_lti} assume that $\{\eta_t, w_t\} = 0$. Then consider the following class of inputs $U_t$ such that $U_t = 0$ for all $t \geq T$ but $U_t$ may not be zero for $t < T$. Here $T$ is chosen arbitrarily. Then 
\begin{equation}
    \label{past-future}
\underbrace{\begin{bmatrix}
Y_{T} \\
Y_{T+1} \\
Y_{T+2} \\
\vdots
\end{bmatrix}}_{\text{Future}} = \Hc_{0, \infty, \infty } \underbrace{\begin{bmatrix}
U_{T-1} \\
U_{T-2} \\
U_{T-3} \\
\vdots
\end{bmatrix}}_{\text{Past}}
\end{equation}
\end{itemize}
\subsection{Model Reduction}
\label{model_reduction}
Given an LTI system $M = (C, A, B)$ of order $n$ with its doubly infinite system Hankel matrix as $\Hc_{0, \infty, \infty}$. We are interested in finding the best $k$ order lower dimensional approximation of $M$, \textit{i.e.}, for every $k < n$ we would like to find $M_k$ of model order $k$ such that $||M-M_k||_{\infty}$ is minimized. Systems theory gives us a class of model approximations, known as balanced truncated approximations, that provide strong theoretical guarantees (See~\cite{glover1984all} and Section 21.6 in~\cite{zhou1996robust}). We summarize some of the basics of model reduction below. Assume that $M$ has distinct Hankel singular values.

Recall that a model $M=(C, A, B)$ is equivalent to $\tilde{M}=(CS, S^{-1}AS, S^{-1}B)$ with respect to its transfer function. Define 
\begin{align*}
    Q &= A^{\top}QA + C^{\top}C \\
    P &= A P A^{\top} + BB^{\top}
\end{align*}

For two positive definite matrices $P, Q$ it is a known fact that there exist a transformation $S$ such that $S^{\top}QS = S^{-1}PS^{-1 \top} = \Sigma$ where $\Sigma$ is diagonal and the diagonal elements are decreasing. Further, $\sigma_i$ is the $i^{th}$ singular value of $\Hc_{0, \infty, \infty}$. Then let $\tilde{A} = S^{-1}AS, \tilde{C} = CS, \tilde{B} = S^{-1}B$. Clearly $\tM = (\tA, \tB, \tC)$ is equivalent to $M$ and we have 
\begin{align}
    \Sigma &= \tA^{\top}\Sigma \tA + \tC^{\top}\tC \nonumber\\
    \Sigma &= \tA \Sigma \tA^{\top} + \tB \tB^{\top} \label{balanced}
\end{align}
Here $\tC, \tA, \tB$  is a balanced realization of $M$. 
\begin{prop}
    \label{balanced_realization}
    Let $\Hc_{0, \infty, \infty} = U \Sigma V^{\top}$. Here $\Sigma \succeq 0 \in \Rb^{n \times n}$. Then 
    \begin{align*}
        \tC &= [U\Sigma^{1/2}]_{1:p, :} \\
        \tA &= \Sigma^{-1/2}U^{\top} [U\Sigma^{1/2}]_{p+1:, :} \\
        \tB &= [\Sigma^{1/2}V^{\top}]_{:, 1:m}
    \end{align*}
The triple $(\tC, \tA, \tB)$ is a balanced realization of $M$. For any matrix $L$, $L_{:, m:n}$ (or $L_{m:n, :}$) denotes the submatrix with only columns (or rows) $m$ through $n$.
\end{prop}
\begin{proof}
Let the SVD of $\Hc_{0, \infty, \infty} = U \Sigma V^{\top}$. Then $M$ can constructed as follows: $U\Sigma^{1/2}, \Sigma^{1/2}V^{\top}$ are of the form 
\begin{align*}
    U \Sigma^{1/2} = \begin{bmatrix}
    C S \\
    CA S\\
    CA^2 S \\
    \vdots \\
    \end{bmatrix},  \Sigma^{1/2}V^{\top} = \begin{bmatrix}
    S^{-1} B & S^{-1} AB &  S^{-1} A^2 B \hdots
    \end{bmatrix}
\end{align*}
where $S$ is the transformation which gives us Eq.~\eqref{balanced}. This follows because
\begin{align*}
 \Sigma^{1/2} U^{\top} U \Sigma^{1/2} &= \sum_{k=0}^{\infty}S^{\top}A^{k \top}C^{\top} C A^k S \\
 &= \sum_{k=0}^{\infty}S^{\top}A^{k \top}S^{-1 \top} S^{\top}C^{\top} CS S^{-1} A^k S \\
 &= \sum_{k=0}^{\infty} \tA^{k \top} \tC^{\top} \tC \tA^k = \tA^{\top} \Sigma \tA + \tC^{\top} \tC = \Sigma
\end{align*}
Then $\tC = U\Sigma^{1/2}_{1:p, :}$ and 
\begin{align*}
    U\Sigma^{1/2} \tilde{A} &= [U\Sigma^{1/2}]_{p+1:, :} \\
    \tilde{A} &= \Sigma^{-1/2}U^{\top}[U\Sigma^{1/2}]_{p+1:, :} 
\end{align*}
We do a similar computation for $B$. 
\end{proof}
It should be noted that a balanced realization $\tC, \tA, \tB$ is unique except when there are some Hankel singular values that are equal. To see this, assume that we have 
$$\sigma_1 > \ldots > \sigma_{r-1} > \sigma_r = \sigma_{r+1} = \ldots = \sigma_s >  \sigma_{s+1} > \ldots \sigma_n$$
where $s-r >0$. For any unitary matrix $Q \in \Rb^{(s-r+1) \times (s-r+1)}$, define $Q_0$
\begin{equation}
Q_0 = \begin{bmatrix}
I_{(r-1) \times (r-1)} & 0 & 0\\
0 & Q & 0 \\
0 & 0 & I_{(n-s) \times (n-s)}
\end{bmatrix} \label{transform_Q}
\end{equation}
Then every triple $(\tC Q_0, Q_0^{\top} \tA Q_0, Q_0^{\top} \tB)$ satisfies Eq.~\eqref{balanced} and is a balanced realization. Let $M_k = (\tC_k, \tA_{kk}, \tB_k)$ where 
\begin{align}
   \tA =  \begin{bmatrix} 
   \tA_{kk} & \tA_{0k} \\
   \tA_{k0} & \tA_{00}
    \end{bmatrix}, \tB = \begin{bmatrix} \tB_{k} \\ \tB_{0} \end{bmatrix}, \tC = \begin{bmatrix} \tC_{k} & \tC_{0} \label{k_balanced}
    \end{bmatrix}
\end{align}
Here $\tA_{kk}$ is the $k \times k$ submatrix and corresponding partitions of $\tB, \tC$. The realization $M_k = (\tC_k, \tA_{kk}, \tB_k)$ is the $k$--order balanced truncated model. Clearly $M \equiv M_n$ which gives us $\tC =\tC_{nn}, \tA=\tA_{nn}, \tB=\tB_{nn}$, \textit{i.e.}, the balanced version of the true model. We will show that for the balanced truncation model we only need to care about the top $k$ singular vectors and not the entire model. 
\begin{prop}
\label{bt_model}
For the $k$ order balanced truncated model $M_k$, we only need top $k$ singular values and singular vectors of $\Hc_{0, \infty, \infty}$.
\end{prop}
\begin{proof}
From the preceding discussion in Proposition~\ref{balanced_realization} and Eq.~\eqref{k_balanced} it is clear that the first $p \times k$ block submatrix of $U\Sigma^{1/2}$ (corresponding to the top $k$ singular vectors) gives us $\tC_k$. Since
$$\tA = \Sigma^{-1/2}U^{\top}[U\Sigma^{1/2}]_{p+1:, :}$$ 
we observe that $\tA_{kk}$ depend only on the top $k$ singular vectors $U_k$ and corresponding singular values. This can be seen as follows: $[U\Sigma^{1/2}]_{p+1:, :}$ denotes the submatrix of $U\Sigma^{1/2}$ with top $p$ rows removed. Now in $U\Sigma^{1/2}$ each column of $U$ is scaled by the corresponding singular value. Then the $\tA_{kk}$ submatrix depends only on top $k$ rows of $\Sigma^{-1/2}U^{\top}$ and the top $k$ columns of $[U\Sigma^{1/2}]_{p+1:, :}$ which correspond to the top $k$ singular vectors. 
\end{proof}

	\section{Isometry of Input Matrix: Proof of Lemma~\ref{energy_conc_main}}
\label{sec:isometry}
\begin{theorem}
	\label{thm:isometry}
	Define
	\begin{align*}
	U &\coloneqq \begin{bmatrix}
	U_d & U_{d+1} & \hdots & U_{T+d-1} \\
	U_{d-1} & U_d & \hdots & U_{T+d-2} \\
	\vdots & \vdots & \ddots & \vdots \\
	U_1 & U_2 & \hdots  & U_T
	\end{bmatrix}
	\end{align*}
	where each $U_i \sim \subg(1)$ and isotropic. Then there exists an absolute constant $c$ such that $U$ satisfies:
	\[
	(1/2)T \leq \sigma_{\min}(U U^{\top}) \leq \sigma_{\max}(U U^{\top}) \leq (3/2) T
	\]
	whenever $T \geq cm^2 d (\log^2{(d)}\log^2{(m^2/\delta)} + \log^3{(2d)})$ with probability at least $1-\delta$.
\end{theorem}
\begin{proof}
	
	Define 
	\begin{align*}
	A_{md \times md} &\coloneqq \begin{bmatrix}
	0 & 0 & 0 & \hdots & 0 \\
	I & 0 & 0 & \hdots & 0 \\
	\vdots & \ddots & \ddots & \vdots & \vdots \\
	0 & \hdots & I & 0 & 0 \\
	0 & \hdots & 0 & I & 0 
	\end{bmatrix}, B_{md \times m} \coloneqq \begin{bmatrix}
	I \\
	0 \\
	\vdots \\
	0
	\end{bmatrix}, \wh{U}_k \coloneqq U_{d+k}
	\end{align*}
	Since
	\[
	U = \begin{bmatrix}
	U_{d} & U_{d+1} & \hdots & U_{T+d-1} \\
	U_{d-1} & U_{d} & \hdots & U_{T+d-2} \\
	\vdots & \vdots & \hdots & \vdots \\
	U_1 & U_2 & \hdots & U_{T}
	\end{bmatrix}
	\]
	we can reformulate it so that each column is the output of an LTI system in the following sense: 
	\begin{align}
	x_{k+1} &= Ax_{k} + B \wh{U}(k+1) \label{dynamical_energy}
	\end{align}
	where $UU^{\top} = \sum_{k=0}^{T-1} x_k x_k^{\top}$ and $x_0 = \begin{bmatrix}
	U_d \\
	U_{d-1} \\
	\vdots \\
	U_1
	\end{bmatrix}$. From Theorem~\ref{539_versh} we have that 
	\[
	\frac{3}{4}T I \preceq \sum_{k=0}^{T-1}\wh{U}_k \wh{U}_k^{\top} \preceq  \frac{5}{4}T I
	\]
	with probability at least $1-\delta$ whenever $	T \geq c\Big(m + \log{\frac{2}{\delta}}\Big)$. Define $V_t = \sum_{l=0}^{t-1} x_k x_k^{\top}$ then,
	\begin{align}
	V_T &= A V_{T-1} A^{\top} + B \p*{\sum_{k=0}^{T-1} \wh{U}_k \wh{U}_k^{\top}} B^{\top} + \sum_{k=0}^{T-2}\p*{A x_{k} \wh{U}^{\top}_{k+1} B^{\top} + B \wh{U}_{k+1} x^{\top}_{k} A^{\top}} \label{dynamical_expansion}
	\end{align}
	It can be easily checked that $x_k = \begin{bmatrix}
	U_{d+k} \\
	U_{d+k-1} \\
	\vdots \\
	U_{k+1}
	\end{bmatrix}$ and consequently 
	\begin{align*}
	\sum_{k=0}^{T-2}A x_{k} \wh{U}^{\top}_{k+1} B^{\top} = \sum_{k=0}^{T-2} \begin{bmatrix} 
	0 & 0  & \hdots & 0 & 0 \\
	U_{d+k} U_{d+k+1}^{\top} & 0  &\hdots & 0 & 0\\
	U_{d+k-1} U_{d+k+1}^{\top} & 0  & \hdots & 0 & 0\\
	\vdots &  \vdots & \ddots & \vdots & \vdots\\
	\vdots &  \vdots & \vdots & \ddots & \vdots\\
	U_{k+2} U_{d+k+1}^{\top} & 0  &\hdots & 0 & 0
	\end{bmatrix}.
	\end{align*} 
	Define $L_j \coloneqq \sum_{k=0}^{T-2}  U_{d+k-j+1} U_{d+k+1}^{\top}$ and $L_j$ is a $m \times m$ block matrix. Then 
	\begin{align*}
	T_d = \sum_{l=0}^{d-1} A^l \p*{\sum_{k=0}^{T-2}A x_{k} \wh{U}^{\top}_{k+1} B^{\top}}A^{l \top} = \begin{bmatrix} 
	0 & 0  & \hdots & 0 & 0 & 0 \\
	L_1 & 0  &\hdots & 0 & 0 & 0\\
	L_2 & L_1  & \hdots & 0 & 0 & 0\\
	\vdots &  \vdots & \ddots & \vdots & \vdots & \vdots\\
	\vdots &  \vdots & \vdots & \ddots & \vdots & \vdots\\
	L_{d-1} & 0  & \hdots & 0 & L_1 & 0
	\end{bmatrix}.
	\end{align*} 	
	Use Lemma~\ref{lem:cross_term} to show that 
	\begin{equation}
	\label{eq:1}
	\nrm*{T_d} \leq cm\sqrt{Td}\log{(d)} \log{(m^2/\delta)}
	\end{equation}
	with probability at least $1-\delta$. Then 
	\begin{align*}
	V_T = \sum_{l=0}^{d-1} A^l B \p*{\sum_{k=0}^{T-1} \wh{U}_k \wh{U}_k^{\top}} B^{\top} A^{l \top} + T_d - \sum_{l=0}^{d-1} A^l x_{T-1} x_{T-1}^{\top} A^{l \top}. 
	\end{align*}
	From Theorem~\ref{539_versh} we have with probability atleast $1-\delta$ that 
	\begin{equation}
	(3/4)TI \preceq \sum_{l=0}^{d-1} A^l B \Bigg(\sum_{k=0}^{T-1} \wh{U}_k \wh{U}_k^{\top}\Bigg) B^{\top} A^{l \top} \preceq (5/4)TI \label{eq:2}
	\end{equation}
	whenever $T \geq c\Big(m + \log{\frac{2}{\delta}}\Big)$. Observe that 
	\begin{align*}
	\nrm*{\sum_{l=1}^{d} A^{l} x_{T-1} x_{T-1}^{\top} A^{l \top}} &= \sigma_1^2([Ax_{T-1}, A^2 x_{T-1}, \hdots, A^d x_{T-1}])
	\end{align*}
	The matrix $[Ax_{T-1}, A^2 x_{T-1}, \hdots, A^d x_{T-1}]$ is the lower triangular submatrix of a random Toeplitz matrix with i.i.d $\subg(1)$ entries as in Theorem~\ref{toep_norm}. Then using Theorem~\ref{toep_norm} and Proposition~\ref{prop:lower_tri} we get that with probability at least $1-\delta$ we have 
	\begin{equation}
	\nrm*{[Ax_{T-1}, A^2 x_{T-1}, \hdots, A^d x_{T-1}]} \leq c m(\sqrt{d \log{(2d)}} \log{(2d)} + \sqrt{d\log{(1/\delta)}}). \label{eq:3}
	\end{equation}
	Then $\nrm*{\sum_{l=1}^{d} A^{l} x_{T-1} x_{T-1}^{\top} A^{l \top}} \leq cm^2d(\log^{3}{(2d)} +  \log{(1/\delta)} + \log{(2d)} \sqrt{\log{(2d)} \log{(1/\delta)}})$ with probability at least $1-\delta$. By ensuring that Eqs.~\eqref{eq:1}, \eqref{eq:2} and \eqref{eq:3} hold simultaneously we can ensure that $cm\sqrt{Td}\log{(d)} \log{(m^2/\delta)} \leq T/8$ and $cm^2d(\log^{3}{(2d)} +  \log{(1/\delta)} + \log{(2d)} \sqrt{\log{(2d)} \log{(1/\delta)}}) \leq T/8$ for large enough $T$ and absolute constant $c$.
\end{proof}
\begin{lem}
	\label{lem:cross_term}
	Let $\{U_j \in \Rb^{m \times 1}\}_{j=1}^{T+d}$ be independent $\subg(1)$ random vectors. Define $L_j \coloneqq \sum_{k=0}^{T-2}U_{d+k-j+1} U_{d+k+1}^{\top}$ for all $j \geq 1$ and 
	\begin{align*}
	T_d \coloneqq \begin{bmatrix} 
	0 & 0  & \hdots & 0 & 0 & 0 \\
	L_1 & 0  &\hdots & 0 & 0 & 0\\
	L_2 & L_1  & \hdots & 0 & 0 & 0\\
	\vdots &  \vdots & \ddots & \vdots & \vdots & \vdots\\
	\vdots &  \vdots & \vdots & \ddots & \vdots & \vdots\\
	L_{d-1} & 0  & \hdots & 0 & L_1 & 0
	\end{bmatrix}.
	\end{align*}  
	Then with probability at least $1-\delta$ we have 
	\[
	\nrm*{T_d} \leq cm\sqrt{Td}\log{(d)} \log{(m/\delta)}.
	\]
\end{lem}
\begin{proof}
	Since $L_j$s are block matrices, the techniques in~\cite{meckes2007spectral} cannot be directly applied. However, by noting that $E$ can be broken into a sum of $m$ matrices where the norm of each matrix can be bounded by a Toeplitz matrix we can use the result from~\cite{meckes2007spectral}. For instance if $m=2$ and $\{u_i\}_{i=1}^{\infty}$ are independent $\subg(1)$ random variables then we have
	\begin{align*}
	T_d = \begin{bmatrix}
	\begin{bmatrix}
	0 & 0 \\
	0 & 0
	\end{bmatrix} & 	   \begin{bmatrix}
	0 & 0 \\
	0 & 0
	\end{bmatrix} & \hdots \\
	\begin{bmatrix}
	u_1 & u_2 \\
	u_3 & u_4
	\end{bmatrix} & 	   \begin{bmatrix}
	0 & 0 \\
	0 & 0
	\end{bmatrix} & \hdots \\
	\begin{bmatrix}
	u_5 & u_6 \\
	u_7 & u_8
	\end{bmatrix} & \begin{bmatrix}
	u_1 & u_2 \\
	u_3 & u_4
	\end{bmatrix} & \hdots \\
	\vdots & \vdots & \ddots
	\end{bmatrix}.
	\end{align*}
	Now, 
	\begin{align*}
	T_d = \underbrace{\begin{bmatrix}
		\begin{bmatrix}
		0 & 0 \\
		0 & 0
		\end{bmatrix} & 	   \begin{bmatrix}
		0 & 0 \\
		0 & 0
		\end{bmatrix} & \hdots \\
		\begin{bmatrix}
		u_1 & 0 \\
		u_3 & 0
		\end{bmatrix} & 	   \begin{bmatrix}
		0 & 0 \\
		0 & 0
		\end{bmatrix} & \hdots \\
		\begin{bmatrix}
		u_5 & 0 \\
		u_7 & 0
		\end{bmatrix} & \begin{bmatrix}
		u_1 & 0 \\
		u_3 & 0
		\end{bmatrix} & \hdots \\
		\vdots & \vdots & \ddots
		\end{bmatrix}}_{=M_1} + 	\underbrace{\begin{bmatrix}
		\begin{bmatrix}
		0 & 0 \\
		0 & 0
		\end{bmatrix} & 	   \begin{bmatrix}
		0 & 0 \\
		0 & 0
		\end{bmatrix} & \hdots \\
		\begin{bmatrix}
		0 & u_2 \\
		0 & u_4
		\end{bmatrix} & 	   \begin{bmatrix}
		0 & 0 \\
		0 & 0
		\end{bmatrix} & \hdots \\
		\begin{bmatrix}
		0 & u_6 \\
		0 & u_8
		\end{bmatrix} & \begin{bmatrix}
		0 & u_2 \\
		0 & u_4
		\end{bmatrix} & \hdots \\
		\vdots & \vdots & \ddots
		\end{bmatrix}}_{=M_2},
	\end{align*}	
	then $||T_d|| \leq \sup_{ 1 \leq i \leq 2}||M_i||$. Furthermore for each $M_i$ we have 
	\begin{align*}
	M_1 = \underbrace{\begin{bmatrix}
		\begin{bmatrix}
		0 & 0 \\
		0 & 0
		\end{bmatrix} & 	   \begin{bmatrix}
		0 & 0 \\
		0 & 0
		\end{bmatrix} & \hdots \\
		\begin{bmatrix}
		u_1 & 0 \\
		0 & 0
		\end{bmatrix} & 	   \begin{bmatrix}
		0 & 0 \\
		0 & 0
		\end{bmatrix} & \hdots \\
		\begin{bmatrix}
		u_5 & 0 \\
		0 & 0
		\end{bmatrix} & \begin{bmatrix}
		u_1 & 0 \\
		0 & 0
		\end{bmatrix} & \hdots \\
		\vdots & \vdots & \ddots
		\end{bmatrix}}_{=M_{11}} + 	\underbrace{\begin{bmatrix}
		\begin{bmatrix}
		0 & 0 \\
		0 & 0
		\end{bmatrix} & 	   \begin{bmatrix}
		0 & 0 \\
		0 & 0
		\end{bmatrix} & \hdots \\
		\begin{bmatrix}
		0 & 0 \\
		u_3 & 0
		\end{bmatrix} & 	   \begin{bmatrix}
		0 & 0 \\
		0 & 0
		\end{bmatrix} & \hdots \\
		\begin{bmatrix}
		0 & 0 \\
		u_7 & 0
		\end{bmatrix} & \begin{bmatrix}
		0 & 0 \\
		u_3 & 0
		\end{bmatrix} & \hdots \\
		\vdots & \vdots & \ddots
		\end{bmatrix}}_{=M_{12}},
	\end{align*}
	and $||M_1|| \leq ||M_{11}|| + ||M_{12}||$. The key idea is to show that $M_{i1}$ are Toeplitz matrices (after removing the zeros in the blocks) and we can use the standard techniques described in proof of Theorem 1 in~\cite{meckes2007spectral}. Then we will show that each $||M_{ij}|| \leq C$ with high probability and $||T_d|| \leq m C$. 
	
	For brevity, we will assume for now that $U_i$ are scalars and at the end we will scale by $m$. By standard techniques described in proof of Theorem 1 in~\cite{meckes2007spectral}, we have that the finite Toeplitz matrix $T_d + T_d^{\top}$ is $d \times d$ submatrix of the infinite Laurent matrix
	\[
	M = [L_{|j-k|}\textbf{1}_{|j-k| < d-1}]_{j,k \in \Zb}.
	\]
	Consider $M$ as an operator on $\ell^2(\Zb)$ in the canonical way, and let $\psi: \ell^2(\Zb) \rightarrow L^2[0, 1]$ denote the usual linear trigonometric isometry $\psi(e_j)(x) = e^{2\pi ij x}$. Then $\psi M_d \psi^{-1}: L^2 \rightarrow L^2$ is the operator correpsonding to 
	\[
	f(x) = \sum_{j=-(d-1)}^{d-1} L_{|j|} e^{2\pi i j x} = L_0 + 2 \sum_{j=1}^{d-1}\cos{(2 \pi j x)} L_j
	\]
	Therefore,
	\[
	\nrm*{T_d + T_d^{\top}} \leq \nrm*{M} = \nrm*{f}_{\infty} = \sup_{0 \leq x \leq 1} |Y_x|
	\]
	where $Y_x = 2 \sum_{j=1}^{d-1} \cos{(2 \pi j x)} L_j$. Furthermore note that $Y_x$ has the following form
	\begin{equation}
	\label{eq:yx_form}
	Y_x = U^{\top} \underbrace{\begin{bmatrix} 
		0 & c^{x}_1  & c^{x}_2 & \hdots & c^{x}_{d-1} & 0 & \hdots & 0 \\
		0 & 0  & c^{x}_1 &\hdots & c^{x}_{d-1} & 0 & \hdots & 0\\
		\vdots & \vdots  & \hdots & \ddots & \hdots & \ddots & \vdots & \vdots\\
		\vdots &  \vdots & \vdots & \vdots& \ddots &\vdots & \ddots & \vdots \\
		0 &  0 & \hdots & 0& 0&c^{x}_1 & \hdots & c^{x}_{d-1} \\
		0 &  0 & \hdots & 0 & 0 & 0 & \hdots & 0 \\
		\vdots &  \vdots & \vdots & \vdots & \vdots & \vdots & \vdots & \vdots \\
		\end{bmatrix}}_{=C_x} U.
	\end{equation}
	Here $U = \begin{bmatrix}
	U_1 \\
	U_2 \\
	\vdots \\
	U_{T+d}
	\end{bmatrix}$ and $c^x_j = 2 \cos{(2\pi j x)}$.
	For any $x$ and assuming $U_j \sim \subg(1)$, we have from Theorem~\ref{thm:hanson_wright}
	\begin{equation}
	\label{eq:ys_dist}
	\Pb\Big(\abs{Y_x/\sqrt{Td}} \leq t\Big) \leq 2 \exp{-c(t \wedge t^2)} 
	\end{equation}
	The tail behavior of $Y_x/\sqrt{Td}$ is not strictly subgaussian and we need to use Theorem~\ref{thm:emp_process}. The function $\psi$ can be found as Eq. 1 of~\cite{van2013bernstein} (equivalent upto universal constants) with $L=2$ and its inverse being
	\[
	\psi^{-1}(t) = \sqrt{\log{(1+t)}} + \log{(1+t)}.
	\]
	We have that 
	\[
	\nrm*{\sup_t |Y_t|}_{\psi} \leq \nrm*{Y_0}_{\psi} + \sqrt{Td} \int_0^{1} \psi^{-1}(N(\epsilon/2, d)) d\epsilon,
	\]
	where $d(s, t) = \nrm*{(Y_s - Y_t)/\sqrt{Td}}_{\psi}$ and $N(\epsilon, d)$ is the minimal number of balls of radius $\epsilon$ needed to cover $[0, 1]$ where $d(\cdot, \cdot)$ is the pseudometric. Since $Y_s$ has distribution as in Eq.~\eqref{eq:ys_dist}, it follows that $d(s, t) \leq c |s-t|$ for some absolute constant $c$. Then
	\[
	\int_0^1 \psi^{-1}(N(\epsilon/2, d)) d\epsilon \leq c
	\]
	for some universal constant $c > 0$. This ensures that $\nrm*{\sup_t |Y_t|}_{\psi} \leq c\sqrt{Td}$. Since $\Ex[X] \leq \nrm*{X}_{\psi}$ we have that $\Ex[\sup_{0 \leq x \leq 1}|Y_x|] \leq \sqrt{Td}$. This implies $\Ex[\nrm*{T_d + T_d^{\top}}] \leq \sqrt{Td}$, and using Proposition~\ref{prop:lower_tri} we have $\Ex[\nrm*{T_d} ] \leq c\sqrt{Td} \log{(d)}$. Furthermore, we can make a stronger statement because $\nrm*{\sup_t |Y_t|}_{\psi} \leq c\sqrt{Td}$ which implies that 
	\[
	\nrm*{T_d} \leq c\sqrt{Td}\log{(d)}\log{(1/\delta)}
	\]
	with probability at least $1-\delta$. Then recalling that in the general case that $L_j$s of $T_d$ were $m \times m$ block matrices we scale by $m$ and get with probability at least $1-\delta$
	\[
	\nrm*{T_d} \leq cm\sqrt{Td}\log{(d)} \log{(m^2/\delta)}
	\]
	where the union is over all $m^2$ elements being less that $c\sqrt{Td}\log{(d)}\log{(m^2/\delta)}$.
\end{proof}
	\section{Error Analysis for Theorem~\ref{hankel_est}}
\label{appendix_error}
For this section we assume that $U_t \sim \subg(L^2)$.
\subsection{Proof of Theorem~\ref{hankel_convergence}}
\label{hankel_conv_proof}
Recall Eq.~\eqref{input-output-eq} and \eqref{compact-dynamics}, \textit{i.e.},
\begin{align}
\tilde{Y}^{+}_{l, d} &= \Hc_{0, d, d}\tilde{U}^{-}_{l-1, d} + \Tc_{0, d}\tilde{U}^{+}_{l, d} + \Hc_{d, d, l-d-1}\tilde{U}^{-}_{l-d-1, l-d-1} \nonumber\\
&+ \Oc_{0, d, d}\tilde{\eta}^{-}_{l-1, d} + \Tc\Oc_{0, d}\tilde{\eta}^{+}_{l, d} + \Oc_{d, d, l-d-1}\tilde{\eta}^{-}_{l-d-1, l-d-1} + \tilde{w}^{+}_{l, d}\label{compact-dynamics1}
\end{align}
Assume for now that we have $T+2d$ data points instead of $T$. It is clear that 
\[
\hHc_{0, d, d} = \arg \min_{\Hc } \sum_{l=0}^{T-1} ||\tilde{Y}^{+}_{l+d+1, d} - \Hc \tilde{U}^{-}_{l+d, d}||_2^2  = \p*{\sum_{l=0}^{T-1}\tilde{Y}^{+}_{l+d+1, d} \p*{\tilde{U}^{- }_{l+d, d}}^{\top} }V_T^{+}
\]
where
\begin{equation}
\label{sample_cov}
V_T = \sum_{l=0}^{T-1} \tilde{U}^{-}_{l+d, d} \tilde{U}^{- \prime}_{l+d, d},
\end{equation}
or 
\[
V_T = U U^{\prime}
\]
where 
\begin{align*}
U &\coloneqq \begin{bmatrix}
U_d & U_{d+1} & \hdots & U_{T+d-1} \\
U_{d-1} & U_d & \hdots & U_{T+d-2} \\
\vdots & \vdots & \ddots & \vdots \\
U_1 & U_2 & \hdots  & U_T
\end{bmatrix}.
\end{align*}
It is show in Theorem~\ref{thm:isometry} that $V_T$ is invertible with probability at least $1-\delta$. So in our analysis we can write this as
\begin{equation*}
\p*{\sum_{l=0}^{T-1} \tilde{U}^{-}_{l+d, d} \tilde{U}^{- \top}_{l+d, d}}^{+} = \p*{\sum_{l=0}^{T-1} \tilde{U}^{-}_{l+d, d} \tilde{U}^{- \top}_{l+d, d}}^{-1} 
\end{equation*}
%%%%%%%%%

From this one can conclude that 
\begin{align}
\bl \bl \hat{\Hc} - \Hc_{0, d, d} \bl \bl_2 &= \bl \bl\Big(\sum_{l=0}^{T-1} \tilde{U}^{-}_{l+d, d} \tilde{U}^{- \top}_{l+d, d}\Big)^{-1} \Big(\sum_{l=0}^{T-1} \tilde{U}^{-}_{l+d, d}  \tilde{U}^{+ \top}_{l+d+1, d} \Tc_{0, d}^{ \top} \nonumber \\
&+ \tilde{U}^{-}_{l+d, d} \tilde{U}^{- \top}_{l, l}\Hc_{d, d, l}^{\top} + \tilde{U}^{-}_{l+d, d}\tilde{\eta}^{- \top}_{l+d, d} \Oc_{0, d, d}^{\top} \nonumber \\
&+ \tilde{U}^{-}_{l+d, d} \tilde{\eta}^{+ \top}_{l+d+1, d}\Tc\Oc^{\top}_{0, d} + \tilde{U}^{-}_{l+d, d} \tilde{\eta}^{- \top}_{l, l} \Oc_{d, d, l}^{\top} + \tilde{U}^{-}_{l+d, d} \tilde{w}^{+ \top}_{l+d+1, d}\Big)\bl \bl_2 \label{diff_eq2}
\end{align}
\vspace{2mm}
Here as we can observe $\tilde{U}^{- \top}_{l, l}, \tilde{\eta}^{- \top}_{l, l}$ grow with $T$ in dimension. Based on this we divide our error terms in two parts:
\begin{align}
E_1 = \Big(\sum_{l=0}^{T-1} \tilde{U}^{-}_{l+d, d} \tilde{U}^{- \top}_{l+d, d}\Big)^{-1}\Bigg(\tilde{U}^{-}_{l+d, d} \tilde{U}^{- \top}_{l, l}\Hc_{d, d, l}^{\top}  + \tilde{U}^{-}_{l+d, d} \tilde{\eta}^{- \top}_{l, l} \Oc_{d, d, l}^{\top} \Bigg) \label{T_error}
\end{align}
and
\begin{align}
E_2 = \Big(\sum_{l=0}^{T-1} \tilde{U}^{-}_{l+d, d} \tilde{U}^{- \top}_{l+d, d}\Big)^{-1}\Bigg(\tilde{U}^{-}_{l+d, d} \tilde{\eta}^{+ \top}_{l+d+1, d}\Tc\Oc^{\top}_{0, d}  + \tilde{U}^{-}_{l+d, d}  \tilde{U}^{+ \top}_{l+d+1, d} \Tc_{0, d}^{\top} +  \label{d_error} \\
 \tilde{U}^{-}_{l+d, d} \tilde{\eta}^{+ \top}_{l+d+1, d}\Tc\Oc^{\top}_{0, d} + \tilde{U}^{-}_{l+d, d} \tilde{w}^{+ \top}_{l+d+1, d} \Bigg) \nonumber
\end{align}
% and
% \begin{equation*}
% \Big(\sum_{l=1}^T \tilde{U}^{-}_{l+d, d} \tilde{U}^{- \prime}_{l+d, d}\Big)^{\dagger} \Big(\sum_{l=1}^T \tilde{U}^{-}_{l+d, d}  \tilde{\eta}^{- \prime}_{l, l} O_{d, d, l}^{\prime} \Big)
% \end{equation*}
Then the proof of Theorem~\ref{hankel_est} will reduce to Propositions~\ref{error_prob1}--\ref{error_prob3}. We first analyze 
\[
\bl \bl V^{-1/2}_T \Big(\sum_{l=0}^{T-1} \tilde{U}^{-}_{l+d, d}  \tilde{U}^{- \top}_{l, l} \Hc_{d, d, l}^{\top} \Big) \bl \bl_2 
\]
The analysis of $||V^{-1/2}_T (\sum_{l=0}^{T-1} \tilde{U}^{-}_{l+d, d}  \tilde{\eta}^{- \top}_{l, l} O_{d, d, l}^{\top})||$ will be almost identical and will only differ in constants.
%%%%%
\begin{prop}
	\label{error_prob1}
	For $0 < \delta < 1$, we have with probability at least $1 - 2\delta$
	\begin{align*}
	\bl \bl V^{-1/2}_T \Big(\sum_{l=0}^{T-1} \tilde{U}^{-}_{l+d, d}  \tilde{U}^{- \top}_{l, l} \Hc_{d, d, l}^{\top} \Big) \bl \bl_2 \leq 4 \sigma \sqrt{\log{\frac{1}{\delta}} + pd + m}
	\end{align*}
where $\sigma = \sqrt{\sigma(\sum_{k=1}^d \Tc_{d+k, T}^{\top}\Tc_{d+k, T})}$. 
\end{prop}

\begin{proof}
	We proved that $\frac{TI}{2} \preceq V_T \preceq \frac{3TI}{2}$ with high probability, then  
	\begin{align}
	&\Pb\Big(\bl \bl V^{-1/2}_T\Big(\sum_{l=0}^{T-1} \tilde{U}^{-}_{l+d, d}  \tilde{U}^{- \prime}_{l, l} \Hc_{d, d, l}^{\prime} \Big) \bl \bl_2 \geq a, \frac{TI}{2} \preceq V_T \preceq \frac{3TI}{2} \Big) \nonumber \\
	&\leq  \Pb \Big(\bl \bl\sqrt{\frac{2}{T}}\Big(\sum_{l=0}^{T-1} \tilde{U}^{-}_{l+d, d}  \tilde{U}^{- \prime}_{l, l} \Hc_{d, d, l}^{\prime} \Big) \bl \bl_2 \geq a, \frac{TI}{2} \preceq V_T \preceq \frac{3TI}{2} \Big) \nonumber\\
	&\leq \Pb \Big(2 \sup_{ v \in \Nc_{\frac{1}{2}}} \bl \bl\sqrt{\frac{2}{T}}\Big(\sum_{l=0}^{T-1} \tilde{U}^{-}_{l+d, d}  \tilde{U}^{- \prime}_{l, l} \Hc_{d, d, l}^{\prime} v \Big) \bl \bl_2 \geq a\Big) + \Pb\Big(\frac{TI}{2} \preceq V_T \preceq \frac{3TI}{2} \Big) - 1 \nonumber \\
	&\leq 5^{pd} \Pb \Big(2  \bl \bl\sqrt{\frac{2}{T}}\Big(\sum_{l=0}^{T-1} \tilde{U}^{-}_{l+d, d}  \tilde{U}^{- \prime}_{l, l} \Hc_{d, d, l}^{\prime} v \Big) \bl \bl_2 \geq a\Big) - \delta. \label{normalized_err}
	\end{align}
	Define the following $\eta_{l, d} = \tilde{U}^{- \top}_{l, l} \Hc_{d, d, l}^{\top} v, X_{l, d}  = \sqrt{\frac{2}{T}} \tilde{U}^{-}_{l+d, d}$. Observe that $\eta_{l,d}, \eta_{l+1,d}$ have contributions from $U_{l-1}, U_{l-2}$ etc. and do not immediately satisfy the conditions of Theorem~\ref{selfnorm_main}. Instead we will use the fact that $X_{i, d}$ is independent of $U_j$ for all $j \leq i$.
	\begin{align*}
	\bl \bl V^{-1/2}_T\Big(\sum_{l=0}^{T-1} \tilde{U}^{-}_{l+d, d}  \tilde{U}^{- \prime}_{l, l} \Hc_{d, d, l}^{\prime} \Big) \bl \bl_2 &\leq 2 \sup_{ v \in \Nc_{\frac{1}{2}}} {||\sqrt{\frac{2}{T}} \sum_{l=0}^{T-1} \tilde{U}^{-}_{l+d, d}\tilde{U}^{- \prime}_{l, l} \Hc_{d, d, l}^{\prime} v||} \\
	&\leq 2 \sup_{ v \in \Nc_{\frac{1}{2}}} {||\sum_{l=0}^{T-1} X_{l, d}\eta_{l, d}||}.
	\end{align*}
	Define $\Hc_{d, d, l}^{\top} v = [\beta_1^{\top}, \beta_2^{\top}, \ldots, \beta_{l}^{\top}]^{\top}$. $\beta_i$ are $m \times 1$ vectors when LTI system is MIMO. Then $\eta_{l,d}= \sum_{k=0}^{l-1} U^{\top}_{l-k} \beta_{k+1}$. Let $\alpha_l = {X_{l,d}}$. Then consider the matrix 
	\begin{align*}
	\Bc_{T \times mT}= \begin{bmatrix}
	\beta_1^{\top} & 0 & 0& \ldots \\
	\beta^{\top}_2 & \beta^{\top}_1 & 0 & \ldots\\
	\vdots & \vdots & \ddots & \vdots \\
	\beta_T^{\top} & \beta_{T-1}^{\top}& \ldots &  \beta^{\top}_1
	\end{bmatrix}.
	\end{align*}
	Observe that the matrix $||\Bc_{T \times mT}||_2 = \sqrt{\sigma(\sum_{k=1}^d \Tc_{d+k, T}^{\top}\Tc_{d+k, T})} \leq \sqrt{d} ||\Tc_{d, \infty}||_2< \infty$ which follows from Lemma~\ref{bound_toeplitz}. Then 
	\begin{align*}
	 \sum_{l=0}^{T-1} X_{l, d}\eta_{l, d}   &= [\alpha_1, \hdots, \alpha_T] \Bc \begin{bmatrix}U_1 \\
	U_2 \\
	\vdots \\
	U_T \end{bmatrix} \\
	&= [\sum_{k=1}^{T} \alpha_k \beta^{\top}_{k}, \sum_{k=2}^{T} \alpha_{k} \beta^{\top}_{k-1}, \hdots,  \alpha_{T} \beta^{\top}_{1}]\begin{bmatrix}U_1 \\
	U_2 \\
	\vdots \\
	U_T \end{bmatrix} \\
	&= \sum_{j=1}^T \Big(\sum_{k=j}^T \alpha_k \beta^{\top}_{k} U_j\Big).
	\end{align*}
	Here $\alpha_i = X_{i, d}$ and recall that $X_{i, d}$ is independent of $U_j$ for all $i \geq j$. Let $\gamma^{\prime} = \alpha^{\prime} \Bc$. Define $\Gc_{T+d-k} = \tilde{\sigma}(\{U_{k+1}, U_{k+2}, \ldots, U_{T+d}\})$ where $\tilde{\sigma}(A)$ is the sigma algebra containing the set $A$ with $\Gc_0 = \phi$. Then $\Gc_{k-1} \subset \Gc_k$. Furthermore, since $\gamma_{j-1}, U_j$ are $\Gc_{T+d+1-j}$ measurable and $U_j$ is conditionally (on $\Gc_{T+d-j}$) subGaussian, we can use Theorem~\ref{selfnorm_main} on $\gamma^{\prime}U = \alpha^{\prime}\Bc U$ (where $\gamma_{j} = X_{T+d-j}, U_j = \eta_{T+d-j+1}$ in the notation of Theorem~\ref{selfnorm_main}). Then with probability at least $1-\delta$ we have 
\begin{equation}
    \label{selfnorm_ub}
    \Big| \Big| \Big(\alpha^{\prime}\Bc\Bc^{\prime} \alpha + V \Big)^{-1/2}\gamma^{\prime}U \Big| \Big| \leq L\sqrt{\Big(\log{\frac{1}{\delta}} + \log{\frac{\det(\alpha^{\prime}\Bc\Bc^{\prime} \alpha + V)}{\det(V)}}\Big)} .
\end{equation}	
For any fixed $V > 0$. With probability at least $1 - \delta$, we know from Theorem~\ref{thm:isometry} that $\alpha^{\prime} \alpha \preceq \frac{3I}{2} \implies \alpha^{\prime}\Bc\Bc^{\prime} \alpha \preceq  \frac{3\sigma_1^2(\Bc)I}{2}$. By combining this event and the event in Eq.~\eqref{selfnorm_ub} and setting $V = \frac{3\sigma_1^2(\Bc)I}{2}$, we get with probability at least $1-2\delta$ that 
\begin{align}
      ||\alpha^{\prime} \Bc U||_2=||\gamma^{\prime}U||_2 \leq \sqrt{3}\sigma_1(\Bc)L\sqrt{\Big(\log{\frac{1}{\delta}} + pd\log{3} + m\Big)}. \label{err_1_ub}
\end{align}
Replacing $\delta \rightarrow 5^{-pd}\frac{\delta}{2}$, we get from Eq.~\eqref{normalized_err} 
	\[
	\bl \bl V^{-1/2}_T \Big(\sum_{l=0}^{T-1} \tilde{U}^{-}_{l+d, d}  \tilde{U}^{- \prime}_{l, l} \Hc_{d, d, l}^{ \prime} \Big) \bl \bl_2 \leq \sqrt{6} \log{(5)} L \sigma_1(\Bc) \sqrt{\log{\frac{1}{\delta}} + pd + m}
	\]
with probability at least $1-\delta$. Since $L=1$ we get our desired result.
\end{proof}
Then similar to Proposition~\ref{error_prob1}, we analyze $\bl \bl V^{-1/2}_T \Big(\sum_{l=0}^{T-1} \tilde{U}^{-}_{l+d, d}  \tilde{U}^{+ \top}_{l+d+1, d} \Tc_{0, d}^{ \top} \Big) \bl \bl_2$
\begin{prop}
	\label{error_prob2}
	For $0 < \delta < 1$ and large enough $T$, we have with probability at least $1 - \delta$ 
	\begin{align*}
	\bl \bl V^{-1/2}_T \Big(\sum_{l=0}^{T-1} \tilde{U}^{-}_{l+d, d}  \tilde{U}^{+ \top}_{l+d+1, d} \Tc_{0, d}^{ \top} \Big) \bl \bl_2 \leq 4 \sigma \sqrt{\log{\frac{1}{\delta}} + pd + m}
	\end{align*}
	where 
	$$\sigma \leq \sup_{||v||_2 = 1}\bl \bl \begin{bmatrix}
    v^{\top} CA^{d}B & v^{\top} CA^{d-1}B & v^{\top} CA^{d-2}B & \hdots & v^{\top} CB & 0 \\
    0 & \ddots & \ddots & \ddots & \ddots &  0 \\
    0 & \ddots & \ddots & \ddots & \ddots & \ddots \\
    0 & v^{\top} CA^{d}B & v^{\top} CA^{d-1}B & \hdots & \hdots & v^{\top} CB
    \end{bmatrix} \bl \bl_2 \leq \sum_{j=0}^d ||CA^jB||_2 \leq \beta \sqrt{d}.$$ 
\end{prop}
\begin{proof}

Note $	\bl \bl V^{-1/2}_T \Big(\sum_{l=0}^{T-1} \tilde{U}^{-}_{l+d, d}  \tilde{U}^{+ \top}_{l+d+1, d} \Tc_{0, d}^{ \top} \Big) \bl \bl_2 \leq \bl \bl \sqrt{\frac{2}{T}} \Big(\sum_{l=0}^{T-1} \tilde{U}^{-}_{l+d, d}  \tilde{U}^{+ \top}_{l+d+1, d} \Tc_{0, d}^{ \top} \Big) \bl \bl_2$ with probability at least $1 - \delta$ for large enough $T$. Here $\Tc_{0, d}^{\top}$ is $md \times pd$ matrix. Then define $X_l = \sqrt{\frac{2}{T}} \tilde{U}^{-}_{l+d, d}$ and the vector $M_l \in \Rb^{pd}$ as $M_l^{\top} = \tilde{U}^{+ \top}_{l+d+1, d} \Tc_{0, d}^{ \top}$. Then 
\begin{align*}
\Pb(||\sum_{l=0}^{T-1} X_l M_{l}^{\top}||_2 \geq t) &\underbrace{\leq}_{\frac{1}{2}-\text{net}} 5^{pd}  \Pb(||\sum_{l=0}^{T-1} X_l M_l^{\top} v||_2 \geq t/2)
\end{align*}
where $M_{l}^{\top}v$ is a real value. Let $\beta \coloneqq \Tc_{0, d}^{ \top} v$, then  $M_{l}^{\top}v = \tilde{U}^{+ \top}_{l+d+1, d} \beta$. 
%
%
%%\begin{equation}
%%
%%% =[
%%%\underbrace{0}_{=M_{l1}} , \underbrace{U_{l+d+1}^{\top} B^{\top} C^{\top}}_{=M_{l2}} , \underbrace{U_{l+d+1}^{\top} B^{\top} A^{\top} C^{\top} + U_{l+d+2}^{\top} B^{\top} C^{\top}}_{=M_{l3}}, \hdots]
%%\end{equation}
%
%Now $\sum_{l=0}^{T-1} X_l M_l^{\top} = [\sum_{l=0}^{T-1} X_l M_{l1}, \sum_{l=0}^{T-1} X_l M_{l2}, \hdots]$. We will show that $||\sum_{l=0}^{T-1} X_l M_{l1}||_2 = O(1)$ and consequently $||\sum_{l=0}^{T-1} X_l M_l||_2 =O(\sqrt{d})$ with high probability. We will analyze $||\sum_{l=0}^{T-1} X_l M_{ld}||_2$ (the same analysis applies to all columns). Due to the structure of $X_l, M_l$ we have that $X_l$ is independent of $M_l$. 
%
This allows us to write $X_l M_{l}^{\top}v$ in a form that will enable us to apply Theorem~\ref{selfnorm_main}.
\begin{align}
    \sum_{l=0}^{T-1} X_l M_{l}^{\top} v &= \underbrace{[X_0, X_1, \hdots, X_{T-1}]}_{=X} \underbrace{\begin{bmatrix}
    \beta_1^{\top} & \beta_2^{\top} & \hdots & \beta_d^{\top} & \hdots & 0 \\
    0 & \beta_1^{\top} & \ddots & \ddots & \ddots & 0 \\
    0 & \ddots & \ddots & \ddots & \ddots & \ddots \\
    0 & \hdots & 0 & \beta_1^{\top} & \hdots & \beta_d^{\top}
    \end{bmatrix}}_{=\Ic}\underbrace{\begin{bmatrix}
    U_{d+1} \\
    U_{d+2} \\
    \vdots \\
    U_{T+2d}
    \end{bmatrix}}_{=N}
\end{align}
Here $\Ic$ is $\Rb^{T \times (mT + md)}$. It is known from Theorem~\ref{thm:isometry} that $XX^{\top} \preceq \frac{3I}{2}$ with high probability and consequently $X \Ic \Ic^{\top}X^{\top} \preceq \frac{3\sigma^2_1(\Ic)I}{2}$. Define $\bcF_{l} = \tilde{\sigma}(\{U_l\}_{j=1}^{d+l})$ as the sigma field generated by $(\{U_l\}_{j=1}^{d+l}$. Furthermore $N_l$ is $\bcF_{l}$ measurable, and $[X\Ic]_l$ is $\bcF_{l-1}$ measurable and we can apply Theorem~\ref{selfnorm_main}. Now the proof is similar to Proposition~\ref{error_prob1}. Following the same steps as before we get with probability at least $1-\delta$
\begin{align*}
    ||\sum_{l=0}^{T-1} X_l M_{l}^{\top}v||_2 = ||\sum_{l=0}^{T-1} [X \Ic]_l N_l ||_2 \leq \sqrt{3}  \sigma_1(\Ic) L \sqrt{\log{\frac{1}{\delta}} + pd\log{3} + m}
\end{align*}
and substituting $\delta \rightarrow 5^{-pd} \delta$ we get 
\[
||\sum_{l=0}^{T-1} X_l M_{l}^{\top}||_2 \leq   \sqrt{6} \log{(5)} \sigma_1(\Ic) L \sqrt{\log{\frac{1}{\delta}} + pd + m}
\]
and
\begin{equation}
    \label{error_prob_sharp}
    ||\sum_{l=0}^{T-1} X_l M_{l}||_2 \leq  4  \sigma_1(\Ic) L \sqrt{\log{\frac{1}{\delta}} + pd + m}.
\end{equation}
\end{proof}
The proof for noise and covariate cross terms is almost identical to Proposition~\ref{error_prob2} but easier because of independence. Finally note that $\sigma_1(\Ic) \leq \sqrt{\sum_{i=1}^d\|\beta_i\|^2_2}\sqrt{d} = \sqrt{\|\Tc_{0, d}^{ \top} v\|_2^2} \sqrt{d} \leq \beta \sqrt{d}$.
%%%%%%%%
\begin{prop}
	\label{error_prob3}
	For $0 < \delta < 1$, we have with probability at least $1 - \delta$
	\begin{align*}
	\bl \bl V^{-1/2}_T \Big(\sum_{k=0}^{ T}  \tilde{U}^{-}_{l+d, d} \tilde{\eta}^{+ \prime}_{l+1+d, d}\Tc \Oc^{\prime}_{0, d} \Big) \bl \bl_2 &\leq 4 \sigma_A \sqrt{\log{\frac{1}{\delta}} + pd + m} \\
	\bl \bl V^{-1/2}_T \Big(\sum_{k=0}^{ T}  \tilde{U}^{-}_{l+d, d} \tilde{\eta}^{- \prime}_{l, l}\Oc^{ \prime}_{d, d, l} \Big) \bl \bl_2 &\leq 4 \sigma_B \sqrt{\log{\frac{1}{\delta}} + pd + m} \\
	\bl \bl V^{-1/2}_T \Big(\sum_{k=0}^{ T}  \tilde{U}^{-}_{l+d, d} \tilde{\eta}^{- \prime}_{l+d, d} \Oc^{\prime}_{0, d, d} \Big) \bl \bl_2 &\leq 4 \sigma_C \sqrt{\log{\frac{1}{\delta}} + pd + m} \\
	\bl \bl V^{-1/2}_T \Big(\sum_{k=0}^{ T}  \tilde{U}^{-}_{l+d, d} \tilde{w}^{+ \prime}_{l+1+d, d}\Big) \bl \bl_2 &\leq 4 \sigma_D \sqrt{\log{\frac{1}{\delta}} + pd + m} 
	\end{align*}
	Here $\sigma = \max{(\sigma_A, \sigma_B, \sigma_C, \sigma_D)}$ where 
	$$\sigma_A  \vee \sigma_C \leq \sup_{||v||_2 = 1}\bl \bl \begin{bmatrix}
    v^{\top} CA^{d} & v^{\top} CA^{d-1} & v^{\top} CA^{d-2} & \hdots & 0 \\
    0 & \ddots & \ddots & \ddots & 0 \\
    0 & \ddots & \ddots & \ddots & \ddots \\
    0 & \hdots & v^{\top} CA^{d} & \hdots & v^{\top} C
    \end{bmatrix} \bl \bl_2 \leq \sum_{j=0}^d ||CA^j||_2 \leq \beta R \sqrt{d}$$ 
	$\sigma_B = \sqrt{\sigma(\sum_{k=1}^d \Tc \Oc_{d+k, T}^{\top}\Tc \Oc_{d+k, T})} \leq \beta R \sqrt{d}, \sigma_D \leq c$.
\end{prop}
%%%%%%%%%%%%%%%
By taking the intersection of all the aforementioned events for a fixed $\delta$ we then have with probability at least $1 - \delta$

	\begin{align*}
	\bl \bl \hat{\Hc}_{0, d, d} - \Hc_{0, d, d} \bl \bl_2 &\leq 16 \sigma \sqrt{\frac{1}{T}} \sqrt{m + pd + \log{\frac{d}{\delta}}}
	\end{align*}
% \begin{remark}
%     \label{const_sigma}
% Although not exactly precise, $||\hat{\Hc} - \Hc_{0, d, d}||_2$ can be interpreted as the Hankel norm of the difference between $d$--FIR approximation of the original LTI system and its estimate. Recall that the Hankel norm of $M$ is the largest singular value of $\Hc_{0, \infty, \infty}(M)$ and is typically close to the $\Hc_{\infty}$--norm (See Proposition~\ref{sys_norm}). In Propositions~\ref{error_prob1}-\ref{error_prob3} $\sigma$ has $\sqrt{d}$ dependence (due to the upper bound), when in fact $\sigma$ does not scale as $d$. This can be seen by $||CA^dB|| =O(\rho^d)$ where $\rho = \rho(A)$ and since $\sigma \leq \sum_{k=0}^d||CA^kB||_2$ we do not have a dependence on $d$. We remark that following Theorem 1.2-1.3 in~\citep{tu2017non} this analysis is also tight and falls under the class of $\ell_{\infty}$-constrained input systems. The error, $\epsilon$, in~\citep{tu2017non} scales as $\epsilon \leq \sqrt{\frac{d \log{d}}{T}}$ which is what we obtain here.
% \end{remark}
%%%%%%%%%

	%\input{tc_content/introduction.tex}
	%\input{tc_content/preliminaries.tex}
	% \section{Invertibility of Sample Covariance matrix $V_T$}
% \label{invertibility}
% In this section we will prove the invertibility of the energy term in a different way. Although to prove the simple invertibility is easier, showing that 

% $$c_1 T I \preceq V_T \preceq c_2 T I$$

% where $c_1, c_2$ are independent of $d$ is harder. Such sharp bounds are needed to show that minimax lower bounds are achieved in LTI identification. 

	%\input{content/appendix_minimax.tex}
	\section{Subspace Perturbation Results}
\label{subspace-perturb-results}
In this section we present variants of the famous Wedin's theorem (Section 3 of~\cite{wedin1972perturbation}) that depends on the distribution of Hankel singular values. These will be ``sign free'' generalizations of the gap--Free Wedin Theorem from~\cite{allen2016lazysvd}. First we define the Hermitian dilation of a matrix.
\begin{align*}
\Hc(S) = \begin{bmatrix}0 & S \\
S^{\prime} & 0 \end{bmatrix}
\end{align*}
The Hermitian dilation has the property that $||S_1 - S_2|| \leq \epsilon \Longleftrightarrow ||\Hc(S_1) - \Hc(S_2)|| \leq \epsilon$. Hermitian dilations will be useful in applying Wedin's theorem for general (not symmetric) matrices.
\begin{prop}
\label{sin_theta_thm}
Let $S, \hat{S}$ be symmetric matrices and $||S - \hat{S}|| \leq \epsilon$. Further, let $v_j, \hat{v}_j$ correspond to the $j^{th}$ eigenvector of $S, \hat{S}$ respectively such that $\lambda_1 \geq \lambda_2 \geq \ldots \geq \lambda_n$ and $\hat{\lambda}_1 \geq \hat{\lambda}_2 \geq \ldots \geq \hat{\lambda}_n$. Then we have 
\begin{equation}
    \label{sin_thm}
|\langle v_j, \hat{v}_k \rangle| \leq \frac{\epsilon}{{|\lambda_{j} - \hat{\lambda}_{k}|}}
\end{equation}
if either $\lambda_j$ or $\hat{\lambda}_{k}$ is not zero.
\end{prop}
\begin{proof}
Let $S = \lambda_j v_j v_j^{\prime} + V \Lambda_{-j} V^{\prime}$ and $\hat{S} = \hat{\lambda}_k \hat{v}_k \hat{v}_k^{\prime} + \hat{V} \hat{\Lambda}_{-k} \hat{V}^{\prime}$, wlog assume $|\lambda_j| \leq |\hat{\lambda}_k|$. Define $R = S- \hat{S}$
\begin{align*}
S &= \hat{S} + R \\
v_j^{\prime} S \hat{v}_k &= v_j^{\prime}\hat{S}\hat{v}_k + v_j^{\prime}R\hat{v}_k
\end{align*}
Since $v_j, \hat{v}_k$ are eigenvectors of $S$ and $\hat{S}$ respectively. 
\begin{align*}
\lambda_j v_j^{\prime}\hat{v}_k &=  \hat{\lambda}_k v_j^{\prime}\hat{v}_k + v_j^{\prime}R\hat{v}_k \\
|\lambda_j - \hat{\lambda}_k||v_j^{\prime}\hat{v}_k| &\leq \epsilon
\end{align*}
% Define $\lambda = \inf_{p}|\lambda_p|$ and let $\epsilon = \delta \lambda$ for some $0 \leq \delta \leq 1/2$.
\end{proof}
Proposition~\ref{sin_theta_thm} gives an eigenvector subjective Wedin's theorem. Next, we show how to extend these results to arbitrary subsets of eigenvectors. 
\begin{prop}
\label{gen_wedin}
	For $\epsilon > 0$, let $S, P$ be two symmetric matrices such that $||S-P||_2 \leq \epsilon$. Let 
	\[
	S= U \Sigma^S U^{\top}, P= V \Sigma^P V^{\top}
	\]
    Let $V_+$ correspond to the eigenvectors of singular values $\geq \beta$, $V_{-}$ correspond to the eigenvectors of singular values $\leq \alpha$ and $\bar{V}$ are the remaining ones. Define a similar partition for $S$.	Let $\alpha < \beta$
	\begin{align*}
		||U_{-}^{\top} V_{+}|| &\leq \frac{\epsilon}{\beta - \alpha} 
% 		||[U_{+}, \bar{U}]^{\top} V_{+}|| &\geq 1 - \frac{\epsilon}{\beta - \alpha}
	\end{align*}	
\end{prop}
\begin{proof}
The proof is similar to before. $S, P$ have a spectral decomposition of the form
\begin{align*}
    S &= U_+ \sigA_+ U_+^{\prime} + U_{-} \sigA_{-} U_{-}^{\prime} + \bar{U}{\sigA}_0 \bar{U}^{\prime} \\
    P &= V_+ \sigB_+ V_+^{\prime} + V_{-} \sigB_{-} V_{-}^{\prime} + \bar{V}{\sigB}_0 \bar{V}^{\prime}    
\end{align*}
Let $R=S-P$ and since $U_+$ is orthogonal to $U_{-}, \bar{U}$ and similarly for $V$
\begin{align*}
    U_{-}^{\prime}S &= \sigA_{-} U_{-}^{\prime} = U_{-}^{\prime}P + U_{-}^{\prime} R \\
    \sigA_{-} U_{-}^{\prime} V_{+} &= U_{-}^{\prime}V_{+}\sigB_{+} + U_{-}^{\prime} R V_{+}
    \end{align*}
    Diving both sides by $\sigB$
    \begin{align*}
    \sigA_{-} U_{-}^{\prime} V_{+} (\sigB_+)^{-1} &=  U_{-}^{\prime}V_{+} + U_{-}^{\prime} R V_{+} (\sigB_+)^{-1} \\
    ||\sigA_{-} U_{-}^{\prime} V_{+} (\sigB_+)^{-1}|| &\geq ||U_{-}^{\prime}V_{+}|| - ||U_{-}^{\prime} R V_{+} (\sigB_+)^{-1}|| \\
    \frac{\alpha}{\beta}||U_{-}^{\prime} V_{+}|| &\geq ||U_{-}^{\prime}V_{+}|| - \frac{\epsilon}{\beta} \\
    ||U_{-}^{\prime} V_{+}|| &\leq \frac{\epsilon}{\beta - \alpha}
\end{align*}
% Let $[U_{+}, \bar{U}] = U_0$. To prove the second part, observe that 
% \begin{align*}
%     ||U_0^{\prime} V_{+}|| &= ||U_0 U_{0}^{\prime} V_{+}|| = ||(I - U_- U_{-}^{\prime} V_{+})|| \\
%     &\geq 1 - ||U_{-}^{\prime} V_{+}||
% \end{align*}
\end{proof}
Let $S_k, P_k$ be the best rank $k$ approximations of $S, P$ respectively. We develop a sequence of results to see how $||S_k - P_k||$ varies when $||S-P|| \leq \epsilon$ as a function of $k$.
% Proposition implies that the orthonormal vectors of $A, B$ which are only $\epsilon$--apart in $\ell_2$--norm are reasonable close (up to sign changes). Let $\tilde{U}$ be the remaining set of singular vectors of $A$ (these must exist) then 
% \begin{align*}
% ||\tilde{U}\tilde{U}^{\prime} V|| &= ||(I - UU^{\prime})V|| \\
% &\geq ||V|| - ||\tilde{U}^{\prime} V|| \\
% &\geq 1 - \frac{\epsilon}{\tau \mu}
% \end{align*}
%%%%%%%%%%%%%

% We will also need a converse, \textit{i.e.}, then 
% \begin{prop}
% 	\label{sigdiff_mat}
% Let $A \rightarrow (U, \Sigma, V), \hat{A} \rightarrow (\hat{U}, \hat{\Sigma}, \hat{V})$. Further if $||U - \hat{U}|| \leq \epsilon_1, ||V - \hat{V}|| \leq \epsilon_1, ||\Sigma - \hat{\Sigma}|| \leq \epsilon_2$, then
% \[
% ||A - \hat{A}|| \leq 2\epsilon_1 \max{\{\sigma(A), \sigma(\hat{A})\}} +2 \epsilon_2
% \]
% \end{prop}
% \begin{proof}
% 	Follows by simple expansion through SVD.
% \end{proof}
\begin{prop}
\label{eigenspace_distance}
Let $S, P$ be such that 
\[
||S-P|| \leq \epsilon
\]
%Let the singular values of $S$ be arranged as follows:
%\begin{equation*}
%\sigma_1(S) > \ldots > \sigma_{r-1}(S) > \sigma_r(S) = \sigma_{r+1}(S) = \ldots = \sigma_s(S) >  \sigma_{s+1}(S) > \ldots \sigma_n(S) > \sigma_{n+1}(S) = 0 
%\end{equation*}
Furthermore, let $\epsilon$ be such that 
\begin{equation}
\epsilon \leq \inf_{\{1 \leq i \leq r-1\} \cup \{s+1  \leq i \leq n\}} \Big(\frac{\sigma_i(P) - \sigma_{i+1}(P)}{2}\Big) \label{interlacing_ppt}
\end{equation}
and $U_j^{S}, V^{S}_j$ be the left and right singular vectors of $S$ corresponding to $\sigma_j(S)$. 
% \[
% S^A_k = \{A_k | \hspace{1mm} ||A - A_k|| =  \inf_{\text{rank}(C) \leq k} ||A - C||\}
% \]
% and similarly define $S^B_k$. 
There exists a unitary transformation $Q$ such that
\begin{align*}
\sigma_{\max}([U_r^{P}, \ldots, U_s^{P}]Q - [U_r^{S}, \ldots, U_s^{S}]) &\leq  \frac{2\epsilon}{\min{\Big(\sigma_{r-1}(P) - \sigma_r(S), \sigma_{s}(S) - \sigma_{s+1}(P)\Big)}} \\
\sigma_{\max}( [V_r^{P}, \ldots, V_s^{P}]Q- [V_r^{S}, \ldots, V_s^{S}]) &\leq  \frac{2\epsilon}{\min{\Big(\sigma_{r-1}(P) - \sigma_r(S), \sigma_{s}(S) - \sigma_{s+1}(P)\Big)}}.
\end{align*}

\end{prop}
\begin{proof}
Let $r \leq k \leq s$. First divide the indices $[1, n]$ into 3 parts $K_1 = [1, r-1], K_2 = [r, s], K_3 = [s+1, n]$. Although we focus on only three groups extension to general case will be a straight forward extension of this proof. Define the Hermitian dilation of $S, P$ as $\Hc(S), \Hc(P)$ respectively. Then we know that the eigenvalues of $\Hc(S)$ are 
$$\cup_{i=1}^n\{\sigma_i(S), -\sigma_i(S)\}$$
Further the eigenvectors corresponding to these are $$\cup_{i=1}^n\Bigg\{\frac{1}{\sqrt{2}}\begin{bmatrix}u^{S}_i \\ v^{S}_i \end{bmatrix}, \frac{1}{\sqrt{2}}\begin{bmatrix}u^{S}_i \\ -v^{S}_i \end{bmatrix}\Bigg\}$$ 
Similarly define the respective quantities for $\Hc(P)$. Now clearly, $||\Hc(S) - \Hc(P)|| \leq \epsilon$ since $||S-P|| \leq \epsilon$. Then by Weyl's inequality we have that 
    \[
    |\sigma_i(S) - \sigma_i(P)| \leq \epsilon
    \]
%    where $\sigma_i(M) \geq \sigma_{i+1}(M)$ for any matrix $M$. 
Now we can use Proposition~\ref{sin_theta_thm}. 
% Observe that the spectrum of $\Hc(A)$ is 
%     $$\sigma_1(A) \geq \sigma_2(A) \geq \ldots \geq \sigma_n(A) \geq -\sigma_n(A) \geq \ldots \geq -\sigma_1(A)$$
To ease notation, define $\sigma_i(S) = \lambda_i(\Hc(S))$ and $\lambda_{-i}(\Hc(S)) = -\sigma_i(S)$ and let the corresponding eigenvectors be $a_i, a_{-i}$ for $S$ and $b_{i}, b_{-i}$ for $P$ respectively. Note that we can make the assumption that $\langle a_i, b_i \rangle \geq 0$ for every $i$. This does not change any of our results because $a_i, b_i$ are just stacking of left and right singular vectors and $u_i v_i^{\top}$ is identical for $u_i, v_i$ and $-u_i, -v_i$.

Then using Proposition~\ref{sin_theta_thm} we get for every $(i, j) \not \in K_2 \times K_2$ and $i \neq j$
    \begin{equation}
    \label{eq1}
    |\langle a_i, b_j \rangle| \leq \frac{\epsilon}{|\sigma_i(S) - \sigma_j(P)|}      
    \end{equation}
    similarly
    \begin{equation}
        \label{eq2}
    |\langle a_{-i}, b_{j} \rangle| \leq \frac{\epsilon}{|\sigma_i(S) + \sigma_j(P)|}
    \end{equation}
    Since 
    $$a_i = \frac{1}{\sqrt{2}}\begin{bmatrix}u^{S}_i \\ v^{S}_i \end{bmatrix}, a_{-i} = \frac{1}{\sqrt{2}}\begin{bmatrix}u^{S}_i \\ -v^{S}_i \end{bmatrix}, b_j = \frac{1}{\sqrt{2}}\begin{bmatrix}u^{P}_i \\ v^{P}_i \end{bmatrix}$$
    and $\sigma_i(S), \sigma_i(P) \geq 0$ we have by adding Eq.~\eqref{eq1},\eqref{eq2} that 
    \[
    \max{\Big(|\langle u^{S}_i, u^{P}_j\rangle|,|\langle v^{S}_i, v^{P}_j\rangle|\Big)} \leq \frac{\epsilon}{|\sigma_i(S) - \sigma_j(P)|}      
    \]
Define $U^S_{K_i}$ to be the matrix formed by the orthornormal vectors $\{a_{j}\}_{j \in K_i}$ and $U^S_{K_{-i}}$ to be the matrix formed by the orthonormal vectors $\{a_{j}\}_{j \in -K_{i}}$. Define similar quantities for $P$. Then 
\begin{align}
    &(U^{S}_{K_2})^{\top} U^P_{K_2} (U^P_{K_2})^{\top} U^S_{K_2} = (U^{S}_{K_2})^{\top} (I - \sum_{j \neq 2} U^P_{K_j} (U^P_{K_j})^{\top})U^S_{K_2} \nonumber\\
    &= (U^{S}_{K_2})^{\top} (I - \sum_{|j| \neq 2} U^P_{K_j} (U^P_{K_j})^{\top} - U^P_{K_{-2}} (U^P_{K_{-2}})^{\top})U^S_{K_2} \nonumber\\
    &= I - (U^{S}_{K_2})^{\top} \sum_{|j| \neq 2} U^P_{K_j} (U^P_{K_j})^{\top} U^S_{K_2} - (U^{S}_{K_2})^{\top}U^P_{K_{-2}} (U^P_{K_{-2}})^{\top}U^S_{K_2} \label{cross_eq}
\end{align}
Now $K_1, K_{-1}$ corresponds to eigenvectors where singular values $\geq \sigma_{r-1}(P)$, $K_3, K_{-3}$ corresponds to eigenvectors where singular values $\leq \sigma_{s+1}(P)$. We are in a position to use Proposition~\ref{gen_wedin}. Using that on Eq.~\eqref{cross_eq} we get the following relation
\begin{align}
 (U^{P}_{K_2})^{\top} U^S_{K_2} (U^S_{K_2})^{\top} U^P_{K_2} &\succeq I\Bigg(1 - \frac{\epsilon^2}{(\sigma_{r-1}(P) - \sigma_s(S))^2}  - \frac{\epsilon^2}{(\sigma_{s}(S) - \sigma_{s+1}(P))^2} \Bigg)    \nonumber \\
 &- (U^{S}_{K_2})^{\top}U^P_{K_{-2}} (U^P_{K_{-2}})^{\top}U^S_{K_2} \label{inter_step}
\end{align}
In the Eq.~\eqref{inter_step} we need to upper bound $(U^{S}_{K_2})^{\top}U^P_{K_{-2}} (U^P_{K_{-2}})^{\top}U^S_{K_2}$. To this end we will exploit the fact that all singular values corresponding to $U^{S}_{K_2}$ are the same. Since $||\Hc(S) - \Hc(P)|| \leq \epsilon$, then 
\begin{align*}
    \Hc(S) &= U^S_{K_2} \sigA_{K_2} (U^S_{K_2})^{\top} +  U^S_{K_{-2}} \sigA_{K_{-2}} (U^S_{K_{-2}})^{\top} +  U^S_{K_0} \sigA_{K_0} (U^S_{K_0})^{\top} \\
    \Hc(P) &= U^P_{K_2} \sigB_{K_2} (U^P_{K_2})^{\top} +  U^P_{K_{-2}} \sigB_{K_{-2}} (U^P_{K_{-2}})^{\top} +  U^P_{K_0} \sigB_{K_0} (U^P_{K_0})^{\top} 
\end{align*}
Then by pre--multiplying and post--multiplying we get 
\begin{align*}
    (U^S_{K_2})^{\top}\Hc(S) U^P_{K_{-2}} &=  \sigA_{K_2} (U^S_{K_2})^{\top}U^P_{K_{-2}} \\ (U^S_{K_2})^{\top}\Hc(P) U^P_{K_{-2}} &=  (U^S_{K_2})^{\top}U^P_{K_{-2}}\sigB_{K_{-2}}   
\end{align*}
Let $\Hc(S)-\Hc(P) = R$ then 
\begin{align*}
 (U^S_{K_2})^{\top}(\Hc(S)-\Hc(P)) U^P_{K_{-2}} &=  (U^S_{K_2})^{\top}R U^P_{K_{-2}} \\
 \sigA_{K_2} (U^S_{K_2})^{\top}U^P_{K_{-2}} - (U^S_{K_2})^{\top}U^P_{K_{-2}}\sigB_{K_{-2}}  &= (U^S_{K_2})^{\top}R U^P_{K_{-2}}
\end{align*}
Since $\sigA_{K_2} = \sigma_s(A) I$ then 
\begin{align*}
||(U^S_{K_2})^{\top}U^P_{K_{-2}}(\sigma_s(S) I - \sigB_{K_{-2}})||  &= ||(U^S_{K_2})^{\top}R U^P_{K_{-2}}|| \\
||(U^S_{K_2})^{\top}U^P_{K_{-2}}|| &\leq \frac{\epsilon}{\sigma_s(S) + \sigma_{s}(P)}
\end{align*}
Similarly  
$$||(U^P_{K_2})^{\top}U^S_{K_{-2}}|| \leq \frac{\epsilon}{\sigma_s(P) + \sigma_{s}(S)} $$
Since $\sigma_{s}(P) + \sigma_s(S) \geq \sigma_s(S) - \sigma_{s+1}(P)$ combining this with Eq.~\eqref{inter_step} we get 
\begin{equation}
    \label{k2_cross}
    \sigma_{\min}((U^S_{K_2})^{\top} U^P_{K_2}) \geq 1 - \frac{3\epsilon^2}{\min{\Big(\sigma_{r-1}(P) - \sigma_s(S), \sigma_{s}(S) - \sigma_{s+1}(P)\Big)^2}}
\end{equation}
Since 
$$\epsilon \leq \inf_i \Big(\frac{\sigma_i(P) - \sigma_{i+1}(P)}{2}\Big),$$
for Eq.~\eqref{k2_cross}, we use the inequality $\sqrt{1 - x^2} \geq 1 - x^2$ whenever $x < 1$ which is true when Eq.~\eqref{interlacing_ppt} is true. This means that there exists unitary transformation $Q$ such that 
\[
||U_{K_2}^S  - U_{K_2}^P Q|| \leq \frac{2\epsilon}{\min{\Big(\sigma_{r-1}(P) - \sigma_s(S), \sigma_{s}(S) - \sigma_{s+1}(P)\Big)}}
\]
\end{proof}
\begin{remark}
\label{hermitian_comment}
Note that $S, P$ will be Hermitian dilations of $\Hc_{0, \infty, \infty}, \hHc_{0, \hd, \hd}$ respectively in our case. Since the singular vectors of $S$ (and $P$) are simply stacked version of singular vectors of $\Hc_{0, \infty, \infty}$ (and $\hHc_{0, \hd, \hd}$), our results hold directly for the singular vectors of $\Hc_{0, \infty, \infty}$ (and $\hHc_{0, \hd, \hd}$)
\end{remark}

%\begin{remark}
%\label{Q_transform}
%The usefulness of Proposition~\ref{eigenspace_distance} comes from the fact that it works even when there is no gap between the singular values. This comes at the cost of the fact that we learn the singular vectors corresponding to same singular value only up to the unitary transformation $Q$. This is sufficient for model approximation since we are agnostic to unitary transformations, \textit{i.e.}, if the true model parameters are $M=(C, A, B)$ then we find $CQ, Q^{\top}AQ, Q^{\top}B$ which is sufficient for our identification procedure as it is clear from the discussion in Section~\ref{model_reduction}, specifically Eq.~\eqref{transform_Q}. Note that each singular vector corresponding to a unique singular value is learnt up to a factor of $\pm 1$, however as we discussed in the proof we can always assume that we recovered the correct sign for such singular vectors so that Proposition~\ref{eigenspace_distance} is satisfied. In the next result, we will implicitly assume that we compare against subspaces transformed by $Q$ as this does not, in principle, affect the reconstruction of $C, A, B$. \end{remark}
%Define $\Delta_+$ as follows, let $\sigma_{n+1} = 0$ then
%\begin{equation}
%\Delta_{+} = \inf_{\sigma_i \neq \sigma_{i+1}} \Big(1 - \frac{\sigma_{i+1}}{\sigma_i}\Big)\label{delta_def}
%\end{equation}
Let $r \leq k \leq s$. First divide the indices $[1, n]$ into 3 parts $K_1 = [1, r-1], K_2 = [r, s], K_3 = [s+1, n]$.

\begin{prop}[System Reduction]
	\label{reduction2}
Let $||S-P|| \leq \epsilon$ and the singular values of $S$ be arranged as follows:
\begin{equation*}
\sigma_1(S) > \ldots > \sigma_{r-1}(S) > \sigma_r(S) \geq \sigma_{r+1}(S) \geq \ldots \geq \sigma_s(S) >  \sigma_{s+1}(S) > \ldots \sigma_n(S) > \sigma_{n+1}(S) = 0 
\end{equation*}
Furthermore, let $\epsilon$ be such that 
\begin{equation}
\epsilon \leq \inf_{\{1 \leq i \leq r-1\} \cup \{s+1 \leq i \leq n\}} \Big(\frac{\sigma_i(P) - \sigma_{i+1}(P)}{2}\Big).
\end{equation}
 Define $K_0 = K_1 \cup K_2$, then
\begin{align*}
||U^S_{K_0} (\sigA_{K_0})^{1/2} - U^P_{K_0} (\sigB_{K_0})^{1/2}||_2 &\leq 2 \epsilon \sqrt{\sum_{i=1}^{r-1}\sigma_i/\zeta_{i}^2 +  \sigma_r/\zeta_r^2} + \sup_{1\leq i \leq s}|\sqrt{\sigma_i} - \sqrt{\hat{\sigma}_i}| 
\end{align*}
and $\sigma_i = \sigma_i(S), \hat{\sigma}_i = \sigma_i(P)$. Here $\zeta_i = \min{(\sigma_i - \sigma_{i+1}, \sigma_{i} - \sigma_{i+1})}$ and $\zeta_r = \min{(\sigma_{r-1} - \sigma_{r}, \sigma_{s} - \sigma_{s+1})}$.
\end{prop}
\begin{proof}

\vspace{2mm}
Since $U_{K_0}^S = [U_{K_1}^S U_{K_2}^S]$ and likewise for $B$, we can separate the analysis for $K_1, K_2$ as follows
\begin{align*}
    ||U^S_{K_0} (\sigA_{K_0})^{1/2} - U^P_{K_0} (\sigB_{K_0})^{1/2}|| &\leq ||(U^S_{K_0}  - U^P_{K_0}) (\sigA_{K_0})^{1/2}|| + ||U^P_{K_0}((\sigA_{K_0})^{1/2} - (\sigB_{K_0})^{1/2})|| \\
    &\leq \sqrt{||(U^S_{K_1} - U^P_{K_1}) (\sigA_{K_1})^{1/2}||_2^2 + ||(U^S_{K_2} - U^P_{K_2}) (\sigA_{K_2})^{1/2}||_2^2} \\
    &+ ||(\sigA_{K_0})^{1/2} - (\sigB_{K_0})^{1/2}||
\end{align*}
Now $||(\sigA_{K_0})^{1/2} - (\sigB_{K_0})^{1/2}|| = \sup_{l} |\sqrt{\sigma_l(S)} - \sqrt{\sigma_l(P)}|$. Recall that $\sigma_r(S) = \hdots = \sigma_k(S) = \hdots = \sigma_{s-1}(S)$ and by conditions on $\epsilon$ we are guaranteed that $\frac{\epsilon}{\sigma_i - \sigma_j} < 1/2$ for all $1 \leq i \neq j \leq r$. We will combine our previous results in Proposition~\ref{sin_theta_thm}--\ref{eigenspace_distance} to prove this claim. Specifically from Proposition~\ref{eigenspace_distance} we have 
\begin{align*}
    ||(U^S_{K_2} - U^P_{K_2}) (\sigA_{K_2})^{1/2}|| &\leq \frac{2\epsilon \sqrt{\sigma_r(S)}}{\min{\Big(\sigma_{r-1}(P) - \sigma_r(S), \sigma_{r}(S) - \sigma_{s+1}(P)\Big)}}
\end{align*}
On the remaining term we will use Proposition~\ref{eigenspace_distance} on each column 
\begin{align*}
    ||(U^S_{K_1} - U^P_{K_1}) (\sigA_{K_1})^{1/2}|| &\leq ||[\sqrt{\sigma_1(S)} c_1, \ldots, \sqrt{\sigma_{|K_1|}(S)} c_{|K_1|}]|| \leq \sqrt{\sum_{j=1}^{r-1} \sigma_j^2 ||c_j||^2} \\
    &\leq \epsilon \sqrt{\sum_{j=1}^{r-1}\frac{2 \sigma_j(S) }{\min{\Big(\sigma_{j-1}(P) - \sigma_j(S), \sigma_{j}(S) - \sigma_{j+1}(P)\Big)^2}}}
\end{align*}
\end{proof}
In the context of our system identification, $S = \Hc_{0, \infty, \infty}$ and $P = \hHc_{0, \hd, \hd}$. $P$ will be made compatible by padding it with zeros to make it doubly infinite. Then $U_{K_0}^S, U_{K_0}^P$ (after padding) has infinite rows. Define $Z_0 = U_{K_0}^S (\sigA_{K_0})^{1/2}(1:, :), Z_1 = U_{K_0}^S (\sigA_{K_0})^{1/2}(p+1:, :)$ (both infinite length) and similarly we will have $\hat{Z}_0, \hat{Z}_1$. Note that from a computational perspective we do not need to $Z_0, Z_1$; we only need to work with $\hat{Z}_0=U_{K_0}^P (\sigB_{K_0})^{1/2}(1:, :), \hat{Z}_1=U_{K_0}^P (\sigB_{K_0})^{1/2}(p+1:, :)$ and since most of it is just zero padding we can simply compute on $\hat{Z}_0(1:pd, :), \hat{Z}_1(1:pd, :)$. 
\begin{prop}
\label{A_err}
Assume $Z_1 = Z_0A$. Furthermore, $||S-P||_2 \leq \epsilon$ and let $\epsilon$ be such that 
\begin{equation}
\epsilon \leq \inf_{\{1 \leq i \leq r-1\} \cup \{s+1  \leq i \leq n\}} \Big(\frac{\sigma_i(P) - \sigma_{i+1}(P)}{2}\Big)
\end{equation}
then 
\begin{align*}
||(Z_0^{\prime}Z_0)^{-1}Z_0^{\prime}Z_1 -  (\hat{Z}_0^{\prime}\hat{Z}_0)^{-1}\hat{Z}_0^{\prime}\hat{Z}_1|| &\leq \frac{\Cc \epsilon(\gamma +1)}{\sigma_s}\Bigg(\sqrt{ \frac{\sigma^2_s }{((\sigma_s - \sigma_{s+1})\wedge (\sigma_{r-1} - \sigma_{s}) )^2}} \\
&+ \sqrt{\sum_{i=1}^{r-1}\frac{\sigma_i \sigma_s }{(\sigma_i - \sigma_{i+1})^2 \wedge (\sigma_{i-1} - \sigma_{i})^2}}\Bigg)    
\end{align*}
where $\sigma_1(A) \leq \gamma$.
\end{prop}
\begin{proof}
Note that $Z_1 = Z_0 A$, then 
\begin{align*}
&||(Z_0^{\prime}Z_0)^{-1}Z_0^{\prime}Z_1 -  (\hat{Z}_0^{\prime}\hat{Z}_0)^{-1}\hat{Z}_0^{\prime}\hat{Z}_1||_2 \\
= &||A  -  (\hat{Z}_0^{\prime}\hat{Z}_0)^{-1}\hat{Z}_0^{\prime}\hat{Z}_1||_2 = ||(\hat{Z}_0^{\prime}\hat{Z}_0)^{-1}\hat{Z}_0^{\prime}\hat{Z}_0 A - (\hat{Z}_0^{\prime}\hat{Z}_0)^{-1}\hat{Z}_0^{\prime}\hat{Z}_1||_2 \\
= &||(\hat{Z}_0^{\prime}\hat{Z}_0)^{-1}\hat{Z}_0^{\prime}\hat{Z}_0 A - (\hat{Z}_0^{\prime}\hat{Z}_0)^{-1}\hat{Z}_0^{\prime}Z_0 A + (\hat{Z}_0^{\prime}\hat{Z}_0)^{-1}\hat{Z}_0^{\prime}Z_0 A - (\hat{Z}_0^{\prime}\hat{Z}_0)^{-1}\hat{Z}_0^{\prime}\hat{Z}_1||_2 \\
\leq &||(\hat{Z}_0^{\prime}\hat{Z}_0)^{-1}\hat{Z}_0^{\prime}\hat{Z}_0 A - (\hat{Z}_0^{\prime}\hat{Z}_0)^{-1}\hat{Z}_0^{\prime}Z_0 A||_2 + ||(\hat{Z}_0^{\prime}\hat{Z}_0)^{-1}\hat{Z}_0^{\prime}Z_0 A - (\hat{Z}_0^{\prime}\hat{Z}_0)^{-1}\hat{Z}_0^{\prime}\hat{Z}_1||_2 \\
\leq &||(\hat{Z}_0^{\prime}\hat{Z}_0)^{-1}\hat{Z}_0^{\prime}||_2 \Big(||Z_0 A - \hat{Z}_0 A||_2 + || \underbrace{Z_0 A}_{\text{Shifted version of } Z_0} - \hat{Z}_1||_2 \Big) \\
\end{align*}
Now, $||(\hat{Z}_0^{\prime}\hat{Z}_0)^{-1}\hat{Z}_0^{\prime}||_2 \leq (\sqrt{\sigma_s - \epsilon})^{-1}$, $||Z_0 A - \hat{Z}_1||_2 \leq ||Z_0 - \hat{Z}_0||_2$ since $Z_1 = Z_0A$ is a submatrix of $Z_0$ and $\hat{Z}_1$ is a submatrix of $\hat{Z}_0$ we have $||Z_0A - \hat{Z}_1||_2 \leq ||Z_0 - \hat{Z}_0||_2$ and $||Z_0 A - \hat{Z}_0 A||_2 \leq ||A||_2 ||Z_0 - \hat{Z}_0||_2$
\begin{align*}
\leq &\frac{c \epsilon(\gamma +1)}{{\sigma_s}} \Bigg(\sqrt{\frac{\sigma^2_s }{((\sigma_s - \sigma_{s+1})\wedge (\sigma_{r-1} - \sigma_{s}) )^2}} + \sqrt{\sum_{i=1}^{r-1}\frac{\sigma_i \sigma_s }{(\sigma_i - \sigma_{i+1})^2 \wedge (\sigma_{i-1} - \sigma_{i})^2}}\Bigg)
\end{align*} 
\end{proof}
% \begin{corollary}
% \label{subspace_reduction_cor}
% If $\epsilon \leq \frac{\sigma_s \Delta_+}{2}$, then 
% \begin{align*}
% ||U^S_{K_0} (\sigA_{K_0})^{1/2} - U^P_{K_0} (\sigB_{K_0})^{1/2}||_2 &\leq \frac{\Cc \epsilon \sqrt{\Gamma_r}}{\Delta_+ \sqrt{\sigma_r}}  \\
% ||(Z_0^{\prime}Z_0)^{-1}Z_0^{\prime}Z_1 -  (\hat{Z}_0^{\prime}\hat{Z}_0)^{-1}\hat{Z}_0^{\prime}\hat{Z}_1|| &\leq \frac{\Cc (\gamma + 1)\epsilon \sqrt{\Gamma_r}}{ \sigma_r \Delta_+}
% \end{align*}
% where $\Gamma_r = \min{(r, \frac{1}{\Delta_+})}$. 
% \end{corollary}
% \begin{proof}
% From Proposition~\ref{reduction2} and \ref{A_err} we have the error bounds. Then note that 
% \[
% \sqrt{\frac{\sigma^2_s }{((\sigma_s - \sigma_{s+1})\wedge (\sigma_{r-1} - \sigma_{s}) )^2}} \leq \frac{1}{\Delta^2_+}
% \]
% and some simple arithmetic shows that 
% \[
% \sum_{i=1}^{r-1}\frac{\sigma_i \sigma_s }{(\sigma_i - \sigma_{i+1})^2 \wedge (\sigma_{i-1} - \sigma_{i})^2} \leq \sum_{i=1}^{r-1} \frac{\sigma_s}{\sigma_i \Delta_+^2} \leq \frac{\Gamma_{r-1}}{\Delta_+^2}
% \]
% where $\Gamma_{r-1} = \min{(r-1, \frac{1}{\Delta_+})}$. This follows because $\sum_{i=1}^{r-1}\frac{\sigma_s}{\sigma_i} = \sum_{i=1}^{r-1}(1 - \Delta_+)^{i}$.
% \end{proof}
	\section{Hankel Matrix Estimation Results}
\label{sec:hankel_est}
In this section we provide the proof for Theorem~\ref{hankel_est_thm}. For any matrix $P$, we define its doubly infinite extension $\bar{P}$ as
\begin{equation}
\bar{P} = \begin{bmatrix}
P & 0 & \ldots \\
0 & 0 & \ldots \\
\vdots & \vdots & \vdots
\end{bmatrix}    \label{pad}
\end{equation}
\begin{prop}
\label{truncation_error}
Fix $d > 0$. Then we have
\[
||\Hc_{d, \infty, \infty}||_2  \leq ||\Hc_{0, \infty, \infty} - \bar{\Hc}_{0, d, d}||_2 \leq \sqrt{2} ||\Hc_{d, \infty, \infty}||_2 \leq \sqrt{2} ||\Tc_{d, \infty}||_2
\]
\end{prop}
\begin{proof}
Define $\tilde{C}_d, \tilde{B}_d$ as follows
\begin{align*}
\tilde{C}_d &= \begin{bmatrix}
0_{md\times n} \\
C \\
CA \\
\vdots \\
\end{bmatrix} \\
\tilde{B}_d &= \begin{bmatrix}
0_{n\times pd} & B & AB & \hdots 
\end{bmatrix}
\end{align*}
Now pad $\Hc_{0, d, d}$ with zeros to make it a doubly infinite matrix and call it $\bar{\Hc}_{0, d, d}$ and we get that 
\begin{align*}
||\bar{\Hc}_{0, d, d} - \Hc_{0, \infty, \infty}|| &= \begin{bmatrix}
0 & M_{12} \\
M_{21} & M_{22}
\end{bmatrix}
\end{align*}
Note here that $M_{21}$ and $M_0 = \begin{bmatrix} M_{12} \\ M_{22}\end{bmatrix}$ are infinite matrices. Further $|| \Hc_{d, \infty, \infty}||_2 = ||M_0||_2 \geq ||M_{21}||_2$. Then
\begin{align*}
||\bar{\Hc}_{0, d, d} - \Hc_{0, \infty, \infty}||  &\leq \sqrt{||M_{12}||_2^2 + ||M_0||_2^2} \leq \sqrt{2} || \Hc_{d, \infty, \infty}||_2 
\end{align*}
Further $||\bar{\Hc}_{0, d, d} - \Hc_{0, \infty, \infty}|| \geq ||M_0|| = || \Hc_{d, \infty, \infty}||_2$.
 \end{proof}
 \begin{prop}
 \label{truncation_bounds}
 For any $d_1 \geq d_2$, we have 
 \[
||\Hc_{0, \infty, \infty} - \bar{\Hc}_{0, d_1, d_1}||_2 \leq  \sqrt{2}  ||\Hc_{0, \infty, \infty} - \bar{\Hc}_{0, d_2, d_2}||_2
 \]
 \end{prop}
 \begin{proof}
 Since $||\Hc_{d_1, \infty, \infty}||_2  \leq ||\Hc_{0, \infty, \infty} - \bar{\Hc}_{0, d_1, d_1}||_2 \leq  \sqrt{2}||\Hc_{d_1, \infty, \infty}||_2 $ from Proposition~\ref{truncation_error}. It is clear that $||\Hc_{d_1, \infty, \infty}||_2  \leq ||\Hc_{d_2, \infty, \infty}||_2 $. Then 
 \begin{align*}
\frac{1}{\sqrt{2}} ||\Hc_{0, \infty, \infty} - \bar{\Hc}_{0, d_1, d_1}||_2 \leq ||\Hc_{d_1, \infty, \infty}||_2  \leq ||\Hc_{d_2, \infty, \infty}||_2 \leq ||\Hc_{0, \infty, \infty} - \bar{\Hc}_{0, d_2, d_2}||_2
 \end{align*}
 \end{proof}
\begin{prop}
\label{toeplitz_decay}
Fix $d > 0$. Then 
\[
||\Tc_{d, \infty}(M)||_2 \leq   \frac{\tilde{M}\rho(A)^d}{1 - \rho(A)}
\]
\end{prop} 
\begin{proof}
Recall that 
	\begin{align*}
		\Tc_{d, \infty}(M) = \begin{bmatrix}
		0 & 0 & 0 & \hdots & 0 \\
		CA^{d}B & 0 & 0& \hdots & 0\\
		CA^{d+1}B & CA^{d}B & 0& \hdots & 0\\
		\vdots & \ddots & \ddots & \vdots & \vdots 
		\end{bmatrix}
	\end{align*}
Then $||\Tc_{d, \infty}(M)||_2 \leq \sum_{j=d}^{\infty} ||CA^{j}B||_2$. Now from Eq. 4.1 and Lemma 4.1 in~\cite{tu2017non} we get that $||CA^{j}B||_2 \leq \tilde{M} \rho(A)^j$. Then 
\[
\sum_{j=d}^{\infty} ||CA^{j}B||_2 \leq \frac{\tilde{M} \rho(A)^d}{1 - \rho(A)}
\]
\end{proof}
 \begin{remark}
 Proposition~\ref{toeplitz_decay} is just needed to show exponential decay and is not precise. Please refer to~\cite{tu2017non} for explicit rates.
 \end{remark}
%%%%%
Next we show that $T^{(\kappa)}_{*}(\delta)$ and $d_{*}(T, \delta)$ defined in Eq.~\eqref{ts_delta} given by
\begin{align}
    d_{*}(T, \delta) &= \inf\Bigg\{d \Bigg| 16 \beta R  \sqrt{d} \sqrt{\frac{m + pd + \log{\frac{T}{\delta}}}{T}} \geq ||\Hc_{0, d, d} - \Hc_{0, \infty, \infty}||_2 \Bigg\} \nonumber \\
    T^{(\kappa)}_{*}(\delta) &= \inf\Big\{T \Big |\frac{T}{c m^2 \log^3{(Tm/\delta)}}  \geq d_{*}(T, \delta), \hspace{2mm} d_{*}( T, \delta) \leq \frac{\kappa d_{*}(\frac{T}{\kappa^2}, \delta)}{8} \Big\} 
    \label{ds_delta_0}
\end{align}
The existence of $d_{*}(T, \delta)$ is predicated on the finiteness of $T^{(\kappa)}_*(\delta)$ which we discuss below.
\subsection{Existence of $T^{(\kappa)}_*(\delta) < \infty$}
\label{t1_t2}
Construct two sets 
\begin{align}
        T_{1}(\delta) &= \inf\Big\{T \Big | d_{*}(T, \delta) \in \Dc(T) \Big\} \label{T1} \\
    T_2(\delta) &= \inf\Big\{T \Big | d_{*}( t, \delta) \leq \frac{\kappa d_{*}(\frac{t}{\kappa^2}, \delta)}{8}  , \hspace{3mm} \forall t\geq T \Big\} \label{T2}
\end{align}
Clearly, $T^{(\kappa)}_*(\delta) < T_1(\delta) \vee T_2(\delta)$. A key assumption in the statement of our results is that $T^{(\kappa)}_*(\delta) < \infty$. We will show that it is indeed true. Let $\kappa \geq 16$.
\begin{prop}
\label{t1_exist}
For a fixed $\delta > 0$, $T_1(\delta) < \infty$ with $d_{*}(T, \delta) \leq \frac{c \log{(c T + \log{\frac{1}{\delta}})} - \log{R} +  \log{(\tilde{M}/\beta)}}{\log{\frac{1}{\rho}}}$. Here $\rho=\rho(A)$.
\end{prop}
\begin{proof}
Note the form for $d_*(T, \delta)$, it is the minimum $d$ that satisfies 
\[
16 \beta  R\sqrt{d} \sqrt{\frac{m + pd + \log{\frac{T}{\delta}}}{T}} \geq ||\Hc_{0, d, d} - \Hc_{0, \infty, \infty}||_2
\]
Since from Proposition~\ref{truncation_error} and \ref{toeplitz_decay} we have $||\Hc_{0, d, d} - \Hc_{0, \infty, \infty}||_2 \leq \frac{3 \tilde{M}\rho^d}{1 - \rho(A)}$, then $d_*(T, \delta) \leq d$ that satisfies
\[
16 \beta R \sqrt{d} \sqrt{\frac{m+ pd + \log{\frac{T}{\delta}}}{T}} \geq \frac{3 \tilde{M}\rho^d}{1 - \rho(A)}
\]
which immediately implies $d_{*}(T, \delta) \leq d =  \frac{c \log{(c T -\log{R} + \log{\frac{1}{\delta}})}  + \log{(\tilde{M}/\beta)} }{\log{\frac{1}{\rho}}}$, \textit{i.e.}, $\ds(T, \delta)$ is at most logarithmic in $T$. As a result, for a large enough $T$
\[
cm^2 d \log^2{(d)}\log^2{(m^2/\delta)} + c d \log^3{(2d)} \geq \frac{c \log{(c T + \log{\frac{1}{\delta}})} -\log{R} + \log{(\tilde{M}/\beta)}}{\log{\frac{1}{\rho}}}
\]
\end{proof}
The intuition behind $T_2(\delta)$ is the following: $d_{*}(T, \delta)$ grows at most logarithmically in $T$, as is clear from the previous proof. Then $T_2(\delta)$ is the point where $d_{*}(T, \delta)$ is still growing as $\sqrt{T}$ (\textit{i.e.}, ``mixing'' has not happened) but at a slightly reduced rate.
\begin{prop}
\label{t2_exist}
For a fixed $\delta > 0$, $T_2(\delta) < \infty$.
\end{prop}
\begin{proof}
Recall from the proof of Proposition~\ref{truncation_error} that $ ||\Hc_{d,\infty, \infty}|| \leq ||\Hc_{0, \infty, \infty} - \Hc_{0, d,d}|| \leq \sqrt{2}||\Hc_{d, \infty, \infty}||$. Now $\Hc_{d, \infty, \infty}$ can be written as
\begin{align*}
    \Hc_{d, \infty, \infty} &= \underbrace{\begin{bmatrix}
    C \\
    CA \\
    \vdots \\
    \end{bmatrix}}_{=\tilde{C}}A^{d} \underbrace{[B, AB, \hdots]}_{=\tilde{B}}
\end{align*}
Define $P_d = A^d \tilde{B} \tilde{B}^{\top} (A^d)^{\top}$. Let $\dk$ be such that for every $d \geq \dk$ and $\kappa \geq 16$
\begin{equation}
P_{d} \preceq \frac{1}{4\kappa} P_{0} \label{eq_cond}
\end{equation}
Clearly such a $\dk < \infty$ would exist because $P_0 \neq 0$ but $\lim_{d \rightarrow \infty} P_d = 0$. Then observe that $P_{2d} \preceq \frac{1}{4\kappa}P_{d}$. Then for every $d \geq \dk$ we have that 
\[
||\Hc_{d, \infty, \infty}|| \geq 4 \kappa ||\Hc_{2d, \infty, \infty}||
\]
Let 
\begin{equation}
    \label{t2_thresh}
T \geq \frac{4\dk \cdot (16)^2 \cdot \beta^2 R^2}{\sigma_0^2} (\dk p+ \log{(T/\delta)})
\end{equation}
where $\sigma_0 = ||\Hc_{\dk, \infty, \infty}||$. Assume that $\sigma_0 > 0$ (if not then are condition is trivially true). Then simple computation shows that 
\begin{align*}
 ||\Hc_{0, \dk, \dk} - \Hc_{0, \infty, \infty}||  &\geq  ||\Hc_{\dk, \infty, \infty}|| \geq \underbrace{16 \beta R \sqrt{\dk} \sqrt{\frac{m + p\dk + \log{\frac{T}{\delta}}}{T}}}_{< \frac{\sigma_0}{2}} 
\end{align*}
This implies that $\ds=\ds(T, \delta) \geq \dk$ for $T$ prescribed as above (ensured by Proposition~\ref{truncation_bounds}). But from our discussion above we also have 
\begin{align*}
    ||\Hc_{0, \ds, \ds} - \Hc_{0, \infty, \infty}|| \geq ||\Hc_{\ds, \infty, \infty}|| \geq 4 \kappa ||\Hc_{2\ds, \infty, \infty}|| \geq 2\kappa ||\Hc_{0, 2\ds, 2\ds} - \Hc_{0, \infty, \infty}||
\end{align*}
This means that if 
\begin{align*}
    ||\Hc_{0, \ds, \ds} - \Hc_{0, \infty, \infty}|| &\leq 16 \beta R \sqrt{\ds} \sqrt{\frac{m + p\ds + \log{\frac{T}{\delta}}}{T}}
\end{align*}
then
\begin{align*}
    ||\Hc_{0, 2\ds, 2\ds} - \Hc_{0, \infty, \infty}|| &\leq \frac{16}{2\kappa}  \beta R \sqrt{\ds} \sqrt{\frac{m + p\ds + \log{\frac{T}{\delta}}}{T}} \leq 16 \beta R \sqrt{2\ds} \sqrt{\frac{m + 2p\ds + \log{\frac{\kappa^2 T}{\delta}}}{\kappa^2 T}}
\end{align*}
which implies that $d_{*}(\kappa^2T, \delta) \leq 2d_{*}(T, \delta)$. The inequality follows from the definition of $d_*(\kappa^2 T, \delta)$. Furthermore, if $\kappa \geq 16$, $2d_{*}(T, \delta) \leq\frac{\kappa}{8}  d_{*}(T, \delta)$ whenever $T$ is greater than a certain finite threshold of Eq.~\eqref{t2_thresh}.
\end{proof}
Eq.~\eqref{eq_cond} happens when $\sigma(A^{d})^2 \leq \frac{1}{4\kappa} \implies \dk = \Oc \Big(\frac{\log{\kappa}}{\log{\frac{1}{\rho}}}\Big)$ where $\rho = \rho(A)$ and $T_2(\delta) \leq c T_1(\delta)$. It should be noted that the dependence of $T_i(\delta)$ on $\log{\frac{1}{\rho}}$ is worst case, \textit{i.e.}, there exists some ``bad'' LTI system that gives this dependence and it is quite likely $T_i(\delta)$ is much smaller. The condition $T \geq T_1(\delta) \vee T_2(\delta)$ simply requires that we capture some reasonable portion of the dynamics and not necessarily the entire dynamics. 
\subsection{Proof of Theorem~\ref{hankel_est_thm}}
\begin{prop}
\label{ds_error}
Let $T \geq T^{(\kappa)}_{*}(\delta)$ and $d_* = d_*(T, \delta)$ then 
\[
||\Hc_{0, \infty, \infty} - \hat{\Hc}_{0, d_*, d_*}|| \leq 2c \beta R  \sqrt{\frac{\ds}{T}} \sqrt{m + pd_* + \log{\frac{T}{\delta}}}
\]
\end{prop}
\begin{proof}
Consider the following error
\begin{align*}
    ||\Hc_{0, \infty, \infty} - \hat{\Hc}_{0, d_*, d_*}||_2 &\leq ||\Hc_{0, d_*, d_*} - \hat{\Hc}_{0, d_*, d_*}||_2  + ||\Hc_{0, \infty, \infty} - {\Hc}_{0, d_*, d_*}||_2 
\end{align*}
From Proposition~\ref{truncation_error} and Eq.~\eqref{ds_delta_0} we get that 
\[
||\Hc_{0, \infty, \infty} - {\Hc}_{0, d_*, d_*}||_2  \leq  16 \beta R \sqrt{\frac{\ds}{T}} \sqrt{m + pd_* + \log{\frac{T}{\delta}}}
\]
Since from Theorem~\ref{hankel_convergence}
\begin{align}
||\Hc_{0, \ds, \ds} - \hat{\Hc}_{0, d_*, d_*}||_2 &\leq  16 \beta R \sqrt{\frac{\ds}{T}} \sqrt{m + pd_* + \log{\frac{T}{\delta}}} \nonumber \\
||\Hc_{0, \infty, \infty} - \hat{\Hc}_{0, d_*, d_*}||_2 &\leq 32 \beta R \sqrt{\frac{\ds}{T}} \sqrt{m + pd_* + \log{\frac{T}{\delta}}} \label{err_final}
\end{align}
\end{proof}
Recall the adaptive rule to choose $d$ in Algorithm~\ref{alg:learn_ls}. From Theorem~\ref{hankel_convergence} we know that for every $d \in \Dc(T)$ we have with probability at least $1-\delta$. 
\[
||\Hc_{0, d, d} - \hat{\Hc}_{0, d, d}||_2 \leq 16 \beta R \sqrt{d}\p*{\sqrt{m + \frac{dp}{T} + \frac{\log{\frac{T}{\delta}}}{T}}}
\]
Let $\alpha(l) =  \sqrt{l}\Big(\sqrt{\frac{lp}{T} + \frac{ \log{\frac{T}{\delta}}}{T}} \Big)$. Then consider the following adaptive rule
\begin{align}
d_0(T, \delta) &=\inf \Big\{ l \Big| ||\hat{\Hc}_{0, l, l} - \hat{\Hc}_{0, h, h}||_2 \leq 16 \beta R (2\alpha(l) + \alpha(h))  \hspace{2mm}\forall  h \in \Dc(T), h \geq l \Big\} \\
\hd = \hd(T, \delta) &= d_0(T, \delta) \vee \log{\p*{\frac{T}{\delta}}} \label{d_choice_app}
\end{align}
for the same universal constant $c$ as Theorem~\ref{hankel_convergence}. Let $d_{*}(T, \delta)$ be as Eq.~\eqref{ds_delta_0}. Recall that $d_{*} = d_{*}(T, \delta)$ is the point where estimation error dominates the finite truncation error. Unfortunately, we do not have apriori knowledge of $d_{*}(T, \delta)$ to use in the algorithm. Therefore, we will simply use Eq.~\eqref{d_choice_app} as our proxy. The goal will be to bound $||\hHc_{0, \hd, \hd} - \Hc_{0, \infty, \infty}||_2$ 
\begin{prop}
\label{d_ds_rel}
Let $T \geq T^{(\kappa)}_{*}(\delta)$, $d_{*}(T, \delta)$ be as in Eq.~\eqref{ds_delta_0} and $\hd$ be as in Eq.~\eqref{d_choice_app}. Then with probability at least $1 -  \delta$ we have 
\[
\hd \leq d_{*}(T, \delta) \vee \log{\Big(\frac{T}{\delta}\Big)}
\]
\end{prop}
\begin{proof}
Let $d_{*} = d_{*}(T, \delta)$. First for all $h \in \Dc(T)  \geq d_{*}$, we note 
\begin{align}
    ||\hat{\Hc}_{0, \ds, \ds} - \hat{\Hc}_{0, h, h}||_2 &\leq ||\hat{\Hc}_{0, \ds, \ds} - \Hc_{0, \ds, \ds}|| + ||\Hc_{0, h, h} - \hat{\Hc}_{0, h, h}||_2 + ||\Hc_{0, h, h} - {\Hc}_{0,\ds, \ds}||_2 \nonumber\\
    &\underbrace{\leq}_{\infty > h \geq d_{*}} ||\hat{\Hc}_{0, \ds, \ds} - \Hc_{0, \ds, \ds}||_2 + ||\Hc_{0, h, h} - \hat{\Hc}_{0, h, h}||_2 + ||\Hc_{0, \infty, \infty} - \Hc_{0, \ds, \ds}||_2.\nonumber
\end{align}
We use the property that $||\Hc_{0, \infty, \infty} - \Hc_{0, \ds, \ds}||_2 \geq ||\Hc_{0, h, h} - \Hc_{0, \ds, \ds}||_2$. Furthermore, because of the properties of $\ds$ we have
\[
||\Hc_{0, \infty, \infty} - \Hc_{0, \ds, \ds}||_2 \leq 16 \beta R \alpha(\ds)
\]
and
\begin{align}
||\hat{\Hc}_{0, \ds, \ds} - \Hc_{0, \ds, \ds}||_2  &\leq 16 \beta R \alpha(\ds), \quad{} ||\Hc_{0, h, h} - \hat{\Hc}_{0, h, h}||_2 \leq 16 \beta R \alpha(h). \label{err_ratio}
\end{align}
and 
\[
||\hat{\Hc}_{0, \ds, \ds} - \hat{\Hc}_{0, h, h}||_2 \leq  16 \beta R (2\alpha(\ds) + \alpha(h) ).
\]
This implies that $d_0(T, \delta) \leq \ds$ and the assertion follows.
\end{proof}
We have the following key lemma about the behavior of $\hHc_{0, \hd, \hd}$.
\begin{lem}
    \label{hd_lemma}
    For a fixed $\kappa \geq 20$, whenever $T \geq T^{(\kappa)}_*(\delta)$ we have with probability at least $1-\delta$
    \begin{equation}
    ||\Hc_{0, \infty, \infty} - \hHc_{0, \hd, \hd}||_2  \leq 3 c \beta R \alpha\p*{\max{\p*{\ds(T, \delta), \log{\p*{\frac{T}{\delta}}}}}}
    \end{equation}
    Furthermore, $\hd = O(\log{\frac{T}{\delta}})$.
\end{lem}
\begin{proof}
Let $\ds > \hd$ then 
\begin{align*}
    ||\Hc_{0, \infty, \infty} - \hHc_{0, \hd, \hd}||_2 &\leq ||\Hc_{0, \infty, \infty} - \Hc_{0, \ds, \ds}||_2 + ||\hHc_{0, \hd, \hd} - \Hc_{0, \hd, \hd}||_2 + ||\hHc_{0, \ds, \ds} - \Hc_{0, \ds, \ds}||_2 \\
    &\leq 3c\beta R \alpha(\ds)
\end{align*}
If $\hd > \ds$ then
\begin{align*}
    ||\Hc_{0, \infty, \infty} - \hHc_{0, \hd, \hd}||_2 &\leq ||\Hc_{0, \infty, \infty} - \Hc_{0, \hd, \hd}||_2 + ||\hHc_{0, \hd, \hd} - \Hc_{0, \hd, \hd}||_2 = 2||\hHc_{0, \hd, \hd} - \Hc_{0, \hd, \hd}||_2  \\
    &\leq 2c \beta R \alpha(\hd) = 2 c \beta R \alpha\p*{\log{\Big(\frac{T}{\delta}\Big)}}
\end{align*}
where the equality follows from Proposition~\ref{d_ds_rel}. The fact that $\hd = O(\log{\frac{T}{\delta}})$ follows from Proposition~\ref{t1_t2}.
\end{proof}
In the following we will use $\Hc_l = \Hc_{0, l, l}$ for shorthand.
\begin{prop}
\label{hd_size}
Fix $\kappa  \geq 16$, and $T \geq T_{*}^{(\kappa)}(\delta)$. Then \[
   ||\hHc_{0, \hat{d}(T, \delta), \hat{d}(T, \delta)} - \Hc_{0, \infty, \infty}||_2 \leq 12 c \beta R\sqrt{\hat{d}( T, \delta)}\sqrt{\frac{m+p\hat{d}( T, \delta) + \log{\frac{ T}{\delta}}}{T}}
\]
with probability at least $1-\delta$.
\end{prop}
\begin{proof}
Assume that $\log{\Big(\frac{T}{\delta}\Big)} \leq d_{*}(T, \delta)$. Recall the following functions
\begin{align*}
    d_{*}(T, \delta) &= \inf{\Big\{d \Big |c \beta R \sqrt{d}\sqrt{\frac{m + pd + \log{\frac{T}{\delta}}}{T}} \geq ||\Hc_{d} - \Hc_{\infty}||_2 \Big\}} \\
    d_0(T, \delta) &= \inf{\Big\{l \Big | ||\hat{\Hc}_l - \hat{\Hc}_h||_2 \leq c \beta R (\alpha(h) + 2\alpha(l)) \hspace{2mm} \forall h \geq l, \hspace{2mm} h \in \Dc(T) \Big\}} \\
    \hd(T, \delta) &= d_0(T, \delta) \vee \log{\Big(\frac{T}{\delta}\Big)}
\end{align*}
It is clear that $d_{*}(\kappa^2 T, \delta) \leq (1 + \frac{1}{2p}) \kappa d_{*}(T, \delta)$ for any $\kappa \geq 16$. Assume the following 
\begin{itemize}
    \item $d_{*}( T, \delta) \leq \frac{\kappa}{8} d_{*}(\kappa^{-2} T, \delta)$ (This relation is true whenever $T \geq T^{(\kappa)}_*(\delta)$),
    \item $||\Hc_{\hat{d}(T, \delta)} - \Hc_{\infty}||_2 \geq 6c \beta R\sqrt{\hat{d}(T, \delta)}\sqrt{\frac{m+ p\hat{d}( T, \delta) + \log{\frac{T}{\delta}}}{ T}}$,
    \item $\hd( T, \delta) < \ds(\kappa^{-2} T, \delta) -1$.
\end{itemize}
The key will be to show that with high probability that all three assumptions can \textit{not} hold with high probability. For shorthand we define $\dso = \ds(T, \delta), \dsf = \ds(\kappa^{-2} T, \delta), \hdo = \hd(T, \delta), \hdf = \hd(\kappa^{-2} T, \delta)$ and $\Hc_l = \Hc_{0, l, l}, \hHc_l = \hHc_{0, l, l}$. Let $\tT = \kappa^{-2} T$. Then this implies that 
\begin{align*}
   \frac{ c \beta R (\sqrt{\dso} + 2\sqrt{\hdo})}{\kappa} \sqrt{\frac{m + p\dso + \log{\frac{\kappa^2 \tT}{\delta}}}{\tT}} &\geq ||\hHc_{\hdo} - \hHc_{\dso}||_2 \\
  ||\hHc_{\hdo} - \hHc_{\dso}||_2  &\geq  ||\hHc_{\hdo} - \Hc_{\infty}||_2 - ||\hHc_{\dso} - \Hc_{\infty}||_2 \\
   ||\hHc_{\dso} - \Hc_{\infty}||_2  + ||\hHc_{\hdo} - \hHc_{\dso}||_2 &\geq ||\hHc_{\hdo} - \Hc_{\infty}||_2  \\
    ||\hHc_{\dso} - \Hc_{\dso}||_2 + ||\Hc_{\dso} - \Hc_{\infty}||_2  + ||\hHc_{\hdo} - \hHc_{\dso}||_2 &\geq ||\hHc_{\hdo} - \Hc_{\infty}||_2 
\end{align*}
Since by definition of $\ds(\cdot, \cdot)$ we have 
$$||\hHc_{\dso} - \Hc_{\dso}||_2 + ||\Hc_{\dso} - \Hc_{\infty}||_2  \leq \frac{2 c \beta R}{\kappa} \sqrt{\dso} \sqrt{\frac{m+ p\dso + \log{\frac{\kappa^2 \tT}{\delta}}}{\tT}}$$
and by assumptions $\dso \leq \frac{\kappa}{8} \dsf, \hdo \leq \dsf$ then as a result $(\sqrt{\dso} + 2\sqrt{\hdo})\sqrt{\dso} \leq (\frac{2\kappa}{8} + 1) \dsf$
\begin{align*}
    &||\hHc_{\hdo} - \Hc_{\infty}||_2 \leq ||\hHc_{\dso} - \Hc_{\dso}||_2 + ||\Hc_{\dso} - \Hc_{\infty}||_2 +  \underbrace{||\hHc_{\hdo} - \hHc_{\dso}||_2}_{\Downarrow}  \\
    &\leq \underbrace{\frac{2 c \beta R \sqrt{\dso}}{\kappa} \sqrt{\frac{m + p\dso + \log{\frac{\kappa^2 \tT}{\delta}}}{\tT}}}_{\text{Prop.}~\ref{ds_error}} +\underbrace{\frac{ c \beta R (\sqrt{\dso} + 2\sqrt{\hdo})}{\kappa} \sqrt{\frac{m + p\dso + \log{\frac{\kappa^2 \tT}{\delta}}}{\tT}}}_{\text{Definition of }\hd^{(1)}} \\
    &||\hHc_{\hdo} - \Hc_{\infty}||_2 \leq \Big(\frac{1}{2} + \frac{1}{\kappa} \Big)c \beta R \sqrt{\dsf}\sqrt{\frac{m + p\dsf + \log{\frac{\tT}{\delta}}}{\tT}}
\end{align*}
where the last inequality follows from $(\sqrt{\dso} + 2\sqrt{\hdo})\sqrt{\dso} \leq (\frac{2\kappa}{8} + 1) \dsf$. Now by assumption 
\[
||\Hc_{\hdo} - \Hc_{\infty}||_2 \geq 6 c \beta R\sqrt{\hdo}\sqrt{\frac{m + p\hdo + \log{\frac{\kappa^2 \tT}{\delta}}}{\kappa^2\tT}}
\]
it is clear that 
\[
||\hHc_{\hdo} - \Hc_{\infty}||_2 \geq \frac{5}{6}||\Hc_{\hdo} - \Hc_{\infty}||_2
\]
and we can conclude that, since $\frac{6}{5}\Big(\frac{1}{2} + \frac{1}{\kappa} \Big) < \frac{1}{\sqrt{2}}$, 
$$||\Hc_{\hdo} - \Hc_{\infty}||_2 < c \beta R \sqrt{\frac{\dsf}{2}} \sqrt{\frac{m + p\dsf + \log{\frac{ \tT}{\delta}}}{\tT}}$$
which implies that $\hdo \geq \dsf-1$. This is because by definition of $\dsf$ we know that $\dsf$ is the minimum such that
\[
||\Hc_{\dsf} - \Hc_{\infty}||_2 \leq c \beta R \sqrt{\frac{\dsf}{2}} \sqrt{\frac{m + p\dsf + \log{\frac{ \tT}{\delta}}}{\tT}}
\] 
and furthermore from Proposition~\ref{truncation_bounds} we have for any $d_1 \leq d_2$
\[
||\Hc_{0, \infty, \infty} - \Hc_{0, d_1, d_1}|| \geq \frac{1}{\sqrt{2}}||\Hc_{0, \infty, \infty} - \Hc_{0, d_2, d_2}||. 
\]
This contradicts Assumption 3. So, this means that one of three assumptions do not hold. Clearly if assumption $3$ is invalid then we have a suitable lower bound on the chosen $\hd(\cdot, \cdot)$, \textit{i.e.}, since $\ds(\kappa^{-2} T, \delta) \leq d_{*}( T, \delta) \leq \frac{\kappa}{8} \ds(\kappa^{-2}T, \delta)$ we get 
\[
 \frac{\kappa}{8} \hd(\kappa^2 \tT, \delta) \geq \frac{\kappa}{8} \ds(\tT, \delta) - \frac{\kappa}{8} \geq \ds(\kappa^2 \tT, \delta) - \frac{\kappa}{8} \geq \hd(\kappa^2 \tT, \delta) - \frac{\kappa}{8} \geq \ds(\tT, \delta) - \frac{\kappa}{8}
\]
which implies from Lemma~\ref{hd_lemma} that (since we pick $\kappa = 16$, for large enough $T$ $\ds(\tT, \delta) \geq 4$) and we have
\begin{align*}
||\hHc_{\hat{d}(\kappa^2 \tT, \delta)} - \Hc_{\infty}||_2 &\leq 3 c \beta R\sqrt{\ds( \kappa^2 \tT, \delta)}\sqrt{\frac{p\ds(\kappa^2 \tT, \delta) + \log{\frac{\kappa^2 \tT}{\delta}}}{\kappa^2 \tT}}\\ 
&\leq \frac{3\kappa}{8} c  \beta R\sqrt{\hd( \kappa^2 \tT, \delta)}\sqrt{\frac{p\hd(\kappa^2 \tT, \delta) + \log{\frac{\kappa^2 \tT}{\delta}}}{\kappa^2 \tT}}
\end{align*}
Similarly, if assumption $2$ is invalid then we get that 
\[
 ||\Hc_{\hat{d}(\kappa^2 \tT, \delta)} - \Hc_{\infty}||_2 < 6c \beta R\sqrt{\hat{d}( \kappa^2 \tT, \delta)}\sqrt{\frac{p\hat{d}(\kappa^2 \tT, \delta) + \log{\frac{\kappa^2 \tT}{\delta}}}{\kappa^2 \tT}}
\]
and because $\hd(\kappa^2 \tT, \delta) \leq \ds(\kappa^2 \tT, \delta)$ and $||\hHc_{\hat{d}(\kappa^2 \tT, \delta)} - \Hc_{\infty}||_2 \leq ||\Hc_{\hat{d}(\kappa^2 \tT, \delta)} - \Hc_{\infty}||_2 + ||\hHc_{\hat{d}(\kappa^2 \tT, \delta)} - \Hc_{\infty}||_2$ we get in a similar fashion to Proposition~\ref{ds_error}
\[
||\hHc_{\hat{d}(\kappa^2 \tT, \delta)} - \Hc_{\infty}||_2 \leq 12c \beta R\sqrt{\hat{d}( \kappa^2 \tT, \delta)}\sqrt{\frac{p\hat{d}(\kappa^2 \tT, \delta) + \log{\frac{\kappa^2 \tT}{\delta}}}{\kappa^2 \tT}}
\]

Replacing $\kappa^2 \tT = T$ it is clear that for any $\kappa \geq 16$
\begin{equation}
    \label{hd_err}
    ||\hHc_{\hat{d}(T, \delta)} - \Hc_{\infty}||_2 \leq 12c \beta R\sqrt{\hat{d}(  T, \delta)}\sqrt{\frac{p\hat{d}( T, \delta) + \log{\frac{T}{\delta}}}{ T}}
\end{equation}
If $\ds(T, \delta) \leq \log{\p*{\frac{T}{\delta}}}$ then we can simply apply Lemma~\ref{hd_lemma} and our assertion holds.
\end{proof}

	\section{Model Selection Results}
\label{appendix_model_select}
\begin{prop}
\label{prop_delta_unknown}
Let $\Hc_{0, \infty, \infty} = U \Sigma V^{\top}, \hHc_{0, \hd, \hd} = \hat{U} \hat{\Sigma} \hat{V}^{\top}$ and 
$$||\Hc_{0, \infty, \infty} - \hHc_{0, \hd, \hd}|| \leq \epsilon .$$
Let $\hat{\Sigma}$ be arranged into blocks of singular values such that in each block $i$ we have 
\[
\sup_j \hsigma^{i}_{j} - \hsigma^{i}_{j+1} \leq \chi \epsilon
\]
for some $\chi \geq 2$, \textit{i.e.}, 
\[
\hat{\Sigma} = \begin{bmatrix}
\Lambda_1 & 0 & \ldots & 0 \\
0 & \Lambda_2 & \ldots & 0 \\
\vdots & \vdots & \ddots & 0 \\
0 & 0 & \ldots & \Lambda_l \\
\end{bmatrix}
\]
where $\Lambda_i$ are diagonal matrices and $\hsigma^{i}_{j}$ is the $j^{th}$ singular value in the block $\Lambda_i$. Then there exists an orthogonal transformation, $Q$, such that 
\begin{align*}
||\hat{U} \hat{\Sigma}^{1/2}Q - U\Sigma^{1/2}||_2 &\leq 2 \epsilon \sqrt{ \hat{\sigma}_{1}/\zeta_{n_1}^2 + \hat{\sigma}_{n_1+1}/\zeta_{n_2}^2 + \hdots + \hat{\sigma}_{\sum_{i=1}^{l-1}n_i+1}/\zeta_{n_l}^2}  \\
&+ 2\sup_{1\leq i \leq l}\sqrt{\hsigma^i_{\max}} - \sqrt{\hat{\sigma}^i_{\min}} + \frac{\epsilon}{\sqrt{\hsigma_{\hd}}} \wedge \sqrt{\epsilon}.
\end{align*}
Here $\sup_{1\leq i \leq l}\sqrt{\hsigma^i_{\max}} - \sqrt{\hat{\sigma}^i_{\min}} \leq \frac{\chi }{\sqrt{\hsigma^i_{\max}}} \epsilon \hd \wedge \sqrt{\chi \hd \epsilon}$ and 
\[
\zeta_{n_i} = \min{({\hsigma}^{n_{i-1}}_{\min}-{\hsigma}^{n_i}_{\max}, {\hsigma}^{n_i}_{\min}-{\hsigma}^{n_{i+1}}_{\max})}
\]
for $1 < i < l$, $\zeta_{n_1} ={\hsigma}^{n_1}_{\min}-{\hsigma}^{n_2}_{\max}$ and $\zeta_{n_l} = \min{({\hsigma}^{n_{l-1}}_{\min}-{\hsigma}^{n_l}_{\max}, {\hsigma}^{n_l}_{\min})}$.
\end{prop}
\begin{proof}
Let $\hat{U} \hat{\Sigma} \hat{V}^{\top} = \text{SVD}(\hHc_{0, \hd, \hd})$ and ${U} {\Sigma} {V}^{\top} = \text{SVD}(\Hcinf)$ where $||\hHc_{0, \hd, \hd}-\Hcinf||_2 \leq \epsilon$. $\hat{\Sigma}$ is arranged into blocks of singular values such that in each block $i$ we have $\hsigma^{i}_{j} - \hsigma^{i}_{j+1} \leq \chi \epsilon$, \textit{i.e.}, 
\[
\hat{\Sigma} = \begin{bmatrix}
\Lambda_1 & 0 & \ldots & 0 \\
0 & \Lambda_2 & \ldots & 0 \\
\vdots & \vdots & \ddots & 0 \\
0 & 0 & \ldots & \Lambda_l \\
\end{bmatrix}
\]
where $\Lambda_i$ are diagonal matrices and $\hsigma^{i}_{j}$ is the $j^{th}$ singular value in the block $\Lambda_i$. Furthermore, $\hsigma^{i-1}_{\min} - \hsigma^{i}_{\max} > \chi \epsilon$. From $\hat{\Sigma}$ define $\bar{\hat{\Sigma}}$ as follows:
\begin{equation}
\bar{\hat{\Sigma}} = \begin{bmatrix}
\bar{\hsigma}_1 I_{n_1 \times n_1} & 0 & \ldots & 0 \\
0 & \bar{\hsigma}_2 I_{n_2 \times n_2} & \ldots & 0 \\
\vdots & \vdots & \ddots & 0 \\
0 & 0 & \ldots & \bar{\hsigma}_l I_{n_l \times n_l} \\
\end{bmatrix} \label{eq:avg_sigma}
\end{equation}
where $\Lambda_i$ is a $n_i \times n_i$ matrix and $\bar{\hsigma}_i = \frac{1}{n_i}\sum_j \hsigma^{i}_j$. The key idea of the proof is the following: $(A, B, C) \equiv (Q A Q^{\top}, QB , CQ^{\top})$ where $Q$ is a orthogonal transformation and we will show that there exists a block diagonal unitary matrix $Q$ of the form
\begin{equation}
Q  = {\begin{bmatrix}
Q_{n_1 \times n_1} & 0 & \ldots & 0 \\
0 & Q_{n_2 \times n_2} & \ldots & 0 \\
\vdots & \vdots & \ddots & 0 \\
0 & 0 & \ldots & Q_{n_l \times n_l} \\
\end{bmatrix}} \label{eq:blk_diag}
\end{equation}
such that each block $Q_{n_i \times n_i}$ corresponds to a orthogonal matrix of dimensions $n_i \times n_i$ and that $||\hat{U}  \hat{\Sigma}^{1/2} Q - U \Sigma^{1/2}||_2$ is small if $||\hHc_{0, \hd, \hd}-\Hcinf||_2$ is small. Each of the blocks correspond to the set of singular values where the inter-singular value distance is ``small''. To start off, note that from Propositon~\ref{reduction2} there must exist a $Q$ that is block diagonal with orthogonal entries such that 
\begin{align}
||\hat{U}Q \hat{\Sigma}^{1/2} - U\Sigma^{1/2}||_2 &\leq c \epsilon \sqrt{ \hat{\sigma}_{1}/\zeta_{n_1}^2 + \hat{\sigma}_{n_1+1}/\zeta_{n_2}^2 + \hdots + \hat{\sigma}_{\sum_{i=1}^{l-1}n_i+1}/\zeta_{n_l}^2}  + \sup_{1\leq i \leq \hat{d}}|\sqrt{\sigma_i} - \sqrt{\hat{\sigma}_i}|  \label{eq:error_fun}
\end{align}
Here
\[
\zeta_{n_i} = \min{({\hsigma}^{n_{i-1}}_{\min}-{\hsigma}^{n_i}_{\max}, {\hsigma}^{n_i}_{\min}-{\hsigma}^{n_{i+1}}_{\max})}
\]
for $1 < i < l$, $\zeta_{n_1} ={\hsigma}^{n_1}_{\min}-{\hsigma}^{n_2}_{\max}$ and $\zeta_{n_l} = \min{({\hsigma}^{n_{l-1}}_{\min}-{\hsigma}^{n_l}_{\max}, {\hsigma}^{n_l}_{\min})}$. Informally, the $\zeta_i$ measure the singular value gaps between each blocks. 

% By the assumption on each block $i$ that $\hsigma^{i}_{j} - \hsigma^{i}_{j+1} \leq \chi \epsilon$, it is implied that 
% \[
% \hsigma^{i}_{\max}(1 - \chi)^{n_i} \leq \hsigma_{\min}^i
% \]
% where $n_i$ is the size of the block. Now, by design
% \[
% \zeta_{n_i}  \geq \min{(\hsigma_{\sum_{j=1}^{i-1}n_j} \chi, \hsigma_{\sum_{j=1}^{i}n_j} \chi)} = \hsigma_{\sum_{j=1}^{i}n_j} \chi,
% \]
% since $\hat{\sigma}_{\sum_{j=1}^{i-1}n_j}-\hat{\sigma}_{\sum_{j=1}^{i-1}n_j+1} \geq \hsigma_{\sum_{j=1}^{i-1}n_j+1} \chi$. Substituting in Eq.~\eqref{eq:error_fun} we get 
% \[
% ||\hat{U}Q \hat{\Sigma}^{1/2} - U\Sigma^{1/2}||_2 \leq \frac{C \epsilon}{\chi} \sqrt{1 /\hsigma_1 + \sum_{j=1}^i 1/\hsigma_{n_j}} + \sup_{1\leq i \leq \hat{d}}|\sqrt{\sigma_i} - \sqrt{\hat{\sigma}_i}|.
% \]
Furthermore, it can be shown that for any $Q$ of the form in Eq.~\eqref{eq:blk_diag}
\begin{align*}
|| \hat{U} Q\hat{\Sigma}^{1/2} -  \hat{U}\hat{\Sigma}^{1/2}Q||_2 &\leq || \hat{U} Q{\bhSigma}^{1/2} -  \hat{U}Q\hat{\Sigma}^{1/2}||_2 + || \hat{U}\hat{\Sigma}^{1/2}Q -  \hat{U}{\bhSigma}^{1/2}Q||_2 \leq 2 ||\hat{\Sigma}^{1/2} - \bhSigma^{1/2}||_2
\end{align*}
because $\hat{U} Q{\bhSigma}^{1/2} = \hat{U} {\bhSigma}^{1/2}Q$. Note that $ ||\hat{\Sigma}^{1/2} - \bhSigma^{1/2}||_2 \leq \sup_{1\leq i \leq l}\sqrt{\hsigma^i_{\max}} - \sqrt{\hat{\sigma}^i_{\min}}$. Now, when $\hsigma^i_{\max} \geq \chi n_i \epsilon$, then $\sqrt{\hsigma^i_{\max}} - \sqrt{\hat{\sigma}^i_{\min}} \leq \frac{\chi \epsilon}{\sqrt{\hsigma^i_{\max}}}$; on the other hand when $\hsigma^i_{\max} < \chi n_i \epsilon$ then $\sqrt{\hsigma^i_{\max}} - \sqrt{\bar{\hat{\sigma}}^i} \leq \sqrt{\chi n_i \epsilon}$ and this implies that 
\[
\sup_{1\leq i \leq l}\sqrt{\hsigma^i_{\max}} - \sqrt{\hat{\sigma}^i_{\min}} \leq \frac{\chi n_i }{\sqrt{\hsigma^i_{\max}}} \epsilon \wedge \sqrt{\chi n_i \epsilon}.
\]
Finally,
\begin{align*}
||\hat{U} \hat{\Sigma}^{1/2}Q - U\Sigma^{1/2}||_2 &\leq ||\hat{U}Q \hat{\Sigma}^{1/2} - U\Sigma^{1/2}||_2 + || \hat{U} Q\hat{\Sigma}^{1/2} -  \hat{U}\hat{\Sigma}^{1/2}Q||_2 \\
&= 2 \epsilon \sqrt{ \hat{\sigma}_{1}/\zeta_{n_1}^2 + \hat{\sigma}_{n_1+1}/\zeta_{n_2}^2 + \hdots + \hat{\sigma}_{\sum_{i=1}^{l-1}n_i+1}/\zeta_{n_l}^2}  + \sup_{1\leq i \leq \hat{d}}|\sqrt{\sigma_i} - \sqrt{\hat{\sigma}_i}|\\
&+ \frac{\chi \epsilon}{\sqrt{\hsigma^i_{\max}}} \wedge \sqrt{\chi \epsilon}.
\end{align*}
Our assertion follows since $\sup_{1\leq i \leq \hat{d}}|\sqrt{\sigma_i} - \sqrt{\hat{\sigma}_i}| \leq \frac{\epsilon}{\sqrt{\hsigma_{\hd}}} \wedge \sqrt{\epsilon}$.
\end{proof}

\begin{prop}
	\label{prop_c_select}
	Let $\Hc_{0, \infty, \infty} = U \Sigma V^{\top}, \hHc_{0, \hd, \hd} = \hat{U} \hat{\Sigma} \hat{V}^{\top}$ and 
	$$||\Hc_{0, \infty, \infty} - \hHc_{0, \hd, \hd}|| \leq \epsilon .$$
	Let $\hat{\Sigma}$ be arranged into blocks of singular values such that in each block $i$ we have 
	\[
	\sup_j \hsigma^{i}_{j} - \hsigma^{i}_{j+1} \leq \chi \epsilon
	\]
	for some $\chi \geq 2$, \textit{i.e.}, 
	\[
	\hat{\Sigma} = \begin{bmatrix}
	\Lambda_1 & 0 & \ldots & 0 \\
	0 & \Lambda_2 & \ldots & 0 \\
	\vdots & \vdots & \ddots & 0 \\
	0 & 0 & \ldots & \Lambda_l \\
	\end{bmatrix}
	\]
	where $\Lambda_i$ are diagonal matrices and $\hsigma^{i}_{j}$ is the $j^{th}$ singular value in the block $\Lambda_i$. Then there exists an orthogonal transformation, $Q$, such that 
	\begin{align*}
	\max{ \p*{||\hat{C} - C||_2, ||\hat{B} - B||_2}} &\leq 2 \epsilon \sqrt{ \hat{\sigma}_{1}/\zeta_{n_1}^2 + \hat{\sigma}_{n_1+1}/\zeta_{n_2}^2 + \hdots + \hat{\sigma}_{\sum_{i=1}^{l-1}n_i+1}/\zeta_{n_l}^2}  \\
	&+ 2\sup_{1\leq i \leq l}\sqrt{\hsigma^i_{\max}} - \sqrt{\hat{\sigma}^i_{\min}} + \frac{\epsilon}{\sqrt{\hsigma_{\hd}}} \wedge \sqrt{\epsilon} = \zeta,\\
	||A - \hat{A}||_2 &\leq  4 \gamma \cdot \zeta / \sqrt{\hat{\sigma}_{\hd}}.
	\end{align*}
	Here $\sup_{1\leq i \leq l}\sqrt{\hsigma^i_{\max}} - \sqrt{\hat{\sigma}^i_{\min}} \leq \frac{\chi }{\sqrt{\hsigma^i_{\max}}} \epsilon \hd \wedge \sqrt{\chi \hd \epsilon}$ and 
	\[
	\zeta_{n_i} = \min{({\hsigma}^{n_{i-1}}_{\min}-{\hsigma}^{n_i}_{\max}, {\hsigma}^{n_i}_{\min}-{\hsigma}^{n_{i+1}}_{\max})}
	\]
	for $1 < i < l$, $\zeta_{n_1} ={\hsigma}^{n_1}_{\min}-{\hsigma}^{n_2}_{\max}$ and $\zeta_{n_l} = \min{({\hsigma}^{n_{l-1}}_{\min}-{\hsigma}^{n_l}_{\max}, {\hsigma}^{n_l}_{\min})}$.
\end{prop}
\begin{proof}
	The proof follows because all parameters are equivalent up to a orthogonal transform (See discussion preceding Proposition~\ref{balanced_realization}). Following that we use Propositions~\ref{reduction2} and \ref{A_err}. 
\end{proof}

	\section{Order Estimation Lower Bound}
\label{lower_bnd}
\begin{lemma}[Theorem 4.21 in~\cite{boucheron2013concentration}]
\label{birge}
Let $\{\Pb_i\}_{i=0}^N$ be probability laws over $(\Sigma, \Ac)$ and let $\{A_i \in \Ac\}_{i=0}^N$ be disjoint events. If $a = \min_{i=0, \ldots, N} \Pb_i(A_i) \geq 1/(N+1)$,
\begin{align}
    \label{div_relation}
a \leq a \log{\Big(\frac{Na}{1-a}\Big)} + (1-a)\log{\Big(\frac{1-a}{1 - \frac{1-a}{N}}\Big)} \leq \frac{1}{N}\sum_{i=1}^N KL(P_i||P_0)
\end{align}
\end{lemma}
\begin{lemma}[Le Cam's Method]
\label{le_cam}
Let $P_0, P_1$ be two probability laws then 
\begin{align*}
\sup_{\theta \in \{0, 1\}} \Pb_{\theta}[M \neq \hat{M}] \geq \frac{1}{2} - \frac{1}{2}\sqrt{\frac{1}{2}KL(P_0 || P_1)}
\end{align*}
\end{lemma}
\begin{prop}
\label{kl_div}
Let $\Nc_0, \Nc_1$ be two multivariate Gaussians with mean $\mu_0 \in \Rb^{T}, \mu_1  \in \Rb^{T}$ and covariance matrix $\Sigma_0 \in \Rb^{T \times T}, \Sigma_1  \in \Rb^{T \times T}$ respectively. Then the $\text{KL}(\Nc_0, \Nc_1) = \frac{1}{2}\Big(\text{tr}(\Sigma_1^{-1}\Sigma_0) - T + \log{\frac{\text{det}(\Sigma_1)}{\text{det}(\Sigma_0)}} + \Ex_{\mu_1, \mu_0}[(\mu_1 - \mu_0)^{\top}\Sigma_1^{-1}(\mu_1 - \mu_0)]\Big)$.
\end{prop}
In this section we will prove a lower bound on the finite time error for model approximation. In systems theory subspace based methods are useful in estimating the true system parameters. Intuitively, it should be harder to correctly estimate the subspace that corresponds to lower Hankel singular values, or ``energy'' due to the presence of noise. However, due to strong structural constraints on Hankel matrix finding a minimax lower bound is a much harder proposition for LTI systems. Specifically, it is not clear if standard subspace identification lower bounds can provide reasonable estimates for a structured and non i.i.d. setting such as our case. To alleviate some of the technical difficulties that arise in obtaining the lower bounds, we will focus on a small set of LTI systems which are simply parametrized by a number $\zeta$. Consider the following canonical form order $1$ and $2$ LTI systems respectively with $m = p = 1$ and let $R$ be the noise-to-signal ratio bound.
\begin{align}
\label{canonical_form}
A_0 &= \begin{bmatrix}
0 & 1 & 0\\
0 & 0 & 0 \\
\zeta & 0 & 0
\end{bmatrix}, A_1 = A_0 , B_0 = \begin{bmatrix}
0 \\
0 \\
\sqrt{\beta}/R
\end{bmatrix},  B_1 = \begin{bmatrix}
0 \\
\sqrt{\beta}/R \\
\sqrt{\beta}/R
\end{bmatrix}, C_0 = \begin{bmatrix}
0 & 0 & \sqrt{\beta} R
\end{bmatrix}, C_1 =C_0
\end{align}
$A_0, A_1$ are Schur stable whenever $|\zeta| < 1$. 
\begin{align}
\Hc_{\zeta, 0} &= \beta \begin{bmatrix}
1 & 0 & 0 & 0 & 0  & \hdots\\
0 & 0 & 0 & 0 & 0 & \hdots \\
0 & 0 & 0 & 0 & 0 & \hdots\\
0  & 0 & 0 & 0 & 0 & \hdots \\
0 & 0 & 0 & 0 & 0 & \hdots\\
\vdots & \vdots & \vdots & \vdots & \vdots & \vdots
\end{bmatrix} \nonumber \\
\Hc_{\zeta, 1} &= \beta \begin{bmatrix}
1 & 0 & \zeta & 0 & 0 & \hdots \\
0 & \zeta & 0  & 0 & 0 & \hdots \\
\zeta & 0 & 0 & 0 & 0 & \hdots  \\
 0 & 0 & 0 & 0 & 0 & \hdots  \\
0 & 0 & 0 & 0 & 0 & \hdots \\
\vdots & \vdots & \vdots & \vdots & \vdots & \vdots
\end{bmatrix}\label{canonical_hankel}
\end{align}
Here $\Hc_{\zeta, 0}, \Hc_{\zeta, 1}$ are the Hankel matrices generated by $(C_0, A_0, B_0), (C_1, A_1, B_1)$ respectively. It is easy to check that for $\Hc_{\zeta, 1}$ we have $\frac{1}{\zeta}\leq \frac{\sigma_1}{\sigma_2} \leq \frac{1+\zeta}{\zeta}$ where $\sigma_i$ are Hankel singular values. Further the rank of $\Hc_{\zeta, 0}$ is $1$ and that of $\Hc_{\zeta, 1}$ is at least $2$. Also, $\frac{||\Tc \Oc_{0, \infty}((C_i, A_i, B_i))||_2}{||\Tc_{0, \infty}((C_i, A_i, B_i))||_2} \leq R$. 

This construction will be key to show that identification of a particular rank realization depends on the condition number of the Hankel matrix. An alternate representation of the input--output behavior is 
\begin{align}
    \label{collected-io}
    \begin{bmatrix}
    y_T \\
    y_{T-1} \\
    \vdots \\
    y_1
    \end{bmatrix} &= \underbrace{\begin{bmatrix}
    CB & CA_iB & \hdots & CA_i^{T-1}B \\
    0 & CB & \hdots & CA_i^{T-2}B \\
    \vdots & \vdots & \ddots& \vdots \\
    0 & 0 & \hdots & CB
    \end{bmatrix}}_{\Pi_i}    \underbrace{\begin{bmatrix}
    u_{T+1} \\
    u_{T} \\
    \vdots \\
    u_2
    \end{bmatrix}}_{U}\nonumber \\ 
    &+ \underbrace{\begin{bmatrix}
    C & CA_i & \hdots & CA_i^{T-1} \\
    0 & C & \hdots & CA_i^{T-2} \\
    \vdots & \vdots & \ddots& \vdots \\
    0 & 0 & \hdots & C
    \end{bmatrix}}_{O_i}    \begin{bmatrix}
    \eta_{T+1} \\
    \eta_{T} \\
    \vdots \\
    \eta_2
    \end{bmatrix} 
    + \begin{bmatrix}
    w_T \\
    w_{T-1} \\
    \vdots \\
    w_1
    \end{bmatrix}
\end{align}
where $A_i \in \{A_0, A_1\}$. We will prove this result for a general class of inputs, \textit{i.e.}, active inputs. Then we will follow the same steps as in proof of  Theorem 2 in~\cite{tu2018minimax}. 
\begin{align*}
    \text{KL}(P_0||P_1) &= \Ex_{P_0} \Bigg[\log{\prod_{t=1}^T \frac{\gamma_t(u_t | \{u_l, y_l\}_{l=1}^{t-1})P_0(y_t | \{u_l\}_{l=1}^{t-1})}{\gamma_t(u_t | \{u_l, y_l\}_{l=1}^{t-1})P_1(y_t | \{u_l\}_{l=1}^{t-1})}}\Bigg] \\
    &= \Ex_{P_0} \Bigg[\log{\prod_{t=1}^T \frac{P_0(y_t | \{u_l\}_{l=1}^{t-1})}{P_1(y_t | \{u_l\}_{l=1}^{t-1})}}\Bigg]
\end{align*}
Here $\gamma_t(\cdot|\cdot)$ is the active rule for choosing $u_t$ from past data. From Eq.~\eqref{collected-io} it is clear that conditional on $\{u_l\}_{l=1}^T$, $\{y_l\}_{l=1}^{T}$ is Gaussian with mean given by $\Pi_i U$. Then we use Birge's inequality (Lemma~\ref{birge}). In our case $\Sigma_0 = O_0 O_0^{\top} + I, \Sigma_1 = O_1 O_1^{\top} + I$ where $O_i$ is given in Eq.~\eqref{collected-io}. We will apply a combination of Lemma~\ref{birge}, Proposition~\ref{kl_div} and assume $\eta_i$ are i.i.d Gaussian to obtain our desired result. Note that $O_1 = O_0$ but $\Pi_1 \neq \Pi_0$. Therefore, from Proposition~\ref{kl_div} $KL(\Nc_0, \Nc_1) =  \Ex_{\mu_1, \mu_0}[(\mu_1 - \mu_0)^{\top}\Sigma_1^{-1}(\mu_1 - \mu_0)] \leq T\frac{\zeta^2}{R^2}$ where $\mu_i = \Pi_i U$. For any $\delta \in (0, 1/4)$, set $a=1-\delta$ in Proposition~\ref{birge}, then we get whenever 
\begin{align}
    \delta \log{\Big(\frac{\delta}{1-\delta}\Big)} + (1 - \delta) \log{\Big(\frac{1-\delta}{\delta}\Big)} \geq \frac{T\zeta^2}{R^2}
\end{align}
we have $\sup_{i \neq j} \Pb_{A_i}(A_j) \geq \delta$. For $\delta \in [1/4, 1)$ we use Le Cam's method in Lemma~\ref{le_cam} and show that if $8\delta^2 \geq \frac{T \zeta^2}{R^2}$ then $\sup_{i \neq j} \Pb_{A_i}(A_j) \geq \delta$. Since $\delta^2 \geq c \log{\frac{1}{\delta}}$ when $\delta \in [1/4, 1)$ for an absolute constant, our assertion holds.

\end{document}